%% file: main.tex
\documentclass[11pt,letterpaper]{article}
\usepackage{macros}

\newcommand{\agree}{\operatorname{agree}}

\newcommand{\CP}{\mathsf{CP}}

\newcommand{\ISConst}{\calA^{\mathsf{IS}}}
\newcommand{\ColorConst}{\calA^{\mathsf{Col}}}
\newcommand{\VEConst}{\calA^{\mathsf{outer}}}
\newcommand{\NBConst}{\calA^{\mathsf{NB}}}
\newcommand{\BoolConst}{\calA^{\mathsf{bool}}}

\newcommand{\Bernstein}{\mathsf{Be}}
\newcommand{\sens}{\mathsf{sens}}
\newcommand{\w}{\mathsf{wt}}

\renewcommand{\epsilon}{\eps}

\allowdisplaybreaks

\begin{document}
\title{Rounding Large Independent Sets on Expanders}

\author{Mitali Bafna\thanks{mitalib@mit.edu}\\MIT 
\and Jun-Ting Hsieh\thanks{juntingh@cs.cmu.edu. Supported by NSF CAREER Award \#2047933.}\\CMU 
\and Pravesh K. Kothari\\IAS \& Princeton\thanks{kothari@cs.princeton.edu. Supported by NSF CAREER Award \#2047933, NSF \#2211971, an Alfred P. Sloan Fellowship, and a Google Research Scholar Award.}}

\maketitle

\begin{abstract}
We develop a new approach for approximating large independent sets when the input graph is a one-sided spectral expander --- that is, the uniform random walk matrix of the graph has its second eigenvalue bounded away from 1. Consequently, we obtain a polynomial time algorithm to find linear-sized independent sets in one-sided expanders that are almost $3$-colorable or are promised to contain an independent set of size $(1/2-\epsilon)n$. Our second result above can be refined to require only a weaker \emph{vertex expansion} property with an efficient certificate. In a surprising contrast to our algorithmic result, we observe that the analogous task of finding a linear-sized independent set in almost $4$-colorable one-sided expanders (even when the second eigenvalue is $o_n(1)$) is NP-hard, assuming the Unique Games Conjecture. 

All prior algorithms that beat the worst-case guarantees for this problem rely on bottom eigenspace enumeration techniques (following the classical spectral methods of Alon and Kahale \cite{AlonK97}) and require two-sided expansion, meaning a bounded number of negative eigenvalues of magnitude $\Omega(1)$. Such techniques naturally extend to almost $k$-colorable graphs for any constant $k$, in contrast to analogous guarantees on one-sided expanders, which are Unique Games-hard to achieve for $k \geq 4$. 

Our rounding scheme builds on the method of simulating multiple samples from a pseudo-distribution introduced in \cite{BBKSS} for rounding Unique Games instances. The key to our analysis is a new \emph{clustering} property of large independent sets in expanding graphs --- every large independent set has a larger-than-expected intersection with some member of a small list --- and its formalization in the low-degree sum-of-squares proof system.

\end{abstract}

\thispagestyle{empty}
\setcounter{page}{0}

\newpage
\setcounter{tocdepth}{2}
\tableofcontents
\thispagestyle{empty}
\setcounter{page}{0}

\newpage

\input{intro}
\input{overview}

\input{prelims}
\input{is-on-expanders}
\input{3-colorable}
\input{gap-IS}
\input{hypercube}

\input{noisy-hypercube}

\section*{Acknowledgements}
We thank Dor Minzer for pointing us to~\cite{EKLM22}. We thank Avi Wigderson for illuminating discussions and suggesting the analogy between the clustering property shown in this work and the solution space geometry of random optimization problems. We thank Mark Braverman for helpful discussions on SDP rounding, Sidhanth Mohanty for suggesting terminology for our clustering property, and Guru Guruganesh for discussions related to the folklore argument in \Cref{sec:kms}.

\bibliographystyle{alpha}
\bibliography{main}

\appendix
\input{hardness}
\input{KMS}

\input{bm-lemmas}

\input{sqrt-approx}

\end{document}

%% file: intro.tex

\section{Introduction}
Finding large independent sets is a notoriously hard problem in the worst case. The best known algorithms can only find independent sets of size $\wt{O}(\log^2 n)$ in $n$-vertex graphs with independent sets of near-linear size~\cite{Feige04}. In this paper, we are interested in the important setting when the input graph contains an independent set of size $cn$ for a large constant $c<1/2$.\footnote{The classical $2$-approximation algorithm for vertex cover (complement of an independent set) implies an algorithm to find an independent set of size $2\epsilon n$ when the input graph has one of size $(1/2+\epsilon)n$.} In this setting, the problem remains challenging and has served as a benchmark for developing new techniques in approximation algorithm design over the years.
When $c=1/2-\epsilon$ for tiny enough $\epsilon>0$, a generalization of the SDP rounding of Karger, Motwani, and Sudan~\cite{KMS98} finds an independent set of size $n^{1-O(\epsilon)}$ (see \Cref{sec:kms}). When $c\ll 1/2$, all known efficient algorithms~\cite{BoppanaH92,AlonK98} can only find independent sets of size $n^{\delta(c)}$ for some $\delta(c) < 1$, and this is true even when the graph is $k$-colorable (thus $c \geq 1/k$). A decades-long influential research on coloring $k$-colorable graph has progressively improved the constant appearing in the exponent, with the most recent improvement being in 2024~\cite{Wigderson83,Blum94,BlumK97,KMS98,ACC06,Chlamtac09,KawarabayashiT17,KTY24}.
To summarize, in the worst-case, even when $c$ approaches $1/2$, our best known efficient algorithms can only find independent sets of size that is a polynomial factor smaller than $n$.

There is evidence that the difficulties in improving the above algorithms might be inherent. Assuming the Unique Games Conjecture (UGC), for any constant $\eps > 0$, it is NP-hard to find an independent set of size $\epsilon n$ even when the input graph contains an independent set of size $(1/2-\epsilon)n$ \cite{KhotR08,BansalK09}. Similar hardness results suggest that it may be hard to color $3$-colorable graphs with any constant number of colors~\cite{DinurS05,DinurMR06,DinurKPS10,KhotS12,GS20}. 

Given the above worst-case picture, a substantial effort over the past three decades has explored algorithms that work under natural structural assumptions on the input graphs. One line of work studies \emph{planted average-case} models for independent set~\cite{Karp72,Jerrum92,Kucera95} and coloring~\cite{BlumS95,AlonK97}, as well as their semirandom generalizations~\cite{BlumS95,FeigeK01,CharikarSV17,MckenzieMT20,BuhaiKS23,BBKS24}. A related body of research has focused on graphs that satisfy natural, deterministic assumptions, such as expansion, with the goal of isolating simple and concrete properties of random instances that enable efficient algorithms. This approach has been explored for Unique Games~\cite{Trevisan08,AroraKKSTV08,MakarychevM11,ABS15,BBKSS} and UG-hard problems 
like Max-Cut and Sparsest Cut~\cite{DeshpandeHV16,RabaniV17}, and has been instrumental in making progress even for worst-case instances; for example, the works on unique games on expanders eventually led to a subexponential algorithm for arbitrary UG instances~\cite{ABS10}. Over the past decade, such assumptions have also been investigated for independent set and coloring~\cite{AroraG11,DavidF16,KumarLT18}. In particular, a recent work of David and Feige~\cite{DavidF16} gave polynomial-time algorithms for finding large independent sets in \emph{planted} $k$-colorable expander graphs. We discuss these works in more detail next.

\paragraph{Prior Works and One-Sided vs Two-Sided Expansion.}
There is a crucial difference between the expansion assumptions in prior works on coloring vs other problems. A $d$-regular graph whose normalized adjacency matrix $\frac{1}{d}A$ (a.k.a., the uniform random walk matrix) has eigenvalues $1=\lambda_1 \geq \lambda_2\geq \cdots \geq \lambda_n$ is called a \emph{one-sided} spectral expander if for some $\lambda < 1$ ($\lambda$ is called the \emph{spectral gap}), $\lambda_2 \leq \lambda$, and a \emph{two-sided} spectral expander, $\max\{\lambda_2, |\lambda_n|\} \leq \lambda$ 
Most algorithms for problems (e.g., Unique Games and other constraint satisfaction problems) on expanders only need one-sided spectral expansion, as they primarily rely on the \emph{conductance}, or the fraction of edges leaving any subset of vertices in the graph, a combinatorial property closely related to $\lambda_2$ via Cheeger's inequality. In contrast, previous algorithms for finding independent sets in expanders with a planted $k$-coloring rely on two-sided spectral expansion (i.e., control of even the negative end of the spectrum).

This is \emph{not} just a technical quirk;
the main observation underlying such algorithms (due to Alon and Kahale~\cite{AlonK97}, following Hoffman~\cite{Hoffman70}) is that a random graph is a two-sided spectral expander (thus, has no large negative eigenvalues) and that planting a $k$-coloring in it introduces \emph{negative} eigenvalues of large magnitude, whose corresponding eigenvectors are correlated with indicator vectors of the color classes. This allows using the bottom eigenvectors of the graph to obtain a coarse \emph{spectral clustering}. All the works above, including those on deterministic expander graphs~\cite{DavidF16}, build on this basic observation for their algorithmic guarantees.

This basic idea becomes inapplicable if we are working with \emph{one-sided} spectral expanders that behave markedly differently  in the context of graph coloring. To illustrate this point, we observe the following proposition with a simple proof (see~\Cref{sec:hardness-k-col}) which implies that there is likely no efficient algorithm to find an $\Omega(n)$-sized independent set in an $\eps$-almost $4$-colorable graph
(i.e., $4$-colorable if one removes $\eps$ fraction of vertices), even when promised to have nearly perfect one-sided spectral expansion with $\lambda_2 \leq o_n(1)$!

\begin{proposition}[See \Cref{prop:hardness-formal}] \label{prop:hardness}
    Assuming the Unique Games Conjecture, for any constants $\eps,\gamma > 0$, it is NP-hard to find an independent set of size $\gamma n$ in an $n$-vertex regular graph that is $\eps$-almost $4$-colorable and has $\lambda_2 \leq o_n(1)$.
\end{proposition}
This is in sharp contrast to David and Feige's algorithm \cite{DavidF16} which shows how to find a planted $k$-coloring in a sufficiently strong two-sided spectral expander for any constant $k$.\footnote{\cite{DavidF16} focused on finding a partial or full coloring, which requires the planted coloring to be roughly balanced. Their spectral clustering technique can likely find a large independent set even when the coloring is not balanced.}

We prove \Cref{prop:hardness} by a reduction from the UG-hardness of finding linear-sized independent sets in $\eps$-almost $2$-colorable graphs~\cite{BansalK09} and guaranteeing one-sided expansion in addition at the cost of obtaining an almost $4$-colorable graph. A similar reduction allows us to show hardness of finding linear-sized independent sets in exactly $6$-colorable ($\eps=0$) one-sided spectral expanders (see \Cref{prop:hardness-6-colorable}). We remark that the instances produced by the reduction must necessarily have many negative eigenvalues, otherwise spectral clustering algorithms based on the bottom eigenspace~\cite{AlonK97,DavidF16} will succeed in finding linear-sized independent sets.

\paragraph{This Work.} We are thus led to the main question studied in this work:

{\center \emph{Can polynomial-time algorithms find a large independent set in a $3$-colorable one-sided spectral expander?}}
\vspace{3mm}

\Cref{prop:hardness} injects a fair amount of intrigue into this question, but our motivations for studying it go further. In light of the above discussion, an affirmative answer would necessarily require developing a new algorithmic approach that departs from previous spectral clustering methods based on bottom eigenvectors since such techniques do not distinguish between $3$ vs $4$-colorable graphs.

Let us spoil the intrigue: in this work, we develop new algorithms for finding large independent sets via rounding sum-of-squares (SoS) relaxations. Our polynomial-time algorithms succeed in finding linear-sized independent sets in almost $3$-colorable graphs that satisfy one-sided spectral expansion. Given the UG-hardness (i.e., \Cref{prop:hardness}) of finding linear-sized independent sets in an almost $4$-colorable one-sided expander, we obtain a stark and surprising difference between almost $3$-colorable and almost $4$-colorable one-sided expander graphs.

\begin{mtheorem} \label{thm:3-colorable-main}
    There is a polynomial-time algorithm that, given an $n$-vertex regular $10^{-4}$-almost $3$-colorable one-sided spectral expander with $\lambda_2 \leq 10^{-4}$, finds an independent set of size $\geq 10^{-4}n$.
\end{mtheorem}

A natural attempt to find a coloring of the graph is to iterate on the graph obtained by removing the vertices in an independent set.
However, this approach does not work in our worst-case expander setting since the remaining graph may not be an expander.
In fact, this difficulty appears to be inherent even when the graph has two-sided expansion. Towards a formal barrier, David and Feige~\cite{DavidF16} proved that it is NP-hard to find a $3$-coloring planted in a \emph{random} host graph with a not-too-large degree, even though they give an algorithm for finding an $\Omega(n)$-size independent set.

Our techniques succeed without the $3$-colorability assumption if the input graph has an independent set of size $(1/2-\epsilon)n$, and satisfies a weaker quantitative one-sided spectral expansion.

\begin{mtheorem} \label{thm:IS-expander-main}
    For every positive $\epsilon \leq 0.001$, there is a polynomial-time algorithm that, given an $n$-vertex regular graph that contains an independent set of size $(\frac{1}{2}-\eps)n$ and is a one-sided spectral expander with $\lambda_2 \leq 1- 40 \eps$, outputs an independent set of size at least $10^{-3}n$.
\end{mtheorem}
Note that we get an algorithm for $\eps$-almost $2$-colorable one-sided expanders as an immediate corollary. Before this work, no algorithm that beat the worst-case guarantee of outputting a $n^{1-O(\eps)}$-sized independent set was known in this setting.

With more work, our algorithms succeed even under the weaker notion of \emph{vertex} (as opposed to spectral) expansion defined below.

\begin{definition}[Small-set vertex expansion] \label{def:ssve}
The small-set vertex expansion (SSVE) of a graph $G = (V, E)$, with size parameter $\delta \in (0, 1/2]$, is defined as
\begin{equation*}
\Psi_{\delta}(G) \coloneqq \min_{S\subseteq V: 0 < |S| \leq \delta |V|} \frac{|N_G(S)|}{|S|} \mcom
\end{equation*}
where $N_G(S)$ is the set of neighbors of $S$: $N_G(S) = \{u\notin S: \exists v\in S,\ \{u,v\}\in E\}$.
We say that $G$ is a $\delta$-SSVE if $\Psi_\delta(G) \geq 2$.\footnote{Any constant bigger than $1$ suffices for our arguments.}
\end{definition}

We note that vertex expansion is a weaker notion than edge expansion since a graph can be a vertex expander without being an edge expander. Our algorithm succeeds on graphs that admit efficient certificates of vertex expansion (see~\Cref{def:certified-VE}; we only need certificates for expansion of small sets). Such certificates may exist in graphs without spectral expansion. As an example, we show in \Cref{sec:noisy-hypercube} that the well-studied noisy hypercube graph, despite not being a spectral expander, admits a non-trivially small degree sum-of-squares certificate of SSVE (see \Cref{thm:noisy-hypercube-ssve}).

\begin{mtheorem}[Informal \Cref{thm:ssve-main}]\label{thm:ssve-main-intro}
    For all $\delta \in (0,1/2)$, suppose $G$ is an $n$-vertex graph that admits a low-degree sum-of-squares certificate of $\delta$-small-set vertex expansion (SSVE) and has an independent set of size $(\frac{1}{2}-\poly(\delta))n$,
    then there is an algorithm that runs in time $n^{\poly(1/\delta)}$ and outputs an independent set of size $\poly(\delta)n$.
\end{mtheorem}
As an immediate corollary of~\Cref{thm:ssve-main-intro} (see~\Cref{cor:noisy-hyp}), we get a sub-exponential-size certificate that the noisy hypercube does not have large independent sets; see \Cref{sec:noisy-hypercube} for details and discussions.

\paragraph{Our Key Idea: ``solution space geometry'' in the worst-case.}
The main conceptual idea underlying our algorithms is establishing a combinatorial \emph{clustering} property of large independent sets in expanding graphs, which may be interesting on its own.
This property is reminiscent of the cluster structure~\cite{AR06,MMZ05,GS17} in the solution space geometry in \emph{random} optimization problems.
Specifically, we prove that every large independent set in an expander must be ``better-than-random'' correlated with a member of any small list of ``distinct'' independent sets.
More precisely, the intersection of any independent set with some member of the list must be $\Omega(n)$ larger than the expected intersection between random sets of similar sizes.
See \Cref{lem:is-intersection-overview} for a precise statement.
We also show an analogous clustering property for $3$-colorings; see \Cref{lem:coloring-agreement}.

Our algorithms are based on a new rounding scheme for constant-degree sum-of-squares relaxations, which succeeds whenever such an (apparently mild) correlation can be established and formalized as a low-degree sum-of-squares proof.
Our rounding scheme is based on the idea of simulating \emph{multiple independent samples} from the underlying pseudo-distribution, a technique introduced in~\cite{BBKSS} for Unique Games and strengthened further in~\cite{BafnaMinzer}. We discuss our techniques in detail in the next section.

After a version of this manuscript was posted, our new combinatorial clustering property has already inspired follow-up work in extremal combinatorics. In~\cite{Zhu}, Zhu formally studies the solution space geometry of $k$-colorings on $1$-sided expanders, in a variety of regimes of $\lambda_2$ and the correlation that a $k$-coloring must have with a small list of $k$-colorings. In particular, the author generalizes our results (\Cref{lem:coloring-agreement}) to larger $k$, and it is interesting to note that even though an analogue of the clustering property holds for $k$-colorable graphs, the parameters are too weak for our rounding algorithm to obtain a large independent set when $k\geq 4$, as expected from the hardness result in~\Cref{prop:hardness}.
Zhu also shows that when $\lambda_2$ is close to $1$,  there are exponentially many ``distinct colorings''. 

%% file: overview.tex
\section{Technical Overview}
\label{sec:overview}
We provide a brief overview of our rounding framework and analysis in this section. In \Cref{sec:is-on-exp-overview}, we briefly discuss the clustering property and how it leads to our rounding algorithm for one-sided spectral expanders.
The rounding for independent sets on certified small-set vertex expanders, though based on similar principles, is more involved and we discuss it briefly in \Cref{sec:is-on-ssve-overview}.
Then, we describe the proof of the clustering property of independent sets in \Cref{sec:IS-overview} and the clustering property of $3$-colorings in \Cref{sec:3-coloring-overview}. The rounding for 3-colorable graphs follows a similar rounding framework, and we refer the reader to \Cref{sec:3-coloring-main} for details.

\parhead{Polynomial Formulation and SoS Relaxation.} 
Our algorithm rounds a constant-degree sum-of-squares relaxations (see \Cref{sec:sos} for background) of the following system of polynomial inequalities that encode independent sets of size $\geq (1/2-\epsilon)n$ in the input graph on $G(V,E)$.
\begin{equation}\label{eq:is-program}
    \begin{aligned}
        \frac{1}{n} \sum_{u\in V} &x_u \geq \frac{1}{2}-\eps,\\
        & x_u^2 = x_u \mcom & & \forall u\in V \mcom \\
        & x_u x_v = 0 \mcom & & \forall \{u,v\} \in E.
    \end{aligned}
\end{equation}
The relaxation outputs a \emph{pseudo-distribution} over solutions to \eqref{eq:is-program}. For a reader unfamiliar with the sum-of-squares method for algorithm design, it is helpful to think of $\mu$ as constant-degree moments (i.e., expectations under $\mu$ of any constant-degree polynomial of $x$) of a probability distribution over $x\in \zo^n$ satisfying \eqref{eq:is-program}.

\subsection{Rounding Large Independent Sets on One-Sided Spectral Expanders}\label{sec:is-on-exp-overview}
Let $G$ be any regular one-sided spectral expander with $\lambda_2(G) \leq 1-O(\eps)$ containing an independent set of size $(1/2-\eps)n$. Our approach can be summarized as follows: 

\begin{enumerate}[(1)]
\item \label{item:clustering-property}
\textbf{An extremal clustering property of independent sets:} We show (in \Cref{lem:is-intersection-overview}) that there are only two \emph{essentially distinct} $(1/2-\eps)n$-sized independent sets in $G$. Specifically, given any three independent sets $x^{(1)},x^{(2)},x^{(3)}$, at least two of them have a non-trivially large intersection\footnote{Throughout this paper, we write $\E_u$ to denote the expectation over a uniformly random vertex $u \in [n]$ of the input graph. Notice that $\E_u[x_u x'_u]$ then equals $\langle x,x'\rangle/n$.} i.e., $\E_u[x^{(i)}_u x^{(j)}_u] > 1/2-\nu$ for some $\nu \approx 0$ and $i\neq j \in [3]$.

\item \textbf{Recasting large intersection as a polynomial inequality:} Given any three $(1/2-\eps)n$-sized independent sets, $\bx \coloneqq (x^{(1)},x^{(2)},x^{(3)})$, we define $\Phi(\bx) \coloneqq \E_u[x_u^{(1)}x_u^{(2)}]^2+\E_u[x_u^{(2)}x_u^{(3)}]^2+\E_u[x_u^{(1)}x_u^{(3)}]^2$ which is at least $(1/2-\nu)^2 \geq 1/4 - \nu$ as a consequence of \ref{item:clustering-property}.
Here, $\E_u$ is the average with respect to a uniformly random $u \in [n]$ and thus $\Phi(\bx)$ measures the (squared) average pairwise intersections between $x^{(1)}, x^{(2)}, x^{(3)}$.

\item \label{item:low-degree-Phi}
\textbf{A low-degree sum-of-squares proof of largeness of $\Phi(x)$:} We show how the above property can be ``SoS-ized''. That is, $\Phi(\bx)$ is large in expectation over $x^{(1)},x^{(2)},x^{(3)}$ drawn independently from any pseudo-distribution $\mu$ satisfying the independent set 
 constraints in \Cref{eq:is-program}, i.e., $\Phi(\mu) := \pE_{\bx \sim \mu^{\otimes 3}}[\Phi(\bx)] \geq 1/4-\nu$. 

\item \textbf{Rounding: } We give a simple rounding algorithm for $\mu$ with analysis relying on \ref{item:low-degree-Phi} to obtain a large independent set in $G$.
\end{enumerate}

Our rounding analysis actually works as long as the intersection in \ref{item:clustering-property} is \emph{non-trivially larger than expected}, i.e., intersection $\geq (1/4+\nu)n$, where $n/4$ is the expected intersection between random sets of size $\approx n/2$.
However, we would need a different function $\Phi(\bx)$.
For the sake of simplicity, we stick to the case where the intersection is $\geq 1/2-\nu$.


Our final rounding algorithm relies on the idea of \emph{rounding from multiple samples} from a pseudo-distribution first introduced in~\cite{BBKSS}. In their application for rounding Unique Games on certified small-set expanders, they considered a certain  ``shift-partition potential'' (which measured the correlation between two solutions for the input UG instance). Our analysis will rely instead on the above ``average agreement function'' $\Phi(\mu)$.

We will prove the clustering property stated in \ref{item:clustering-property} in \Cref{sec:IS-overview}. Here, let us see how to round when $\Phi(\mu)$ is large.

\parhead{Rounding when $\Phi(\mu)$ is large.} In order to understand the intuition behind our rounding, notice that for 3 \emph{random} subsets of $[n]$ of size $(1/2-\eps)n$, the pairwise agreement function $\Phi$ would be $\approx 3 \cdot (1/4)^2 = 3/16$. Thus, if $\Phi(\mu) \geq 1/4-\nu > 3/16$ for some small $\nu$, then three draws from $\mu$ must be non-trivially correlated. We interpret this property as saying that the (pseudo)-distribution $\mu$ is ``supported" over only two ``distinct'' independent sets. Concretely, 

\[\pE_{\bx \sim \mu^{\otimes 3}}[\Phi(\bx)] = 3 \cdot \pE_{x^{(1)},x^{(2)}\sim \mu} \bracks*{\E_u[x_u^{(1)}x_u^{(2)}]^2} \geq 1/4- \nu \mcom \]
implying that $\pE_{x^{(1)},x^{(2)}}[\E_u[x_u^{(1)}x_u^{(2)}]^2] \geq 1/12-O(\nu) > 1/16+\eta$ for a constant $\eta>0$ if $\nu$ is a small enough constant. Using the independence of $x^{(1)}$ and $x^{(2)}$ this resolves to:
\[\pE_{x^{(1)},x^{(2)}\sim \mu} \bracks*{\E_u[x_u^{(1)}x_u^{(2)}]^2} =\pE \bracks*{\E_{u,v}[x_u^{(1)}x_v^{(1)}x_u^{(2)}x_v^{(2)}] }  =\E_{u,v} \bracks*{\pE[x_u x_v]^2 } > 1/16+\eta \mper \]

The classical idea of \emph{rounding by conditioning} SoS solutions now suggests that we may be able to condition $\mu$ to obtain a $\mu'$ that is essentially supported on a \emph{unique} assignment. Concretely, we argue that by applying a certain repeated conditioning procedure (that reduces ``global correlation''~\cite{BRS11,RT12}) we obtain a modified pseudo-distribution $\mu'$ that satisfies all the original constraints and, in addition, satisfies that for most pairs of vertices $u,v \in [n]$,
we have $\pE_{\mu'}[x_ux_v] \approx \pE_{\mu'}[x_u]\pE_{\mu'}[x_v]$ (where the approximation hides additive constant errors).
Thus,
\[\E_{u,v}\bracks*{\pE_{\mu'}[x_u x_v]^2} \approx \E_{u,v} \bracks*{\pE_{\mu'}[x_u]^2\pE_{\mu'}[x_v]^2} = \E_u\bracks*{\pE_{\mu'}[x_u]^2}^2 > 1/16+\eta \mper \]

We interpret this as saying that two independent ``samples" from the (pseudo)-distribution $\mu'$ have a larger intersection than random sets of same size: $\pE_{x^{(1)},x^{(2)}\sim \mu'}[\E_u[x^{(1)}_u x^{(2)}_u]]=\E_u[\pE_{\mu'}[x_u]^2] > 1/4 + \Omega(\eta)$.
An averaging argument now yields that $\pE_\mu[x_u] > \frac{1}{2}$ for at least an $\Omega(\eta)$ fraction of the vertices. This subset forms an independent set (see~\Cref{fact:obvious-is}) of size $\Omega(\eta n)$.

\subsection{Rounding Independent Set on Small-Set Vertex Expanders}\label{sec:is-on-ssve-overview}
Our rounding and analysis follows a framework similar to the one presented for one-sided spectral expanders above, but each step is a bit more technically involved.
\begin{enumerate}[(1)]
\item \textbf{An extremal clustering property:} If $G$ is a $\delta$-SSVE, then we show that there are at most $t = O(\log(1/\delta))$ \emph{essentially distinct} $(1/2-\eps)n$-sized independent sets. The precise notion of distinctness is more subtle (see~\Cref{lem:is-sos-pf}). But, for this overview, we will work with a generalization of the statement from the previous section: given any set of $2t$ independent sets $\bx = (x^{(1)},\ldots,x^{(2t)})$ on $G$, there exist $t$ of them with intersection $\E_u[x_u^{(i_1)}\ldots x_u^{(i_t)}] > \delta$.
\item \label{item:sos-clustering-ssve}
\textbf{SoSizing the extremal clustering property:} 
If $G$ is a \emph{certified}-SSVE (\Cref{def:certified-VE}), in the sense that there is a degree-$D$ certificate that $\Psi_\delta(G) \geq 2$, then the above fact, formalized appropriately as a polynomial inequality, has a degree $O(tD)$ sum-of-squares proof.
\item \textbf{Rounding: } We use \ref{item:sos-clustering-ssve} to round a pseudo-distribution of degree $\geq O(D)$ to a $\poly(\delta)n$-sized independent set.
\end{enumerate}

\paragraph{Rounding under \ref{item:sos-clustering-ssve}:} Let $\mu$ be a pseudo-distribution satisfying the independent set constraints in \Cref{eq:is-program}. Given the discussion in the previous subsection, we might guess (and this pans out!) that $\binom{2t}{t}$ rounds of conditioning might ensure a $\cD$ with the property that $\pE_{\cD^{\otimes t}}[\E_u[x_u^{(1)}\ldots x_u^{(t)}]] \geq \delta$. Formally, we obtain the inequality below for the pseudo-distribution output by the SoS relaxation:
\[\pE_{\mu^{\otimes t}} \bracks*{ q(\bx) \parens*{ \E_u\bracks*{\prod_{i \in S}x_u^{(i)}} - \delta }} > 0,\]
for some $S\subseteq [2t]$ and SoS polynomial $q(\bx)$. This suggests reweighting $\mu^{\otimes t}$ by $q(\bx)$ (reweighting is a generalization of conditioning of pseudo-distributions introduced in~\cite{BKS17})
to obtain $\pE_{\mu^{\otimes t}|q(\bx)}[\E_u[x_u^{(1)}\ldots x_u^{(t)}]] \geq \delta$. This is almost what we aim for, except that in the reweighted pseudo-distribution, the $t$ copies of $x$ are correlated --- i.e., $\mu^{\otimes t}|q(\bx)$ is no longer a product distribution. Thus, the analog of our rounding analysis from earlier no longer works. To remedy this, we build on tools  from~\cite{BafnaMinzer} to show there is a preprocessing step one can apply so that reweighting by $q(\bx)$ still ensures an approximately-product distribution. 

\subsection{Clustering of Independent Sets in One-Sided Expanders}\label{sec:IS-overview}
Let us now return to the combinatorial guts of our approach. We present in full here, a proof of the following extremal combinatorics statement that eventually can be imported into the low-degree sum-of-squares proof system. 
\begin{lemma}\label{lem:is-intersection-overview}
    Let $G$ be a regular graph containing an independent set of size $(\frac{1}{2}-\eps) n$ and has $\lambda_2(G) \leq 1-C\eps$ for any small enough $\eps$ and some large enough constant $C > 0$.
    Then, for any $3$ independent sets of size at least $(\frac{1}{2}-\eps) n$, two of them have an intersection of size $\geq (\frac{1}{2} - O(\frac{1}{C}) - \eps)n$.
\end{lemma}


\parhead{Intersection between $2$ independent sets.} Let us analyze the intersection between $2$ independent sets $I_1, I_2$ (indicated by $x,y \in \zo^n$) in $G$. By assigning every vertex $u$ of $G$ the label $(x_u, y_u)$, we obtain a partition of vertices of $G$ into subsets with labels in $\{00,01,10,11\}$. Consider now a graph on the label set of $4$ vertices and add an edge between two such labels, say $\ell_1, \ell_2$ (including self-loops) if are vertices $u$ with label $\ell_1$, $v$ with label $\ell_2$ such that $\{u,v\} \in E$ (see \Cref{fig:IS-gadget}).

\begin{figure}[ht!]
    \centering
    \includegraphics[width=0.3\textwidth]{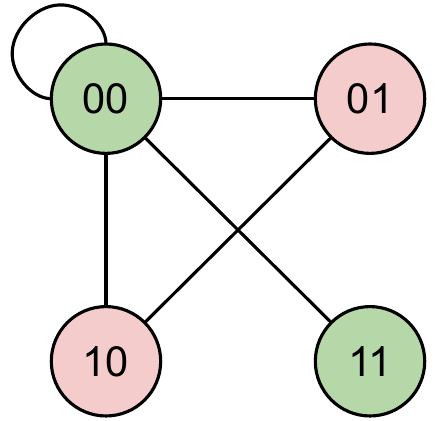}
    \caption{The gadget for $2$ independent sets.
    }
    \label{fig:IS-gadget}
\end{figure}

No edges can exist between $01,10$ and $11$ because $x,y$ indicate independent sets. There can, however, be edges between vertices in the set $00$, hence the self-loop. The graph in \Cref{fig:IS-gadget} is the tensor product $H \otimes H$ where $H$ is graph on $\{0,1\}$ and edges $\{0,0\}$ and $\{0,1\}$. 

Let $\w(ij)$ denote the fraction of vertices $u$ in $G$ such that $x_u = i$ and $y_u = j$. Then, $|I_1\cap I_2|=\w(11)$. Let us now observe:

\begin{claim} 
$\w(00) \leq \w(11) + 2\eps$.
\end{claim}
\begin{proof}
Since $|I_1|,|I_2|\geq (\frac{1}{2}-\eps)n$, we have: $\w(11) + \w(10)$ and $\w(11) + \w(01) \geq \frac{1}{2}-\eps$. Thus, $\w(01) + \w(10) + 2 \w(11) \geq 1-2\eps$.
We use $\w(00) + \w(01) + \w(10) + \w(11) = 1$ to finish. 
\end{proof}

Let's now see how expansion of $G$ enters the picture:
\begin{claim} \label{claim:w11-2-cases}
Fix $\epsilon>0$ small enough. If $\lambda_2 \leq 1-C\eps$ for some large enough constant $C > 0$, then either $\w(00)+\w(11) \leq O(\frac{1}{C})$ or $\w(11) \geq \frac{1}{2} - O(\frac{1}{C}) - \eps$. 
\end{claim}
\begin{proof}
For any subset $S \subseteq V$, we have $e(S,\ol{S}) \geq (1-\lambda_2) \cdot (|S|/n) (1- |S|/n)$.
Applying this to the set of vertices with labels in $\{00,11\}$, we have:
$e(00,01) + e(00,10) \geq (1-\lambda_2) \cdot \w(\{00,11\}) (1 - \w(\{00,11\}))$.

On the other hand, since $G$ is regular, $\w(11) = e(00,11)$ as there are no edges between $01,10$ and $11$.
Similarly, we have $\w(00) = \sum_{\alpha\in\zo^2} e(00,\alpha)$.
Subtracting the two, we get $\w(00) - \w(11) = e(00,00) + e(00,01) + e(00,10) \geq e(00,01) + e(00,10)$.
Therefore, we have
\begin{equation*}
    (1-\lambda_2) \cdot \w(\{00,11\})) (1 - \w(\{00,11\}))
    \leq e(00,01) + e(00,10)
    \leq \w(00) - \w(11) \leq 2\eps \mper 
\end{equation*}

Thus, if $\lambda_2 \leq 1-C\eps$ for some large enough constant $C$, then either $\w(\{00,11\}) \leq \eta$ or $\w(\{00,11\}) \geq 1-\eta$ for $\eta = O(1/C)$. In the latter case, since  $\w(00) \leq \w(11) + 2\eps$, we have $\w(11) \geq \frac{1}{2} - O(\frac{1}{C}) - \eps$.
\end{proof}

\begin{proof}[Proof of \Cref{lem:is-intersection-overview}]
Let's now consider $3$ independent sets. We can now naturally partition the vertices of $G$ into $8$ subsets labeled by elements of $\zo^3$.
In the following, we will use ``$*$'' to denote both possible values. For example, $00*$ means $\{000,001\}$.

From \Cref{claim:w11-2-cases}, we know that $\w(00*)+\w(11*)$ (and analogously $\w(0*0)+\w(1*1)$ and $\w(*00)+\w(*11)$) is either $\leq O(\frac{1}{C}) < \frac{1}{3}$ or $\geq 1 - O(\frac{1}{C})$ for a large enough constant $C$. We now argue that the first possibility cannot simultaneously hold for all three pairs, and thus at least one pair of independent sets must have an intersection of at least $\frac{1}{2}-O(\frac{1}{C})-\eps$, completing the proof. Indeed, $\{00*,11*\} \cup \{0*0,1*1\} \cup \{*00,*11\}$ covers all strings $\zo^3$, since each $\alpha \in \zo^3$ must have either two $0$s or two $1$s. And thus, $\w(\{00*,11*\}) + \w(\{0*0,1*1\}) + \w(\{*00,*11\}) \geq 1$, thus at least one of the three terms exceeds $1/3$.
\end{proof}

\subsection{Clustering in 3-colorable One-Sided Spectral Expanders} \label{sec:3-coloring-overview}

We now discuss an analogous extremal clustering property of $3$-colorings in one-sided spectral expanders. This property is stated in terms of  pairwise ''agreement" between different $3$-colorings --- a natural generalization of intersection that ``mods" out the symmetry between colors.

\begin{definition} The relative \emph{agreement} between two valid 3-colorings $x$ and $y$ according to a permutation $\pi$ is defined by:
\begin{equation*}
    \agree_{\pi}(x,y) \coloneqq \E_{u \in V}\bracks*{\pi(x_u) = y_u} \mcom
\end{equation*}
and the agreement between $x$ and $y$ is defined as the maximum over all permutations:
\begin{equation*}
    \agree(x,y) \coloneqq \max_{\pi\in \bbS_3}\ \agree_{\pi}(x,y) \mper
\end{equation*}
The agreement between two relabelings of the same coloring is $1$ (the maximum possible).
\end{definition}
We will prove the following extremal clustering property of $3$-colorings in a spectral expander that informally says that in any collection of three non-trivial $3$-colorings, two must have a better-than-random agreement.
\begin{lemma} \label{lem:coloring-agreement}
    Let $G = (V,E)$ be a regular $3$-colorable graph with $\lambda_2(G) \leq \frac{\eps}{1+\eps}$ for some small enough $\eps$.
    Then, given any $3$ valid $3$-colorings of $G$ such that no color class has size $> (\frac{1}{2}+\eps)n$, there exist two with an agreement $\geq \frac{1}{2}+\eps$.
\end{lemma}


\parhead{Agreement between $2$ valid $3$-colorings.} We now analyze the agreement between $2$ valid $3$-colorings $x,y \in [3]^n$ of $G$. Similar to \Cref{sec:IS-overview}, the colorings induce a partition of the vertices into $9$ subsets indexed by $\{1,2,3\}^2$, where set $ij$ contains vertices that are assigned $i$ and $j$ by $x$ and $y$ respectively (see \Cref{fig:triangle-gadget}). The $9$-vertex graph in \Cref{fig:triangle-gadget} is exactly $H = K_3 \otimes K_3$ where $K_3$ is a triangle.

\begin{figure}[ht!]
    \centering
    \includegraphics[width=0.45\textwidth]{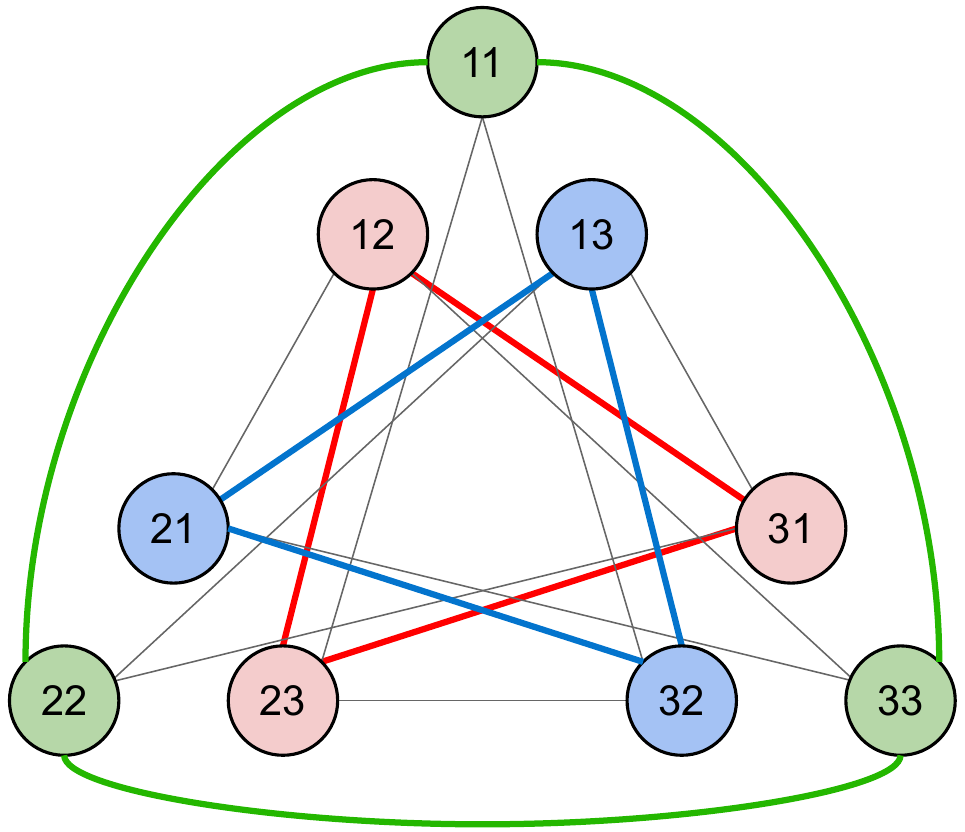}
    \caption{The triangle gadget for $2$ valid $3$-colorings.
    There are $2$ ways to partition the $9$ vertices into $3$ disjoint triangles.
    The highlighted triangles show the partition $\{S_{\pi}\}_{\pi\in \bbS_3^+}$.
    }
    \label{fig:triangle-gadget}
\end{figure}

Define the set
\begin{equation*}
    S_{\pi} \coloneqq \Set{ (\sigma, \pi(\sigma)) : \sigma\in \{1,2,3\} } \mper
\end{equation*}
Then, for any $\pi \in \bbS_3$, $\agree_{\pi}(x,y) = \w(S_{\pi})$. 
\begin{claim}
If $\lambda_2 \leq 1-\frac{1}{1+\epsilon}$, then, 
\begin{equation*}
    \sum_{\pi\in \bbS_3} \w(S_\pi)^2 \geq 2 - \frac{1}{1-\lambda_2} \geq 1 - \eps \mper
    \numberthis \label{eq:sum-of-S-pi-squared}
\end{equation*}
\end{claim}

\begin{proof}
Observe that for any $\pi$, $S_{\pi}$ forms a triangle in $H$.
In fact, there are exactly two ways to partition the $9$ vertex graph above into $3$ disjoint triangles: (1) $\{11,22,33\}$, $\{12,23,31\}$, $\{13,21,32\}$ (highlighted in \Cref{fig:triangle-gadget}), and (2) $\{11,23,32\}$, $\{12,21,33\}$, $\{13,22,31\}$, where each of the $6$ triangles appearing in the list above corresponds to a permutation $\pi \in \bbS_3$.

Now, $e(S_\pi, \ol{S}_\pi) \geq (1-\lambda_2) \cdot \w(S_\pi) (1 - \w(S_\pi))$ for each $\pi$.
Summing up the inequalities over $\pi \in \bbS_3$ gives $(1-\lambda_2) \sum_{\pi} \w(S_\pi) (1 - \w(S_\pi)) = (1-\lambda_2) (2 - \sum_{\pi} \w(S_\pi)^2)$ on the right-hand side and $\sum_{\pi} e(S_\pi, \ol{S}_\pi) = 1$ on the left-hand side. Thus, rearranging gives us
\begin{equation*}
    \sum_{\pi\in \bbS_3} \w(S_\pi)^2 \geq 2 - \frac{1}{1-\lambda_2} \geq 1 - \eps\mper
    \qedhere
\end{equation*}
\end{proof}


\parhead{Small agreement + expansion implies almost bipartite.}
We show the following claim:

\begin{claim} \label{claim:almost-bipartite}
    Suppose $\lambda_2 \leq \frac{\eps}{1+\eps}$ and $\agree(x,y) \leq \frac{1}{2} + \eps$ for small enough $\eps$, then
    one of $\{w(S_\pi)\}_{\pi\in \bbS_3^+}$ and one of $\{w(S_\pi)\}_{\pi\in \bbS_3^-}$ is at most $O(\eps)$.

    As a result, $G$ is almost bipartite, i.e., removing an $O(\eps)$ fraction of vertices results in a bipartite graph.
\end{claim}

Recall that $\agree(x,y) \leq \frac{1}{2}+\eps$ means that $\w(S_\pi) \leq \frac{1}{2}+\eps$ for all $\pi\in \bbS_3$.
To prove \Cref{claim:almost-bipartite}, we formulate it as a $6$-variable lemma (see \Cref{lem:6-variable-lemma}): let $z_1,z_2,\dots,z_6$ be such that $0 \leq z_i \leq \frac{1}{2}+\eps$ for each $i$, $z_1+z_2+z_3 = z_4+z_5+z_6 = 1$, and $\|z\|_2^2 \geq 1 - \eps$, then one of $z_1,z_2,z_3$ and one of $z_4,z_5,z_6$ must be $\leq O(\eps)$.

With this lemma, the first statement in \Cref{claim:almost-bipartite} immediate follows from $\w(S_\pi) \leq \frac{1}{2}+\eps$, $\sum_{\pi\in \bbS_3^+} \w(S_\pi) = \sum_{\pi\in \bbS_3^-} \w(S_\pi) = 1$, and \Cref{eq:sum-of-S-pi-squared}.

For the second statement, let $\pi^+\in \bbS_3^+$ and $\pi^- \in \bbS_3^-$ be the permutations such that $\w(S_{\pi^+})$, $\w(S_{\pi^-}) \leq O(\eps)$.
Note that since $\pi^+$ and $\pi^-$ have different signs, $S_{\pi^+}$ and $S_{\pi^-}$ intersect in exactly one string $\alpha \in [3]^2$.
In fact, $\alpha$ uniquely determines $\pi^+,\pi^-$ since there are exactly two permutations with different signs that map $\alpha_1$ to $\alpha_2$.
Assume without loss of generality (due to symmetry) that $\alpha = 11$, so that $S_{\pi^+} = \{11,22,33\}$ and $S_{\pi^-} = \{11,23,32\}$.
Then, we have $\w(\{11,22,33\})$, $\w(\{11,23,32\}) \leq O(\eps)$.
This means that $\w(\{12,13,21,23\}) \geq 1-O(\eps)$.
Observe that $\{12,13,21,23\}$ forms a bipartite structure between $\{12,13\}$ and $\{21,23\}$, as shown in \Cref{fig:almost-bipartite}.
In particular, the first coloring labels the entire left side with the same color, while the second labels the right side with the same color.

\begin{figure}[ht!]
    \centering
    \includegraphics[width=0.45\textwidth]{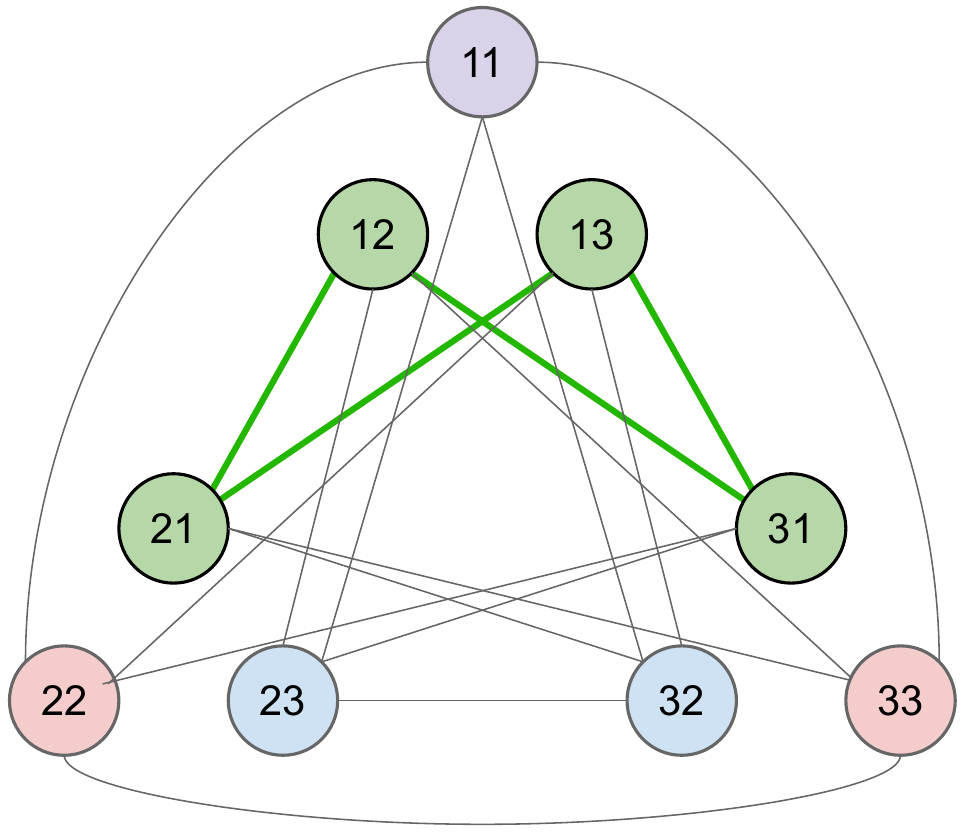}
    \caption{$S_{\pi^+} = \{11,22,33\}$ and $S_{\pi^-} = \{11,23,32\}$, and $\w(S_{\pi^+})$, $\w(S_{\pi^-}) \leq O(\eps)$, which means that
    $\w(\{12,13,21,23\}) \geq 1-O(\eps)$.
    Here $\{12,13,21,23\}$ forms a bipartite structure.
    }
    \label{fig:almost-bipartite}
\end{figure}

\parhead{Agreement between $3$ valid $3$-colorings.}
Naturally, we consider the graph as being partitioned into $27$ subsets indexed by strings $[3]^3$.
Again, we will use ``$*$'' to denote ``free'' coordinate, so for example $11*$ means $\{111,112,113\}$, i.e., the set $11$ if we ignore the third coloring.

Suppose for contradiction that the agreement between each pair of $3$-colorings is at most $\frac{1}{2}+\eps$.
Then, by \Cref{claim:almost-bipartite}, we have that each pair $(i,j)$ of colorings gives a bipartite structure, denoted $T^{(ij)}$, such that $\w(T^{(ij)}) \geq 1-O(\eps)$.
This is best explained by example.
Suppose $T^{(12)} = \{12*,13*,21*,23*\}$, $T^{(13)} = \{1*2,1*3,2*1,3*1\}$ and $T^{(23)} = \{*11,*13,*22,*32\}$.
Then, we can see that $T \coloneqq T^{(12)} \cap T^{(13)} \cap T^{(23)} = \{122,132,211,311\}$.
This is a bipartite structure between $\{122,132\}$ and $\{211,311\}$, where the first coloring labels the entire left side with the same color, while the second and third label the right side with the same color.

Moreover, we have $\w(T) \geq 1 - O(\eps)$.
We now use this to derive a contradiction.
Suppose no colors have size larger than $(\frac{1}{2} + \eps)n$, so $\w(\{122,132\})$, $\w(\{211,311\}) \leq \frac{1}{2}+\eps$.
This implies that $\w(\{122,132\})$, $\w(\{211,311\}) \geq \frac{1}{2}- O(\eps)$.
Next, observe that $\{122,211,311\} \subseteq \{*11,*22\} \subseteq S_{\pi}$ between the second and third colorings for some $\pi$.
Similarly, $\{132,211,311\} \subseteq \{*11,*32\} \subseteq S_{\pi'}$ for some $\pi'$.
Thus, one of them has weight at least $\w(\{211,311\}) + \frac{1}{2} \w(\{122,132\}) \geq \frac{3}{4}-O(\eps)$, contradicting that each pairwise agreement is $\leq \frac{1}{2}+\eps$.

One can verify that the above holds in general; $T$ will contain at most $4$ strings in $[3]^3$ and form the bipartite structure explained above.
This proves \Cref{lem:coloring-agreement}.

%% file: prelims.tex
\section{Preliminaries}
\label{sec:prelims}

\parhead{Notations.}
For any integer $N$, we write $[N] \coloneqq \{1,2,\dots,N\}$.
We will use boldface $\bx$ to denote a collection of vectors: $\bx = (x^{(1)},x^{(2)},\dots,x^{(t)})$.
For a graph $G = (V,E)$ and a subset $S\subseteq V$, we denote the neighbors of $S$ (a.k.a.\ outer boundary) as $N_G(S) \coloneqq \{u\notin S: \exists v\in S, (u,v) \in E \}$, and the \emph{neighborhood} of $S$ as $\Gamma_G(S) \coloneqq S \cup N_G(S)$.
We use $\lambda_2(G)$ to denote the second eigenvalue of the \emph{normalized} adjacency matrix $D_G^{-1/2} A_G D_G^{-1/2}$.
The Laplacian $L_G$ is the matrix $D_G - A_G$ where $D_G$ is the diagonal degree matrix and $A_G$ is the adjacency matrix.
We use $\wt{L}_G$ to denote the normalized Laplacian.
In all of the above, we drop the dependence on $G$ if it is clear from context.

\subsection{Background on Sum-of-Squares}
\label{sec:sos}

We refer the reader to the monograph~\cite{FKP19} and the lecture notes~\cite{BS16} for a detailed exposition of the sum-of-squares method and its usage in algorithm design. 

\parhead{Pseudo-distributions.}
Pseudo-distributions are generalizations of probability distributions.
Formally, a pseudo-distribution on $\R^n$ is a finitely supported \emph{signed} measure $\mu :\R^n \rightarrow \R$ such that $\sum_{x} \mu(x) = 1$. The associated \emph{pseudo-expectation} is a linear operator $\pE_\mu$ that assigns to every polynomial $f:\R^n \rightarrow \R$ the value $\pE_\mu f = \sum_{x} \mu(x) f(x)$, which we call the pseudo-expectation of $f$. We say that a pseudo-distribution $\mu$ on $\R^n$ has \emph{degree} $d$ if $\pE_\mu[f^2] \geq 0$ for every polynomial $f$ on $\R^n$ of degree $\leq d/2$.

A degree-$d$ pseudo-distribution $\mu$ is said to satisfy a constraint $\{q(x) \geq 0\}$ for any polynomial $q$ of degree $\leq d$ if for every polynomial $p$ such that $\deg(p^2) \leq d-\deg(q)$, $\pE_\mu[p^2 q] \geq 0$.
For example, in this work we will often say that $\mu$ satisfies the Booleanity constraints $\{x_i^2 - x_i = 0,\ \forall i\in[n]\}$, which means that $\pE_{\mu}[p(x) (x_i^2 - x_i)] = 0$ for any $i$ and any polynomial $p$ of degree $d-2$.
We say that $\mu$ $\tau$-approximately satisfies a constraint $\{q \geq 0\}$ if for any sum-of-squares polynomial $p$, $\pE_\mu[pq] \geq - \tau \Norm{p}_2$ where $\Norm{p}_2$ is the $\ell_2$ norm of the coefficient vector of $p$. 

We rely on the following basic connection that forms the basis of the sum-of-squares algorithm.

\begin{fact}[Sum-of-Squares algorithm, \cite{Par00,Las01}] \label{fact:sos-algorithm}
Given a system of degree $\leq d$ polynomial constraints $\{q_i \geq 0\}$ in $n$ variables and the promise that there is a degree-$d$ pseudo-distribution satisfying $\{q_i \geq 0\}$ as constraints, there is a $n^{O(d)} \polylog( 1/\tau)$ time algorithm to find a pseudo-distribution of degree $d$ on $\R^n$ that $\tau$-approximately satisfies the constraints $\{q_i \geq 0\}$.
\end{fact}

\paragraph{Sum-of-squares proofs.}
Let $f_1,f_2,\dots,f_m$ and $g$ be multivariate polynomials in $x$.
A \emph{sum-of-squares} proof that the constraints $\{f_1 \geq 0,\dots, f_m \geq 0\}$ imply $g \geq 0$ consists of sum-of-squares polynomials $(p_S)_{S\subseteq[m]}$ such that $g = \sum_{S\subseteq[m]} p_S \prod_{i\in S} f_i$.
The \emph{degree} of such a sum-of-squares proof equals the maximum of the degree of $p_S \prod_{i \in S} f_i$ over all $S$ appearing in the sum above. We write $\{f_i \geq 0,\ \forall i\in[m] \} \sststile{d}{x} \{g \geq 0\}$ where $d$ is the degree of the sum-of-squares proof.

We will rely on the following basic connection between SoS proofs and pseudo-distributions:

\begin{fact}
    Let $f_1,\dots,f_m$ and $g$ be polynomials, and let $\calA = \{f_i(x) \geq 0,\ \forall i\in[m]\}$.
    Suppose $\calA \sststile{d}{x} \{g(x) \geq 0\}$.
    Then, for any pseudo-distribution $\mu$ of degree $\geq d$ satisfying $\calA$, we have $\pE_{\mu}[g] \geq 0$.
\end{fact}

Therefore, an SoS proof of some polynomial inequality directly implies that the same inequality holds in pseudo-expectation.
We will use this repeatedly in our analysis.

\paragraph{SoS toolkit.}
The theory of univariate sum-of-squares (in particular, Luk\'{a}cs Theorem) says that if a univariate polynomial is non-negative on an interval, then this fact is also SoS-certifiable. The following corollary of Luk\'{a}cs theorem is well-known, and we will use it multiple times to convert univariate inequalities into SoS inequalities in a blackbox manner.

\begin{fact}[Corollary of Luk\'{a}cs Theorem] \label{fact:univariate-interval}
    Let $a \leq b \in \R$.
    Let $p\in \R[x]$ be a univariate real polynomial of degree $d$ such that $p(x) \geq 0$ for all $a \leq x \leq b$. 
    Then,
    \begin{equation*}
        \Set{x \geq a,\ x \leq b} \sststile{d}{x} \Set{ p(x) \geq 0 }\mper
    \end{equation*}
\end{fact}

Similarly, true inequalities on the hypercube are also SoS-certifiable.

\begin{fact} \label{fact:boolean-function}
    Let $p$ be a polynomial in $n$ variables.
    Suppose $p(x) \geq 0$ for all $x\in\zo^n$, then
    \begin{equation*}
        \Set{x_i^2-x_i = 0,\ \forall i\in[n]} \sststile{\max(n,\deg(p))}{x} \Set{ p(x) \geq 0 }\mper
    \end{equation*}
\end{fact}

More generally, all true inequalities have SoS certificates under mild assumptions.
In particular, Schm\"{u}dgen's Positivstellensatz establishes the completeness of the SoS proof system under compactness conditions (often called the Archimedean condition).
Moreover, bounds on the SoS degree (given the polynomial and the constraints) were given in \cite{PD01,Schweighofer04}.

\begin{fact}[Positivstellensatz \cite{PD01,Schweighofer04}] \label{fact:black-box-SoS}
    For all polynomials $g_1,g_2,\dots,g_m$ over $x = (x_1,x_2,\dots,x_n)$ defining a non-empty set
    \begin{equation*}
        S \coloneqq \{x\in \R^n: g_1(x) \geq 0,\dots, g_m(x) \geq 0\} \subseteq (-1,1)^n \mcom
    \end{equation*}
    and for every polynomial $f$ of degree $d$ with coefficients bounded by $R$ and $f^* \coloneqq \min_{x\in S} f(x) > 0$, there exists an integer $D = D(n,g_1,\dots,g_m,R,f^*) \in \N$ such that
    \begin{equation*}
        \Set{g_1 \geq 0, \dots, g_m\geq 0}
        \sststile{D}{x}
        \Set{f \geq 0} \mper
    \end{equation*}
\end{fact}

\paragraph{Independent samples from a pseudo-distribution.}
Recall that a given pseudo-expectation operator $\pE_\mu$ has the interpretation as averaging of functions $f(x)$ over a pseudo-distribution $x\sim\mu$. We will need to be able to mimic averaging over $t$ independently chosen samples $x^{(1)},\ldots, x^{(t)} \sim \mu$.\footnote{This was also used in~\cite{BBKSS} and \cite{BafnaMinzer} in the context of SoS algorithms for Unique Games.}
We define the product pseudo-distribution $\mu^{\otimes t}$ along with pseudo-expectation $\pE_{\mu^{\otimes t}}$ as follows:
let $p(\bx) = (x^{(1)})^{\alpha_1}\ldots(x^{(t)})^{\alpha_t}$ be a monomial in variables $\bx = (x^{(1)},\ldots,x^{(t)})$; we define
\begin{equation*}
    \pE_{\mu^{\otimes t}}[p] := \pE_{\mu}[x^{\alpha_1}] \cdot \pE_{\mu}[x^{\alpha_2}] \cdots \pE_\mu[x^{\alpha_t}] \mper
\end{equation*}

It is easy to check that $\pE_{\mu^{\otimes t}}$ is also a pseudo-expectation operator corresponding to $t$ independent samples from the pseudo-distribution $\mu$.
\begin{fact}\label{fact:indep}
If $\pE_{\mu}$ is a valid pseudo-distribution of degree $D$ in variables $x$, then $\pE_{\mu^{\otimes t}} $ is a valid pseudo-distribution of degree $D$. Furthermore, if additional SoS inequalities are true for $\pE_\mu$, they also hold for $\pE_{\mu^{\otimes t}}$.
\end{fact}

\subsection{Basic Facts on Independent Sets}

The following is the well-known $2$-approximation algorithm for minimum vertex cover.
\begin{fact}\label{fact:obvious-algo}
If an $n$-vertex graph $G$ has an independent set of size at least $(\frac{1}{2}+\eps)n$, then there exists a polynomial-time algorithm that outputs an independent set of size at least $2\eps n$.
\end{fact}

We will also rely on the following simple fact repeatedly.

\begin{fact}\label{fact:obvious-is}
For a graph $G = (V,E)$,
let $\mu$ be a pseudo-distribution of degree at least $2$ that satisfies the independent set constraints, i.e., $x_u^2 = x_u$ for all $u\in V$ and $x_u x_v = 0$ for all $\{u,v\} \in E$.
Then, the set of vertices $\{u \in V: \pE_\mu[x_u] > \frac{1}{2} \}$ forms an independent set in $G$.
\end{fact}
\begin{proof}
For all $\{u,v\}\in E$, from the independent set constraints we can derive that $(x_u + x_v)^2 = x_u^2 + 2x_u x_v + x_v^2 = x_u + x_v$, i.e., $(x_u + x_v)$ satisfies the Booleanity constraint, thus $x_u + x_v \leq 1$.
Thus, we have $\pE_{\mu}[x_u+x_v] \leq 1$, which means that $u,v$ cannot both be in the set $\{u\in V: \pE_{\mu}[x_u] > \frac{1}{2}\}$.
\end{proof}

\subsection{Information Theory}
We will use $\mu|_R$ to denote the marginal distribution of a random variable $R \sim \mu$. We use $TV(X,Y)$ to denote the total-variation distance between two distributions $X,Y$.

\begin{definition}[Mutual Information]
Given a distribution $\mu$ over $(X,Y)$, the mutual information between $X,Y$ is defined as:
\[I_\mu(X;Y) = D_{KL}(\mu~ || ~\mu|_X \times \mu|_Y),\]
where $D_{KL}$ is the Kullback-Leibler divergence. We drop $\mu$ from the subscript when the distribution is clear from context. The conditional mutual information between $(X,Y)$ with respect to a random variable $Z$ is defined as:
\[I(X;Y | Z) = \E_{z \sim Z}[I_{\mu | Z = z}(X;Y)].\]
\end{definition}

\begin{fact}[Pinsker's inequality] \label{fact:pinskers}
Given any two distributions $D_1,D_2$:
\[TV(D_1, D_2) \leq \sqrt{\frac{1}{2}D_{KL}(D_1 || D_2)}.\]
\end{fact}

\subsection{Conditioning Pseudo-distributions}

We can {\em reweigh} or {\em condition} a degree-$D$ pseudo-distribution $\mu$ by a polynomial $s(x)$, where $s(x)$ is non-negative under the program axioms, i.e., $\cA \sststile{d}{x} \{s(x) \geq 0\}$ for $d < D$. Technically, this operation defines a new pseudo-distribution $\mu'$ of degree $D-d$ with pseudo-expectation operator $\pE_{\mu'}$ by taking
\[\pE_{\mu'}[x^{\alpha}] = \frac{\pE_{\mu}[x^{\alpha} \cdot s(x)]}{\pE_{\mu}[s(x)]} \mcom \]
for every monomial $x^{\alpha}$ of degree at most $D - d$.

It is easy to verify that $\mu'$ is a valid pseudo-distribution of degree $D-d$ and satisfies the axioms of the original $\mu$.
As an example, under the independent set axioms presented in \eqref{eq:is-program}, since $x_i \geq 0$ is an axiom, one can reweigh $\mu$ by $x_i$, essentially ``conditioning'' the pseudo-distribution on the event $x_i = 1$.
Thus, we will also refer to this operation as conditioning and denote $\mu'$ by $\mu|s(x)$.
Often times, the polynomial $s(x)$ we will ``condition'' on will be a polynomial approximation of the indicator function of some event $E$.
In this case, the above operation can be interpreted as conditioning $\mu$ to satisfy some properties specified by the event $E$.

\paragraph{Approximate polynomials to indicator functions.}
Our arguments require indicators of events such as $f(x)\geq \delta$ where $f$ is a low-degree polynomial, and we will need to condition on such events. Strictly speaking, the function $\Ind[f(x)\geq \delta]$ is not a low-degree polynomial and therefore we cannot condition on it. However, it is not difficult to show that such indicators can be approximated by low-degree polynomials, and in particular we use the the following result, due to~\cite{diakonikolas2010bounded}, that provides a low-degree approximation to a step function.

\begin{lemma}\label{lem:step-approx}
For every $0 < \nu < \delta < 1$, there is a univariate polynomial $Q_{\delta,\nu}$ of degree $O(\frac{1}{\nu}\log^2\frac{1}{\nu})$ such that:
\begin{enumerate}
\item $Q_{\delta,\nu}(x) \in [0,\nu]$ for all $x \in [0, \delta-\nu]$.
\item $Q_{\delta,\nu}(x) \in [1,1+\nu]$ for all $x \in [\delta, 1]$.
\item $Q_{\delta,\nu}$ is monotonically increasing on $(\delta-\nu, \delta)$.
\end{enumerate}
Furthermore, all these facts are SoS-certifiable in degree $\deg(Q_{\delta,\nu})$.
\end{lemma}

\paragraph{Reducing average correlation.}
An important technique we need is reducing the average correlation of random variables through conditioning, which was introduced in \cite{BRS11} (termed global correlation reduction) and is also applicable to pseudo-distributions of sufficiently large degree.
We will use the following version from~\cite{RT12}.

\begin{lemma}\label{lem:ragh-tan}
Let $Y_1,\ldots, Y_M$ be a set of random variables each taking values in $\{1,\ldots, q\}$.
Then, for any $\ell \in \N$, there exists $k \leq \ell$ such that: 
\[\E_{i_1,\ldots,i_k \sim [M]}\E_{i,j \sim [M]}[I(Y_i;Y_j \mid Y_{i_1}, \ldots,Y_{i_k})] \leq \frac{\log q}{\ell - 1}.\]
\end{lemma}
Note that the above lemma holds as long as there is a local collection of distributions over $(Y_1,\ldots, Y_M)$ that are valid probability distributions over all collections of $\ell+2$ variables and are consistent with each other. Of particular interest to us would be the setting where we have a degree $\geq \ell+2$-pseudo-distribution $\mu$ over the variables $(Y_1,\ldots, Y_M)$.

We also require a generalization of the above lemma to $t$-wise correlations.
\begin{lemma}[Lemma 32 of \cite{MR17}]\label{lem:total-corr-reduction}
Let $Y_1,\ldots, Y_M$ be a set of random variables each taking values in $\{1,\ldots, q\}$.
The total $t$-wise correlation of a distribution $\mu$ over $Y_1,\dots,Y_M$ is defined as
\[
\TC_t(\mu) \coloneqq \E_{i_1,\dots,i_t \sim [M]} \bracks*{\KL((Y_{i_1},\dots,Y_{i_t}) \| Y_{i_1} \times \cdots \times Y_{i_t}) } \mper
\]
Then, for any $\ell \in \N$, there exists $k \leq \ell$ such that: 
\[\E_{\substack{i_1,\ldots,i_k \sim [M] \\ (y_{i_1},\ldots, y_{i_k}) \sim \mu}}[\TC_t(\mu \mid Y_{i_1}=y_{i_1},\ldots,Y_{i_k}=y_{i_k})] \leq \frac{t^2 \log q}{\ell} \mper \]
\end{lemma}
Similar to \Cref{lem:ragh-tan}, the above holds for pseudo-distributions of degree $\geq \ell+t$.

%% file: is-on-expanders.tex
\section{Independent Sets on Spectral Expanders}
\label{sec:IS-on-edge-expanders}

We prove the following theorem in this section:
\begin{theorem}[Restatement of \Cref{thm:IS-expander-main}] \label{thm:IS-expander-formal}
    There is a polynomial-time algorithm that, given an $n$-vertex regular graph $G$ that contains an independent set of size $(\frac{1}{2}-\eps)n$ and has $\lambda_2(G) \leq 1 - 40 \eps$ for any $\eps \leq 0.001$, outputs an independent set of size at least $n/1000$.
\end{theorem}

Our algorithm starts by considering a constant degree SoS relaxation of the integer program for Independent Set \eqref{eq:is-program} and obtaining a pseudo-distribution $\mu'$. We then apply a simple rounding algorithm to obtain an independent set in $G$ as shown below.

\begin{mdframed}
    \begin{algorithm}[Find independent set in an expander]
    \label{alg:IS-on-expander}
    \mbox{}
      \begin{description}
      \item[Input:] A graph $G = (V,E)$.
      \item[Output:] An independent set of $G$.
      \item[Operation:] \mbox{}
        \begin{enumerate}
            \item Solve a degree $D = O(1)$ SoS relaxation of the integer program \eqref{eq:is-program} to obtain a pseudo-distribution $\mu'$.
            \item Choose a uniformly random set of $t=O(1)$ vertices $i_1,\ldots,i_t \sim [n]$ and draw $(\sigma_{i_1},\ldots,\sigma_{i_t}) \sim \mu'$. Let $\mu$ be the pseudo-distribution obtained by conditioning $\mu'$ on \\ $(x_{i_1} = \sigma_{i_1},\ldots, x_{i_t}=\sigma_{i_t})$.
            \item Output the set $\{u\in V: \pE_{\mu}[x_u] > \frac{1}{2}\}$.
        \end{enumerate}
      \end{description}
    \end{algorithm}
\end{mdframed}

\subsection{Multiple Assignments from \texorpdfstring{$\mu$}{mu}: Definitions and Facts}

Fix $t\in \N$. Throughout this section, we will work with $t$ assignments $x^{(1)},x^{(2)},\dots,x^{(t)}$ that the reader should think of as independent samples from the pseudo-distribution $\mu$, i.e. each $x^{(i)}$ is an $n$-dimensional vector which is the indicator of a $(1/2-\eps)n$-sized independent set in $G$ and therefore it satisfies the constraints of the integer program \eqref{eq:is-program}. Given $x^{(1)},x^{(2)},\dots,x^{(t)}$ we use boldface $\bx$ to denote $(x^{(1)},\dots,x^{(t)})$, i.e., the collection of variables $x_u^{(i)}$ for $u\in [n]$ and $i\in [t]$. Moreover, for $U \subseteq [t]$, we write $\bx^U \coloneqq (x^{(i)})_{i\in U}$.

\begin{definition} \label{def:IS-formulation}
We denote $\BoolConst_G(x) \coloneqq \{x_u^2 - x_u = 0,\ \forall u \in V\}$, i.e., the Booleanity constraints. Moreover, we write $\ISConst_G(x)$ to denote the \emph{independent set constraints}:
    \begin{equation*}
        \ISConst_G(x) \coloneqq \BoolConst_G(x) \cup \{x_u x_v = 0,\ \forall \{u,v\} \in E\} \mper
    \end{equation*}
    Moreover, with slight abuse of notation, for $t\in \N$ and vectors $x^{(1)},x^{(2)},\dots,x^{(t)}$,
    \begin{equation*}
        \BoolConst_G(\bx) \coloneqq \bigcup_{i\in[t]} \BoolConst_G(x^{(i)}) \mcom
        \quad \ISConst_G(\bx) \coloneqq \bigcup_{i\in[t]} \ISConst_G(x^{(i)}) \mper
    \end{equation*}
    We will drop the dependence on $G$ when the graph is clear from context.
\end{definition}

Given assignments $x^{(1)},\dots,x^{(t)} \in \zo^n$ and $\alpha \in \zo^t$, for each vertex $u\in[n]$, we define below the event that $u$ is assigned $\alpha_i$ by $x^{(i)}$, which is viewed as a degree-$t$ multilinear polynomial of $\bx$. Similarly, for $S\subseteq \zo^t$, we define the event that $u$ receives one of the assignments in $S$.
\begin{definition} \label{def:w}
    Let $t\in \N$, and let $\bx = (x^{(1)}, x^{(2)},\dots, x^{(t)})$. For each $u \in [n]$, $\alpha \in \zo^t$ and $S \subseteq \{0,1\}^t$, we define the following events,
    \begin{gather*}
        \1(u\gets \alpha) \coloneqq \1(x_u^{(1)}=\alpha_1,\ldots, x_u^{(t)}=\alpha_t)= \prod_{i \in [t]} \Paren{x_u^{(i)}}^{\alpha_i} \Paren{1- x_u^{(i)}}^{1-\alpha_i} \mcom \\
        \1(u\gets S) \coloneqq \sum_{\alpha\in S} \1(u\gets \alpha) \mper
    \end{gather*}
    For convenience, we omit the dependence on $\bx$.
    We will also consider the quantity $\w(\alpha)$ which is the fraction of vertices that get assigned $\alpha$:
    \begin{align*}
        \w(\alpha) \coloneqq \E_{u\in[n]}[\1(u \gets \alpha)] \mper
    \end{align*}
    Similarly, $\w(S) \coloneqq \E_{u\in [n]}[\1(u\gets S)]$ for $S\subseteq \zo^t$.
    
    Moreover, we will use the symbol ``$*$'' to denote ``free variables'' ---
    for $\beta \in \{0,1,*\}^t$, $\1(u\gets \beta) \coloneqq \1(u\gets S_\beta)$ and $\w(\beta) \coloneqq \w(S_\beta)$ where $S_\beta = \{\alpha \in \zo^t: \alpha_i = \beta_i \text{ if } \beta_i \neq *\}$.
    For example, $\w(00*) = \w(000) + \w(001)$.
\end{definition}


We note some simple facts (written in SoS form) that will be useful later.
\begin{fact} \label{fact:u-alpha-facts}
The following can be easily verified:
\begin{enumerate}[(1)]
    \item $\BoolConst(\bx) \sststile{2t}{\bx} \Set{ \1(u\gets \alpha)^2 = \1(u\gets \alpha) }$, i.e., $\1(u\gets \alpha)$ satisfies the Booleanity constraint.

    \item $\BoolConst(\bx) \sststile{2t}{\bx} \Set{ \1(u\gets \alpha) \cdot \1(u\gets \beta) = 0 }$ for $\alpha \neq \beta$.
    This also implies that $\1(u\gets S)$ satisfies the Booleanity constraint for any $S \subseteq \zo^t$.

    \item $\sststile{t}{\bx} \Set{ \sum_{\alpha\in \zo^t} \1(u\gets \alpha) = 1 }$.
\end{enumerate}
    
\end{fact}

We next prove the following lemma, which is an ``SoS proof'' that if $x^{(1)},\dots,x^{(t)}$ are indicators of independent sets and $\{u,v\}\in E$, then $u$ and $v$ cannot be both assigned $1$ by any $x^{(i)}$.
As a consequence, any vertex that is assigned all $1$s can only be connected to vertices that are assigned all $0$s, meaning that if $v$ gets $\vec{1}$ then $u$ must get $\vec{0}$.

\begin{lemma} \label{lem:ones-to-zeros}
    Let $t\in \N$ and $\bx = (x^{(1)},x^{(2)},\dots,x^{(t)})$ be variables.
    For any graph $G= (V,E)$ and any $\alpha,\beta\in \{0,1\}^t$ such that $\supp(\alpha) \cap \supp(\beta) \neq \varnothing$, then for all $\{u,v\}\in E$ we have
    \begin{equation*}
        \ISConst_G(\bx) \sststile{2t}{\bx} \Set{ \1(u\gets\alpha) \1(v\gets\beta) = 0 } \mper
    \end{equation*}
    In particular, for all $\{u,v\}\in E$,
    \begin{equation*}
    \begin{aligned}
        \ISConst_G(\bx) & \sststile{2t}{\bx} \Set{ \Paren{1 - \1(u\gets\vec{0})} \cdot \1(v\gets\vec{1}) = 0 } \\
        & \sststile{2t}{\bx} \Set{ \1(u\gets \vec{0}) \geq \1(v\gets\vec{1}) } \mper
    \end{aligned}
    \end{equation*}
\end{lemma}
\begin{proof}
    Let $i\in [t]$ be the index such that $\alpha_i = \beta_i = 1$.
    Then, by \Cref{def:w}, $\1(u\gets\alpha) \cdot \1(v\gets\beta) = x_u^{(i)} x_v^{(i)} \cdot f(\bx)$ for some polynomial $f$ (not depending on $x^{(i)}$).
    The first statement follows since $x_u^{(i)} x_v^{(i)} = 0$ is in the independent set constraints.
    
    For the second statement, $\paren{1 - \1(u\gets\vec{0})} \1(v\gets\vec{1}) = 0$ follows from the polynomial equality $\sum_{\alpha\in \zo^t} \1(u\gets\alpha) = 1$ and that $\vec{1}$ intersects with all $\alpha \neq \vec{0}$.
    Moreover, $\1(u\gets \alpha)$ satisfies the Booleanity constraints (\Cref{fact:u-alpha-facts}).
    Denoting $a \coloneqq \1(u\gets\vec{0})$ and $b \coloneqq \1(u\gets\vec{1})$ for convenience, from $(1-a)b = 0$ and $a^2=a$, $b^2=b$ we have
    \begin{equation*}
        a-b = (a-b)^2 - 2(1-a)b + (a - a^2) + (b - b^2) \geq 0 \mcom
    \end{equation*}
    which completes the proof.
\end{proof}

\subsection{Spectral Gap implies a Unique Solution}
Recall the definitions from \Cref{def:w}.
We need some definitions for edge sets in the graph.

\begin{definition}
    Let $G$ be a graph, let $t\in \N$, and let $\bx = (x^{(1)},\dots,x^{(t)})$.
    For $\alpha,\beta \in \zo^t$, define
    \begin{equation*}
        e(\alpha,\beta) \coloneqq \frac{1}{2|E(G)|} \sum_{\{u,v\}\in E(G)} \1(u\gets \alpha) \1(v\gets \beta) + \1(u\gets \beta) \1(v\gets \alpha) \mper
    \end{equation*}
    Here, we omit the dependence on $\bx$ and $G$ for simplicity.

    Similarly, for $S,T \subseteq \zo^t$, we denote $e(S,T) \coloneqq \sum_{\alpha\in S} \sum_{\beta \in T} e(\alpha,\beta)$.
\end{definition}

Given assignments $x^{(1)},\dots,x^{(t)} \in \zo^n$ and $\alpha,\beta \in \zo^t$, one should view $e(\alpha,\beta)$ as the (normalized) number of edges between vertices that are assigned $\alpha$ and vertices assigned $\beta$.
We note a few properties which can be easily verified:

\begin{fact} \label{fact:e-facts}
    The following can be easily verified:
    \begin{enumerate}[(1)]
        \item Symmetry: $e(\alpha,\beta) = e(\beta,\alpha)$ by definition.
    
        \item Sum of edge weights (double counted) equals $1$: $\sststile{2t}{\bx} \set{ \sum_{\alpha,\beta \in \zo^t} e(\alpha,\beta) = 1 }$, which follows from $\sststile{t}{\bx} \set{ \sum_{\alpha\in \zo^t} \1(u\gets \alpha) = 1 }$.
    
        \item For a regular graph, the weight of a subset equals the weight of incident edges:
        for any $\alpha\in \zo^t$,
        $\sststile{2t}{\bx} \{ \sum_{\beta\in\zo^t} e(\alpha,\beta) = \w(\alpha) \}$.
    
        \item $\ISConst_G(\bx) \sststile{2t}{\bx} \Set{e(\alpha,\beta) = 0}$ for any $\alpha,\beta$ such that $\supp(\alpha) \cap \supp(\beta) \neq \varnothing$ due to \Cref{lem:ones-to-zeros}.
    \end{enumerate}
\end{fact}

We next show the following lemma relating the Laplacian to the cut in the graph.

\begin{lemma} \label{lem:laplacian}
    Let $G$ be a graph and $L_G$ be its Laplacian matrix.
    Let $t\in \N$, $S \subseteq \zo^t$, and let $y_u \coloneqq \1(u\gets S)$ for each vertex $u$, we have
    \begin{equation*}
        \BoolConst(\bx) \sststile{2t}{\bx} \Set{\frac{1}{2|E(G)|} \cdot y^\top L_G y = e(S, \ol{S}) } \mper
    \end{equation*}
\end{lemma}
\begin{proof}
    Since $y_u$ satisfies the Booleanity constraint and $1 - \1(u\gets S) = \1(u\gets \ol{S})$, for any $u,v$,
    \begin{equation*}
    \begin{aligned}
        \BoolConst(\bx) \sststile{2t}{\bx} 
        \Big\{ (y_u - y_v)^2 
        &= \1(u \gets S) + \1(v \gets S) - 2 \cdot \1(u\gets S) \1(v\gets S) \\
        &= \1(u\gets S) \1(v\gets \ol{S}) + \1(v \gets S) \1(u \gets \ol{S})  \Big\}
        \mper
    \end{aligned}
    \end{equation*}
    The lemma then follows by noting that $y^\top L_G y = \sum_{\{u,v\}\in E(G)} (y_u - y_v)^2$.
\end{proof}

For rounding independent sets on spectral expanders, we will only consider $t = 2$ and $3$.
For $t=2$, we get a simple bound that $\w(00) - \w(11) \leq 2\eps$ given that the graph has an independent set of size $(\frac{1}{2}-\eps)n$, i.e., $\E_u[x_u^{(1)}]$ and $\E_u[x_u^{(2)}] \geq \frac{1}{2}-\eps$.
We note that this is the base case of \Cref{lem:w0-leq-w1} for larger $t$.

\begin{lemma}[Special case of \Cref{lem:w0-leq-w1}]
\label{lem:w00-leq-w11}
    Let $\bx = (x^{(1)}, x^{(2)})$.
    \begin{equation*}
        \sststile{2}{\bx} \Set{ \w(00) - \w(11) = 2\eps - \sum_{t\in[2]} \Paren{\E_u[x_u^{(t)}] - \Paren{\frac{1}{2}-\eps}} } \mper
    \end{equation*}
\end{lemma}
\begin{proof}
    First note that $\E_u[x_u^{(1)}] = \w(10) + \w(11)$ and $\E_u[x_u^{(2)}] = \w(01) + \w(11)$.
    Summing up $\E_u[x_u^{(1)}] - (\frac{1}{2}-\eps)$ and $\E_u[x_u^{(2)}] - (\frac{1}{2}-\eps)$ gives $\w(01) + \w(10) + 2 \w(11) - (1-2\eps)$.
    Then, noting that $\w(00) + \w(01) + \w(10) + \w(11) = 1$ completes the proof.
\end{proof}

We next lower bound $\w(00) - \w(11)$ by the expansion of the graph.

\begin{lemma} \label{lem:w00-w11-geq-lambda}
    Let $G$ be a $d$-regular $n$-vertex graph with $\lambda_2 \coloneqq \lambda_2(G) > 0$.
    Let $\bx = (x^{(1)},x^{(2)})$.
    Then,
    \begin{equation*}
        \ISConst_G(\bx) \sststile{4}{\bx}
        \Set{ \w(00) - \w(11) \geq (1-\lambda_2) \cdot \w(\{00,11\})  (1 - \w(\{00,11\})) } \mper
    \end{equation*}
\end{lemma}
\begin{proof}
    Let $S = \{01,10\}$, and define $y_u \coloneqq \1(u\gets S)$.
    By \Cref{lem:laplacian},
    \begin{equation*}
        \ISConst_G(\bx) \sststile{4}{\bx} \Set{ \frac{1}{nd} \cdot y^\top L_G y 
        = e(S, \ol{S})
        = e(00,01) + e(00,10) \leq \w(00) - \w(11) }
        \numberthis \label{eq:cut-leq-w0-w1}
    \end{equation*}
    where $e(S,\ol{S}) = e(00,01) + e(00,10)$ because $\ISConst(\bx) \sststile{2t}{\bx} \Set{ e(01,11) = e(10,11) = 0 }$ (\Cref{fact:e-facts}),
    and the last inequality follows from $\w(00) = \sum_{\alpha\in \zo^2} e(00, \alpha)$ and $\w(11) = e(00,11)$ (again because $e(01,11) = e(10,11) = 0$).

    On the other hand, the trivial eigenvector of $L_G$ is $\vec{1}$ with eigenvalue $0$ while $\lambda_2(\frac{1}{d}L_G) = 1-\lambda_2$, so we have
    \begin{equation*}
        \sststile{2}{y} \Set{ \frac{1}{nd} y^\top L_G y \geq \frac{1}{n} \cdot (1-\lambda_2) \Paren{ \|y\|_2^2 - \frac{1}{n} \iprod{\vec{1}, y}^2 } } \mper
    \end{equation*}
    By the Booleanity constraints, $\BoolConst(y) \sststile{2}{y} \frac{1}{n}(\|y\|_2^2 - \frac{1}{n} \iprod{\vec{1},y}^2) = \E_u[y_u] - \E_u[y_u]^2 = \w(S) (1 - \w(S)) = \w(\ol{S}) (1 - \w(\ol{S}))$, where $\ol{S} = \{00,11\}$.
    Combined with \Cref{eq:cut-leq-w0-w1} finishes the proof.
\end{proof}

Combining \Cref{lem:w00-leq-w11,lem:w00-w11-geq-lambda}, we have that $\E_u[x_u^{(t)}] \geq \frac{1}{2} - \eps$ (i.e., the independent set indicated by $x^{(t)}$ has size at least $(\frac{1}{2}-\eps)n$) together with the expansion of the graph imply that
\begin{equation*}
    (1-\lambda_2) \cdot \w(\{00,11\})  (1 - \w(\{00,11\}))
    \leq \w(00) - \w(11)
    \leq 2\eps \mper
\end{equation*}
When $\lambda_2 \leq 1-C\eps$ for some large enough constant $C$,
then the above implies either $\w(\{00,11\}) \leq \gamma$ or $\w(\{00,11\}) \geq 1-\gamma$ for some small constant $\gamma < \frac{1}{3}$.
In the latter case, since $\w(11) \geq \w(00) - 2\eps$, we have $\w(11) \geq \frac{1}{2} -\frac{\gamma}{2} - \eps$.

Now, we now consider $3$ assignments, where each pair of assignments satisfy the above, i.e., $\w(\{00*,11*\}) (1 - \w(\{00*,11*\})) \leq \frac{2\eps}{1-\lambda_2} \leq \frac{2}{C}$ for all $3$ ``$*$'' locations.
Then, we claim that one of them, say $\w(\{00*,11*\})$, must be $\geq 1-\gamma$.
To see this, notice that the $3$ pairs $\w(\{00*,11*\})$ must sum up to at least $1$ because each $\alpha \in \zo^3$ is covered, i.e., has either two $0$s or two $1$s.
Thus, all $3$ being $\leq \gamma$ leads to a contradiction.

We now formalize this reasoning as an SoS proof.
The following lemma is in fact a special case of \Cref{lem:is-sos-pf} where we conclude a statement for $2t$ assignments using bounds obtained from $t$ assignments.

\begin{lemma} \label{lem:w11-squared}
    Let $G$ be a $d$-regular $n$-vertex graph with $\lambda_2 \coloneqq \lambda_2(G) > 0$, and let $\eps > 0$.
    Let $\bx = (x^{(1)}, x^{(2)}, x^{(3)})$.
    Let $\calA$ be the constraints $\ISConst_G(\bx) \cup \Set{\E_u[x_u^{(t)}] \geq \frac{1}{2} - \eps,\ \forall t\in[3]}$.
    Then,
    \begin{equation*}
        \calA \sststile{6}{\bx}
        \Set{ (\w(11*) + \eps)^2 + (\w(1*1) + \eps)^2 + (\w(*11)+\eps)^2 \geq \frac{1}{4} \Paren{1 - \frac{6\eps}{1-\lambda_2}} } \mper
    \end{equation*}
\end{lemma}
\begin{proof}
    By \Cref{lem:w00-leq-w11}, we have $\calA$ implies that $\w(00*) \leq \w(11*) + 2\eps$.
    Moreover, by \Cref{lem:w00-w11-geq-lambda}, we have 
    \begin{align*}
        2\eps \geq \w(00*) - \w(11*)
        &\geq (1-\lambda_2) \cdot \parens*{ (\w(00*) + \w(11*)) - (\w(00*) + \w(11*))^2 } \\
        &\geq (1-\lambda_2) \cdot \parens*{ (\w(00*) + \w(11*)) - 4(\w(11*) + \eps)^2 } \mper
    \end{align*}

    Next, we sum up the inequalities for all $3$ ``$*$'' locations.
    Observe that $\{00*,11*\} \cup \{0*0,1*1\} \cup \{*00,*11\} = \zo^3$, as any $\alpha\in \zo^3$ must have either 2 zeros or 2 ones.
    This means that the sum of $\w(00*) + \w(11*)$ must be $\geq 1$.
    Thus,
    \begin{equation*}
        \calA \sststile{6}{\bx}
        \Set{ (1-\lambda_2) \cdot \Paren{1 - 4 \Paren{(\w(11*) + \eps)^2 + (\w(1*1) + \eps)^2 + (\w(*11)+\eps)^2}} \leq 6\eps } \mcom
    \end{equation*}
    and rearranging the above completes the proof.
\end{proof}

\subsection{Analysis of \texorpdfstring{\Cref{alg:IS-on-expander}}{Algorithm~\ref{alg:IS-on-expander}}}

We now prove that \Cref{alg:IS-on-expander} successfully outputs an independent set of size $\Omega(n)$.

\begin{lemma} \label{lem:rounding-low-global-correlation}
    Let $\eta,\delta\in (0,1)$ such that $\delta \leq \eta^2/18$, and let $\mu$ be a pseudo-distribution over $\zo^n$ such that $\E_{u,v\in [n]} I_\mu(X_u;X_v) \leq \delta$.
    Suppose $\pE_{\mu^{\otimes 2}}[\w(11)^2] \geq \frac{1}{16} + \eta$, then the set of vertices $u$ such that $\pPr_{\mu}[x_u=1] > \frac{1}{2}$ forms an independent set of size $\eta n/4$.
\end{lemma}
\begin{proof}
    Recall from \Cref{def:w} that $\w(11) = \E_{u\sim [n]}[x_u^{(1)} x_u^{(2)}]$, thus
    \begin{align*}
        \pE_{\mu^{\otimes 2}}[\w(11)^2]
        &= \pE_{\mu^{\otimes 2}} \E_{u,v\sim [n]} \bracks*{x_u^{(1)} x_u^{(2)} x_v^{(1)} x_v^{(2)}} \\
        &= \E_{u,v\sim [n]} \bracks*{ \pE_{\mu}[x_u x_v]^2 }
        = \E_{u,v\sim [n]} \bracks*{ \pPr_{\mu}[x_u=1,\ x_v=1]^2 }
        \mper
    \end{align*}
    
    Now, given that $\mu$ has small average correlation, by Pinsker's inequality (\Cref{fact:pinskers}),
    \begin{equation*}
        \E_{u,v \sim [n]} \Abs{\pPr_{\mu}[x_{u} = 1,\ x_v=1] - \pPr_{\mu}[x_{u} = 1] \pPr_{\mu}[x_v=1]} \leq \sqrt{\delta/2} \mper
    \end{equation*}
    Then, using the fact that $p^2 = q^2 + 2q(p-q) + (p-q)^2 \leq q^2 + 3|p-q|$ for all $p,q \in [0,1]$, we have
    \begin{align*}
        \E_{u,v\sim [n]} \bracks*{ \pPr_{\mu}[x_u=1,\ x_v=1]^2 }
        &\leq \E_{u,v\sim [n]} \bracks*{ \pPr_{\mu}[x_u=1]^2 \pPr_{\mu}[x_v=1]^2 } + 3\sqrt{\delta/2} \\
        &\leq \E_{u\sim [n]} \bracks*{ \pPr_{\mu}[x_u=1]^2 }^2 + 3\sqrt{\delta/2}
        \mper
    \end{align*}
    Thus, since $\pE_{\mu^{\otimes 2}}[\w(11)^2] \geq \frac{1}{16} + \eta$ and $\delta \leq \eta^2 / 18$, we have
    $\E_{u\sim [n]} \bracks*{ \pPr_{\mu}[x_u=1]^2 }^2 \geq \frac{1}{16} + \frac{\eta}{2}$, which means that $\E_{u\sim [n]} \bracks*{ \pPr_{\mu}[x_u=1]^2 } \geq \sqrt{\frac{1}{16} + \frac{\eta}{2}} \geq \frac{1}{4} + \frac{\eta}{2}$.
    It follows that at least $\eta/4$ fraction of vertices have $\pPr_{\mu}[x_u=1] > \frac{1}{2}$.
    By \Cref{fact:obvious-is}, these vertices form an independent set.
\end{proof}

\begin{proof}[Proof of \Cref{thm:IS-expander-main}]
    By the assumption that $G$ contains an independent set of size $(\frac{1}{2}-\eps)n$, the pseudo-distribution $\mu$ satisfies the constraint $\E_u[x_u] \geq \frac{1}{2} - \eps$.
    Let $\bx = (x^{(1)}, x^{(2)}, x^{(3)}) \sim \mu^{\otimes 3}$, then \Cref{lem:w11-squared} states that
    \begin{equation*}
        \pE_{\mu^{\otimes 3}} \bracks*{ (\w(11*) + \eps)^2 + (\w(1*1) + \eps)^2 + (\w(*11)+\eps)^2  }
        \geq \frac{1}{4} \Paren{1 - \frac{6\eps}{1-\lambda_2}} \mper
    \end{equation*}
    
    By symmetry, the $3$ terms on the left-hand side are equal, and
    \begin{equation*}
        \pE_{\mu^{\otimes 3}}[(\w(11*) +\eps)^2]
        = \pE_{\mu^{\otimes 2}}[\w(11)^2 + 2\eps \cdot \w(11) + \eps^2]
        \leq \pE_{\mu^{\otimes 2}}[\w(11)^2] + 2\eps + \eps^2 \mper
    \end{equation*}
    Thus, if $\eps \leq 0.001$ and $\lambda_2 \leq 1-C\eps$ with $C = 40$, then we have
    $\pE_{\mu^{\otimes 2}}[\w(11)^2] \geq \frac{1}{12}(1-\frac{6}{C}) - (2\eps+\eps^2) \geq \frac{1}{15} > \frac{1}{16}$.

    By \cref{lem:ragh-tan}, after we condition $\mu'$ on the values of $O(1/\delta)$ variables as done in Step (2) of \cref{alg:IS-on-expander} to get $\mu$, we have $\E_{u,v\in [n]}[I_{\mu}(X_u; X_v)] \leq \delta$, where $\delta$ is a small enough constant. Then, by \Cref{lem:rounding-low-global-correlation}, at least $\frac{1}{4}(\frac{1}{15}-\frac{1}{16}) \geq \frac{1}{1000}$ fraction of the vertices have $\pPr_{\mu}[x_u=1] > \frac{1}{2}$. By \Cref{fact:obvious-is}, this must be an independent set, thus completing the proof. 
\end{proof}

%% file: 3-colorable.tex
\section{Independent Sets on Almost 3-colorable Spectral Expanders} \label{sec:3-coloring-main}

Recall that an $\eps$-almost $3$-colorable graph is a graph which is $3$-colorable if one removes $\eps$ fraction of the vertices.

\begin{theorem}[Restatement of \Cref{thm:3-colorable-main}] \label{thm:3-colorable}
    For any $\eps \in [0,10^{-4}]$,
    let $G$ be an $n$-vertex regular $\eps$-almost $3$-colorable graph with $\lambda_2(G) \leq 10^{-4}$.
    Then, there is an algorithm that runs in $\poly(n)$ time and outputs an independent set of size at least $10^{-4} n$.
\end{theorem}

\begin{mdframed}
    \begin{algorithm}[Find independent set in a 3-colorable expander]
    \label{alg:coloring-expander}
    \mbox{}
      \begin{description}
      \item[Input:] A graph $G = (V,E)$.

      \item[Output:] An independent set of $G$.

      \item[Operation:] Fix $\gamma = 10^{-3}$ and $\eps = 10^{-4}$.
        \begin{enumerate}
            \item Run the polynomial-time algorithm from \Cref{fact:obvious-algo} and exit if that outputs an independent set of size at least $\gamma n$.
            \item Solve the degree-$d$ SoS algorithm to obtain a pseudo-distribution $\mu'$ that satisfies the almost $3$-coloring constraints and the constraints $\E_u[\1(x_u = \sigma)] \leq \frac{1}{2}+\gamma$ for all $\sigma \in [3]$ and $\E_u[\1(x_u = \bot)] \leq \eps$.
            \item Choose a uniformly random set of $t = O(1)$ vertices $i_1,\ldots,i_t \sim [n]$ and draw \\
            $(\sigma_{i_1},\ldots,\sigma_{i_t}) \sim \mu'$. Let $\mu$ be the pseudo-distribution obtained by conditioning $\mu'$ on $(x_{i_1} = \sigma_{i_1},\ldots, x_{i_t}=\sigma_{i_t})$.
            \item For each $\sigma \in [3]$, let $I_{\sigma} = \{u\in V: \pE_{\mu}[\1(x_u = \sigma)] > \frac{1}{2}\}$.
            Output the largest one.
        \end{enumerate}
      \end{description}
    \end{algorithm}
\end{mdframed}

\subsection{Almost 3-coloring Formulation and Agreement}

We define an almost $3$-coloring of a graph to be an assignment of vertices to $\{1,2,3,\bot\}$ where $\{1,2,3\}$ are the color classes and the fraction of vertices assigned to $\bot$ is small.

\begin{definition}[Almost $3$-coloring constraints] \label{def:coloring-formulation}
    Denote $\Sigma \coloneqq [3] \cup \{\bot\}$.
    Given a graph $G = (V, E)$ and parameter $\eps \geq 0$, let $\ol{x} = \{\ol{x}_{u,\sigma}\}_{u\in V,\sigma\in\Sigma}$ be indeterminants.
    We define the almost $3$-coloring constraints as follows:
    \begin{align*}
        \ColorConst_{G}(\ol{x}) \coloneqq \BoolConst(\ol{x}) & \cup 
        \Set{ \sum_{\sigma\in\Sigma} \ol{x}_{u,\sigma} = 1,\ \forall u\in V}
        \cup \Set{ \ol{x}_{u,\sigma} \ol{x}_{v,\sigma} = 0,\ \forall \{u,v\}\in E,\ \sigma\in [3] } \mper
    \end{align*}
    Moreover, with slight abuse of notation, for $t\in \N$ and assignments $\ol{x}^{(1)},\ol{x}^{(2)},\dots,\ol{x}^{(t)}$,
    \begin{equation*}
        \ColorConst_{G}(\ol{\bx}) \coloneqq \bigcup_{i\in[t]} \ColorConst_{G}(\ol{x}^{(i)}) \mper
    \end{equation*}
    We will drop the dependence on $G$ when it is clear from context.
\end{definition}

\parhead{Notation.} We remark that there is a one-to-one correspondence between almost $3$-coloring assignments $x \in \{1,2,3,\bot\}^n$ and $\ol{x} \in \zo^{n\times 4}$.
Even though formally the SoS program is over variables $\ol{x}$,
from here on we will use the notation $x\in \{1,2,3,\bot\}^n$ as it is equivalent and more intuitive.
For example, we will write $\1(x_u = \sigma)$ to mean $\ol{x}_{u,\sigma}$, and similarly $\pPr_{\mu}[x_u = \sigma] = \pE_{\mu}[\ol{x}_{u,\sigma}]$.

The following definition is almost identical to \Cref{def:w}.

\begin{definition} \label{def:w-coloring}
    Let $t\in \N$, and let $\bx = (x^{(1)}, x^{(2)},\dots, x^{(t)})$.
    For each $\alpha \in \Sigma^t$, we define the following multilinear polynomials,
    \begin{gather*}
        \1(u\gets \alpha) \coloneqq \prod_{i \in [t]} \1(x^{(i)}_u = \alpha_i) \mcom
        \quad \text{for each $u\in [n]$} \mcom\\
        \w(\alpha) \coloneqq \E_{u\in[n]}[\1(u \gets \alpha)] \mper
    \end{gather*}
    For convenience, we omit the dependence on $\bx$.
    
    For $S \subseteq \Sigma^t$, we denote $\1(u\gets S) \coloneqq \sum_{\alpha\in S} \1(u\gets \alpha)$ and $\w(S) \coloneqq \sum_{\alpha \in S} \w(\alpha)$.
    Moreover, we will denote $S_{\bot} \coloneqq \{\alpha \in \Sigma^t: \exists i\in [t],\ \alpha_i = \bot\}$.
\end{definition}

As explained in \Cref{sec:3-coloring-overview}, due to the symmetry of the color classes, we need to define the relative \emph{agreement} between two valid almost 3-colorings according to some permutation $\pi\in \bbS_3$.
For example, consider a coloring $x \in \Sigma^n$ and suppose $y \in \Sigma^n$ is obtained by permuting the $3$ color classes of $x$.
The agreement between $x$ and $y$ should be close to $1$.
Thus, we define the agreement between $x$ and $y$ as
\begin{equation*}
    \max_{\pi\in \bbS_3}\ \E_{u \in V}\bracks*{\pi(x_u) = y_u \neq \bot} \mper
\end{equation*}
Here for simplicity we assume $\pi(\bot) = \bot$.
Formally,
\begin{definition}[Agreement between $2$ valid $3$-colorings]  \label{def:agreement}
    Let $\pi \in \bbS_3$.
    Define
    \begin{equation*}
        S_{\pi} \coloneqq \{(\sigma, \pi(\sigma)) : \sigma\in [3]\} \mper
    \end{equation*}
    For almost 3-colorings $x, y\in \Sigma^n$, we define the agreement between $x$ and $y$ according to permutation $\pi$ to be
    \begin{equation*}
        \agree_{\pi}(x,y)
        \coloneqq \w(S_{\pi})
        = \E_{u\in[n]} \bracks*{ \sum_{\sigma\in [3]} \1(x_u = \sigma,\ y_u = \pi(\sigma))  } \mper
    \end{equation*}
    Furthermore, for any $\ell \in \N$, we write
    \begin{equation*}
        \agree^{(\ell)}(x,y) = \sum_{\pi\in \bbS_{3}} \agree_{\pi}(x,y)^{\ell} \mper
    \end{equation*}
\end{definition}

Here $\agree^{(\ell)}(x,y)$ should be viewed as a polynomial approximation of $\max_{\pi} \agree_{\pi}(x,y)^{\ell}$.

We note some simple facts (written in SoS form) that will be useful later.

\begin{fact} \label{fact:u-alpha-3color}
For any $t\in \N$, the following can be easily verified:
\begin{enumerate}[(1)]
    \item $\BoolConst(\bx) \sststile{2t}{\bx} \Set{ \1(u\gets \alpha)^2 = \1(u\gets \alpha) }$, i.e., $\1(u\gets \alpha)$ satisfies the Booleanity constraint.

    \item $\ColorConst(\bx) \sststile{2t}{\bx} \Set{ \1(u\gets \alpha) \cdot \1(u\gets \beta) = 0 }$ for $\alpha \neq \beta$.
    This also implies that $\1(u\gets S)$ satisfies the Booleanity constraint for any $S \subseteq \Sigma^t$.

    \item $\ColorConst(\bx) \sststile{t}{\bx} \Set{ \sum_{\alpha\in \Sigma^t} \1(u\gets \alpha) = 1 }$, thus $\ColorConst(\bx) \sststile{t}{\bx} \Set{ \sum_{\alpha\in \Sigma^t} \w(\alpha) = 1 }$.
\end{enumerate}    
\end{fact}

Each $S_{\pi}$ corresponds to a triangle in \Cref{fig:triangle-gadget}, and we see that there are two ways to partition the graph into $3$ disjoint triangles.
The next lemma can essentially be proved by looking at \Cref{fig:triangle-gadget} (there $S_\bot$ is not shown), and it is crucial for our analysis.

\begin{lemma} \label{lem:sum-of-S-pi}
    Let $G$ be a regular graph.
    Let $\bbS_3^+$ be the set of $3$ permutations with sign (a.k.a.\ parity) $+1$ and $\bbS_3^-$ be the ones with sign $-1$.
    Then,
    \begin{equation*}
        \ColorConst_G(x,y) \sststile{2}{x,y}
        \Set{\sum_{\pi\in \bbS_3^+} \w(S_{\pi}) = \sum_{\pi\in \bbS_3^-} \w(S_{\pi}) = 1 - \w(S_{\bot})} \mper
    \end{equation*}
    Moreover,
    \begin{equation*}
        \ColorConst_G(x,y) \sststile{2}{x,y}
        \Set{ \sum_{\pi\in \bbS_3} e(S_{\pi}, \ol{S}_{\pi}) \leq 1 } \mper
    \end{equation*}
\end{lemma}
\begin{proof}
    The first statement follows by noting that for each $i, j\in [3]$, there are exactly two permutations  with opposite signs that map $i$ to $j$.
    Thus, $\{S_{\pi}: \pi \in \bbS_3^+\} \cup \{S_{\bot}\}$ and $\{S_{\pi}: \pi \in \bbS_3^-\} \cup \{S_{\bot}\}$ are partitions of the whole graph.
    One can also prove this directly from \Cref{fig:triangle-gadget}.

    For the second statement, note that each edge $(i_1,j_1), (i_2,j_2) \in [3]^2$ in the gadget uniquely identifies the permutation $\pi$ such that $\pi(i_1) = j_1$ and $\pi(i_2) = j_2$.
    This means that each edge not incident to $S_{\pi}$ is contained in exactly one $S_{\pi}$, and we have $\sum_{\pi} e(S_{\pi}, S_{\pi}) = e(\ol{S}_\bot, \ol{S}_\bot) \geq 1 - 2\w(S_\bot)$.
    On the other hand, from the first statement we have $\sum_\pi \w(S_\pi) = 2 - 2\w(S_\bot)$
    Thus, $\sum_{\pi} e(S_{\pi}, \ol{S}_\pi) = \sum_{\pi} (\w(S_{\pi}) - e(S_\pi, S_\pi)) \leq (2 - 2\w(S_\bot)) - (1 - 2\w(S_\bot)) = 1$.
\end{proof}

\subsection{Large Spectral Gap implies Large Agreement}

\begin{lemma} \label{lem:sum-of-S-pi-squared}
    Let $G$ be a $d$-regular $n$-vertex graph with $\lambda_2 \coloneqq \lambda_2(G) > 0$.
    Then,
    \begin{equation*}
        \ColorConst_G(x,y) \sststile{4}{x,y} 
        \Set{ \sum_{\pi \in \bbS_3} \w(S_{\pi})^2 \geq 2 - \frac{1}{1-\lambda_2} - 2\w(S_\bot) } \mper
    \end{equation*}
\end{lemma}
\begin{proof}
    Fix a permutation $\pi\in \bbS_3$, and let $y_u = \1(u\gets S_{\pi})$.
    By \Cref{lem:laplacian}, we have that
    \begin{equation*}
        \ColorConst_G(x,y) \sststile{4}{x,y}
        \Set{ e(S_{\pi}, \ol{S}_{\pi}) = \frac{1}{nd} y^\top L_G y
        \geq \frac{1}{n} \cdot (1-\lambda_2) \parens*{\|y\|_2^2 - \frac{1}{n} \angles{\vec{1},y}^2 } }
    \end{equation*}
    Since $y_u$ satisfies the booleanity constraints, we have $\frac{1}{n} (\|y\|_2^2 - \frac{1}{n} \angles{\vec{1},y}^2) = \E_u[y_u] - \E_u[y_u]^2 = \w(S_{\pi}) (1 - \w(S_{\pi}))$.
    Thus,
    \begin{equation*}
        \ColorConst_G(x,y) \sststile{4}{x,y}
        \Set{ e(S_{\pi}, \ol{S}_{\pi}) \geq (1-\lambda_2) \cdot \w(S_\pi) (1 - \w(S_\pi)) } \mper
    \end{equation*}

    Next, we sum over $\pi \in \bbS_3$.
    By \Cref{lem:sum-of-S-pi}, on the left-hand side we have $\sum_{\pi} e(S_\pi, \ol{S}_\pi) \leq 1$, and on the right-hand side we have $(1-\lambda_2) \sum_{\pi} \w(S_{\pi})(1-\w(S_\pi)) = (1-\lambda_2) (2 - 2\w(S_\bot) - \sum_{\pi} \w(S_\pi)^2)$.
    Rearranging this completes the proof.
\end{proof}

    

In \Cref{thm:3-colorable}, we assume that the graph has spectral gap $1-\lambda_2 \geq 1-\gamma$ and the almost $3$-coloring assignments satisfy $\w(S_\bot) \leq \w(\{\bot*\}) + \w(\{*\bot\}) \leq 2\eps$ for some small enough constants $\eps,\gamma$.
Thus, by \Cref{lem:sum-of-S-pi,lem:sum-of-S-pi-squared}, the $6$ variables $\{\w(S_\pi)\}_{\pi\in \bbS_3}$ satisfy that $\sum_{\pi\in \bbS_3^+} \w(S_\pi) = \sum_{\pi\in \bbS_3^-} \w(S_\pi) \in [1 - 2\eps, 1]$ and $\sum_{\pi} \w(S_\pi)^2 \geq 1 - O(\gamma+\eps)$.
On the other hand, recall from \Cref{def:agreement} that $\w(S_\pi) = \agree_\pi(x,y)$.

We would like to prove \Cref{claim:almost-bipartite}: assuming $\agree_\pi(x,y) \leq \frac{1}{2}+\gamma$ for all $\pi$, then one of $\{\w(S_\pi)\}_{\pi\in \bbS_3^+}$ and one of $\{\w(S_\pi)\}_{\pi \in \bbS_3^-}$ must be small.
This is captured in the following lemma:

\begin{lemma} \label{lem:6-variable-lemma}
    Fix $\gamma \in [0, 0.01]$.
    Let $z_1,z_2,\dots,z_6$ be such that $0 \leq z_i \leq \frac{1}{2} + \gamma$ and $z_1+z_2+z_3 = z_4+z_5+z_6 \leq 1$.
    Suppose $\|z\|_2^2 \geq 1 - \gamma$.
    Then, one of $z_1,z_2,z_3$ and one of $z_4,z_5,z_6$ must be $\leq 8\gamma$.
\end{lemma}
\begin{proof}
    For any $i\in [6]$, we have $\|z\|_2^2 \leq z_i^2 + (\frac{1}{2} +\gamma) \sum_{j\neq i} z_j$ since $z_j \leq \frac{1}{2} + \gamma$ for all $j$.
    Then since $\|z\|_1 \leq 2$, for all $i \in [6]$ we have
    \begin{equation*}
        \|z\|_2^2 \leq z_i^2 + \parens*{\frac{1}{2}+\gamma }(2-z_i) = 1+2\gamma - z_i \parens*{\frac{1}{2}+\gamma-z_i} \mper
    \end{equation*}
    Since $\|z\|_2^2 \geq 1-\gamma$, it follows that
    \begin{equation*}
        z_i \parens*{\frac{1}{2}+\gamma-z_i} 
        \leq 3\gamma  \mcom \quad \forall i\in [6] \mper
    \end{equation*}
    Then, by solving a quadratic inequality, one can verify that when $\gamma \leq 0.01$, the above implies that either $z_i \leq 8\gamma$ or $z_i \geq \frac{1}{2} - 8\gamma$.
    Therefore, since $z_1 + z_2 + z_3 \leq 1$, $z_1,z_2,z_3$ cannot all be the latter, i.e., one of them must be $\leq 8\gamma$. Similarly for $z_4,z_5,z_6$.
\end{proof}

We next consider $3$ almost $3$-coloring assignments.
Recall that $\Sigma = [3] \cup \{\bot\}$.

\begin{lemma} \label{lem:27-variable-lemma}
    Let $0 \leq \eps,\gamma \leq 0.001$.
    Let $\{w(\alpha)\}_{\alpha\in \Sigma^3}$ be variables such that $0 \leq w(\alpha) \leq 1$ and $\sum_{\alpha} w(\alpha) = 1$.
    For any $S \subseteq \Sigma^3$, denote $w(S) = \sum_{\alpha\in S} w(\alpha)$, and let
    \begin{gather*}
        S_{\pi}^{(12)} = \Set{(\sigma,\pi(\sigma),*): \sigma\in[3]} \mcom \\
        S_{\pi}^{(13)} = \Set{(\sigma,*,\pi(\sigma)): \sigma\in[3]} \mcom \\
        S_{\pi}^{(23)} = \Set{(*,\sigma,\pi(\sigma)): \sigma\in[3]} \mcom
    \end{gather*}
    Suppose $w(\sigma**), w(*\sigma*), w(**\sigma) \leq \frac{1}{2} + \gamma$ and $w(\bot**), w(*\bot*), w(**\bot) \leq \eps$.
    Moreover, suppose $\sum_{\pi\in \bbS_3} w(S_\pi^{(ij)})^2 \geq 1 - \gamma$ for all pairs $i< j \in [3]$, then there must be some $\pi$ and $i<j$ such that $w(S_{\pi}^{(ij)}) \geq \frac{1}{2}+ \gamma$.
\end{lemma}
\begin{proof}
    Suppose by contradiction that all $w(S_{\pi}^{(ij)}) \leq \frac{1}{2} + \gamma$.
    Let $\bbS_3^+$ be the set of $3$ permutations with sign (a.k.a.\ parity) $+1$ and $\bbS_3^-$ be the ones with sign $-1$.
    For each pair $i<j$ (say, $(12)$ for now), by \Cref{lem:sum-of-S-pi} we have $\sum_{\pi\in \bbS_3^+} w(S_{\pi}^{(12)}) = \sum_{\pi\in \bbS_3^-} w(S_{\pi}^{(12)}) \leq 1$.
    
    Therefore, the $6$ variables $\{w(S_{\pi}^{(12)})\}_{\pi\in \bbS_3^+} \cup \{w(S_{\pi}^{(12)})\}_{\pi\in \bbS_3^-}$ satisfy the conditions in \Cref{lem:6-variable-lemma}, and thus there are some $\pi^+ \in \bbS_3^+$ and $\pi^- \in \bbS_3^-$ such that $w(S_{\pi^{+}}^{(12)}), w(S_{\pi^-}^{(12)}) \leq 8\gamma$.
    Furthermore, note that since $\pi^+$ and $\pi^-$ have different signs, $S_{\pi^+}^{(12)}$ and $S_{\pi^-}^{(12)}$ intersect in exactly $(\beta_1,\beta_2,*)$ for some $\beta_1,\beta_2 \in [3]$.
    In fact, $\beta_1,\beta_2$ uniquely determine $\pi^+$ and $\pi^-$, as there are exactly two permutations with different signs that map $\beta_1$ to $\beta_2$.

    Assume without loss of generality (due to symmetry) that $\beta_1 = \beta_2 = 1$, thus we have $S_{\pi^+}^{(12)} = \{11*,22*,33*\}$ and $S_{\pi^-}^{(12)} = \{11*,23*,32*\}$.
    Let $T^{(12)} \coloneqq [3]^3 \setminus (S_{\pi^+}^{(12)} \cup S_{\pi^-}^{(12)}) = \{12*,13*,21*,31*\}$, which equals $\{1**,*1*\} \setminus \{11*\}$ (here we do not include $\bot$).
    Notice the structure of $T^{(12)}$ --- ignoring the third assignment, $T^{(12)}$ forms a $2\times 2$ bipartite graph (between $\{12*,13*\}$ and $\{21*,31*\}$ in this case; see \Cref{fig:almost-bipartite}) where one assignment labels the entire left-hand side as one color while the other assignment labels the entire right-hand side as one color.

    Now, for all 3 pairs $(12),(13),(23)$, consider $T \coloneqq T^{(12)} \cap T^{(23)} \cap T^{(13)} \subseteq [3]^3$.
    First, we have $w(T) \geq 1 - 48\gamma - \w(S_\bot) \geq 1 - 48\gamma - 3\eps$, since $\w(S_\bot) \leq 3\eps$ by assumption.
    Next, we claim that for all choices of $\pi^+$ and $\pi^-$ for each pair, $T$ can contain at most $4$ strings in $[3]^3$ and must form a $2 \times 2$ bipartite structure such that each assignment colors one side with one color.

    Let $T^{(12)} = \{a_1**, *a_2*\}\setminus \{a_1a_2*\}$, $T^{(13)} = \{b_1**, **b_2\}\setminus \{b_1*b_2\}$, and $T^{(23)} = \{*c_1*, **c_2\}\setminus \{*c_1c_2\}$ for some $a_1,a_2, b_1,b_2,c_1,c_2 \in [3]$.
    We split into several cases:
    \begin{itemize}
        \item $a_1 = b_1$: in this case, $T^{(12)} \cap T^{(13)} = (\{a_1**\} \setminus \{a_1a_2*, a_1*b_2\}) \cup (\{*a_2b_2\}\setminus \{a_1a_2b_2\})$.
        \begin{enumerate}
            \item $c_1 \neq a_2$, $c_2 \neq b_2$: then, $T = (\{a_1c_1*\} \setminus \{a_1 c_1 b_2, a_1 c_1 c_2\}) \cup (\{a_1*b_2\} \setminus \{a_1 a_2 b_2, a_1 c_1 b_2\})$, i.e., $2$ strings in $[3]^3$.
            For example, $T = \{123,131\}$.

            \item $c_1 = a_2$, $c_2 \neq b_2$: then, $T = (\{a_1*c_2\} \setminus \{a_1a_2c_2\}) \cup (\{*a_2b_2\} \setminus \{a_1a_2b_2\})$, i.e., $4$ strings in $[3]^3$.
            For example, $T = \{122,132,211,311\}$.

            \item $c_1 = a_2$, $c_2 = b_2$: then, $T = \varnothing$.
        \end{enumerate}

        \item $a_1 \neq b_1$: in this case, $T^{(12)} \cap T^{(13)} = (\{a_1*b_2\} \setminus \{a_1a_2b_2\}) \cup (\{b_1a_2*\} \setminus \{b_1a_2b_2\})$, which is already the same case as the second case above.
    \end{itemize}

    For the case when $T = \varnothing$ or $T$ contains $2$ strings, we have $w(T) \leq w(\sigma**)$ for some $\sigma\in[3]$, which means $1 - 48\gamma - 3\eps \leq \frac{1}{2} + \gamma$.
    This is a contradiction.

    For the case when $T$ contains $4$ strings, let $T = \{\alpha^1,\alpha^2,\beta^1,\beta^2\}$ such that $\{\alpha^1,\alpha^2\}$ and $\{\beta^1,\beta^2\}$ form the bipartite structure.
    Assume without loss of generality that the first assignment labels the left with the same color: $\alpha^1_1 = \alpha^2_1 \neq \beta^1_1,\beta^2_1$, 
    and the second and third label the right with the same color: $\beta^1_2 = \beta^2_2 \neq \alpha^1_2,\alpha^2_2$ and $\beta^1_3 = \beta^2_3 \neq \alpha^1_3,\alpha^2_3$.
    Observe that $w(\alpha^1) + w(\alpha^2) \leq w(\alpha^1_1**) \leq \frac{1}{2}+\gamma$ and
    $w(\beta^1) + w(\beta^2) \leq w(*\beta_2^1*) \leq \frac{1}{2}+\gamma$ by the assumptions.
    Since $w(T) \geq 1 - 48\gamma - 3\eps$, it follows that $w(\alpha^1) + w(\alpha^2)$ and $w(\beta^1)+ w(\beta^2) \geq \frac{1}{2}-49\gamma - 3\eps$.
    
    On the other hand, $w(\alpha^1) + w(\beta^1) + w(\beta^2) \leq w(*\alpha^1_2 \alpha^1_3) + w(*\beta^1_2 \beta^1_3) \leq w(S^{(23)}_{\pi})$ and $w(\alpha^2) + w(\beta^1) + w(\beta^2) \leq w(*\alpha^2_2 \alpha^2_3) + w(*\beta^1_2 \beta^1_3) \leq w(S^{(23)}_{\pi'})$ for some permutations $\pi,\pi'\in \bbS_3$.
    However, this means that one of $w(S^{(23)}_{\pi})$, $w(S^{(23)}_{\pi'})$ is at least $\frac{3}{2}(\frac{1}{2}-49\gamma - 3\eps) > \frac{1}{2}+\gamma$ when $\eps, \gamma \leq 0.001$, which is a contradiction.
\end{proof}

We next formalize \Cref{lem:27-variable-lemma} as an SoS proof.

\begin{lemma}[SoS version of \Cref{lem:27-variable-lemma}]
\label{lem:27-variable-lemma-sos}
    Fix constants $\eps, \gamma \in (0,0.001]$ and $\ell \in \N$.
    Let $S_{\pi}^{(ij)} \subseteq [3]^3$ be as defined in \Cref{lem:27-variable-lemma}, and let $\{w(\alpha)\}_{\alpha \in \Sigma^3}$ be indeterminants.
    Let $\calA$ be the set of constraints including 
    \begin{enumerate}[(1)]
        \item $0 \leq w(\alpha) \leq 1$,
        \item $\sum_{\alpha\in \Sigma^3} w(\alpha) = 1$,
        \item  $w(\sigma**), w(*\sigma*), w(**\sigma) \leq \frac{1}{2}+\gamma$ for all $\sigma\in[3]$,
        \item  $w(\bot**), w(*\bot*), w(**\bot) \leq \eps$,
        \item $\sum_{\pi\in\bbS_3} w(S_{\pi}^{(ij)})^2 \geq 1 - \gamma$ for all pairs $i<j\in[3]$.
    \end{enumerate}
    Then, there exists an integer $d = d(\eps,\gamma,\ell)$ such that
    \begin{equation*}
        \calA \sststile{d}{\{w(\alpha)\}}
        \Set{ \sum_{i<j\in[3]} \sum_{\pi\in \bbS_3} w \parens*{S^{(ij)}_{\pi}}^{\ell} \geq \parens*{\frac{1+\gamma}{2}}^{\ell} } \mper
    \end{equation*}
\end{lemma}
\begin{proof}
    \Cref{lem:27-variable-lemma} shows that assuming constraints $\calA$, there must be some $w(S_{\pi}^{(ij)}) \geq \frac{1}{2}+\gamma$.
    This immediately implies that $\sum_{i<j\in[3]} \sum_{\pi\in \bbS_3} w \parens{S^{(ij)}_{\pi}}^{\ell} \geq (\frac{1}{2}+\gamma)^{\ell}$.

    Define $f(w) \coloneqq \sum_{i<j\in[3]} \sum_{\pi\in \bbS_3} w \parens{S^{(ij)}_{\pi}}^{\ell} - (\frac{1+\gamma}{2})^{\ell}$, a degree-$\ell$ polynomial in $64$ variables with bounded coefficients.
    Note that $\calA$ defines a subset $A\subseteq \R^n$ which is compact, and $\min_{w\in A} f(w) \geq \theta$ for some constant $\theta = \theta(\gamma,\ell) > 0$.
    Thus, by the Positivstellensatz (\Cref{fact:black-box-SoS}), $f(w) \geq 0$ has an SoS proof of degree $d$ depending on $\eps,\gamma,\ell$.
\end{proof}

\subsection{Rounding with Large Agreement}

We prove the following key lemma that large agreement and small correlation imply rounding.
Using this, we finish the proof of \Cref{thm:3-colorable} at the end of this section.

\begin{lemma}[Rounding with large agreement] \label{lem:round-3-col}
Fix $\gamma \in (0,1)$.
There exist $\ell \in \N$ and $\delta \in (0,1)$ such that
given a degree-$\ell$ pseudo-distribution $\mu$ satisfying the almost 3-coloring constraints such that
\begin{equation*}
    \pE_{(x,y)\sim \mu^{\otimes 2}}\bracks*{\agree^{(\ell)}(x,y)} \geq \left(\frac{1}{2}+\gamma\right)^{\ell} \mcom
\end{equation*}
and suppose $\mu$ is almost $\ell$-wise independent on average:
\begin{equation*}
    \E_{u_1,\dots,u_{\ell}\in[n]} \KL(\mu(X_{u_1},\dots,X_{u_\ell}) \| \mu(X_{u_1}) \times \cdots \times \mu(X_{u_\ell})) \leq \delta \mcom
\end{equation*}
then one of the sets $I_{\sigma} = \{u\in V: \pPr_{\mu}[x_u=\sigma] > \frac{1}{2}\}$ for $\sigma\in [3]$ has size at least $\Omega(\gamma n)$.
\end{lemma}

The proof of \Cref{lem:round-3-col} relies on the following definition.
\begin{definition}[Collision probability]
    Given a pseudo-distribution $\mu$ over $\Sigma^n$, we define the collision probability of a vertex $u \in [n]$ to be
    \begin{equation*}
        \CP_{\mu}(x_u) \coloneqq \pE_{x,x'\sim\mu}[\1(x_u = x_u'\neq \bot)]
        = \sum_{\sigma\in [3]} \pPr_{\mu}[x_u = \sigma]^2 \mper
    \end{equation*}
    Further, the (average) collision probability $\CP(\mu) = \E_{u\in[n]} \CP(x_u)$.
\end{definition}

We next show a simple lemma which states that large collision probability implies a large fraction of vertices with $\pPr_\mu[x_u=\sigma] > \frac{1}{2}$ for some color $\sigma \in [3]$ (and they form an independent set due to \Cref{fact:obvious-is}).

\begin{lemma} \label{lem:large-CP}
    Suppose a pseudo-distribution $\mu$ over $\Sigma^n$ has collision probability $\CP(\mu) \geq \frac{1}{2} + \gamma$ for some $\gamma \in (0,1/2]$, then there is a $\sigma\in [3]$ such that at least $\gamma/3$ fraction of $u\in[n]$ have $\pPr_{\mu}[x_u = \sigma] \geq \frac{1}{2} + \frac{\gamma}{2}$.
\end{lemma}
\begin{proof}
    Observe that $\CP_{\mu}(x_u) \leq \max_{\sigma\in [3]} \pPr_{\mu}[x_u=\sigma]$ because $\sum_{\sigma\in [3]} \pPr_{\mu}[x_u = \sigma] \leq 1$.
    Thus, we have $\E_{u\in [n]} \max_{\sigma\in[3]} \pPr_{\mu}[x_u = \sigma] \geq \frac{1}{2}+\gamma$.
    This implies that at least $\gamma$ fraction of $u\in [n]$ has $\max_{\sigma\in[3]} \pPr_{\mu}[x_u = \sigma] \geq \frac{1}{2}+ \frac{\gamma}{2}$.
    Then, there must be a $\sigma \in [3]$ such that at least $\gamma / 3$ fraction of $u\in[n]$ have $\pPr_{\mu}[x_u = \sigma] \geq \frac{1}{2} + \frac{\gamma}{2}$.
\end{proof}

In light of \Cref{lem:large-CP}, to prove \Cref{lem:round-3-col}, it suffices to show that the pseudo-distribution $\mu$ has large collision probability.

\begin{proof}[Proof of \Cref{lem:round-3-col}]
    We first prove an upper bound on $\pE[\agree^{(\ell)}(x,y)]$:
    \begin{equation}\label{eq:ub-agr}
        \pE_{x,y\sim \mu} \Brac{\agree^{(\ell)}(x, y)} \leq 6 \Paren{\CP(\mu)^\ell + 2\sqrt{2\delta}} \mper
    \end{equation}
    For any permutation $\pi$, recalling \Cref{def:agreement},
    \begin{align*}
        \agree_{\pi}(x,y)^\ell
        &= \Pr_{u_1,\dots,u_\ell\in [n]} \Brac{x_{u_i} = \pi(y_{u_i}) \neq \bot,\ \forall i\in[\ell]} \\
        &= \E_{u_1,\dots, u_\ell\in [n]} \sum_{\sigma_1,\dots,\sigma_\ell \in [3]} \1\Paren{x_{u_i} = \pi(y_{u_i}) = \sigma_i,\ \forall i\in[\ell]} \mper
    \end{align*}
    Thus, summing over $\pi\in \bbS_3$ and using the independence between $x$ and $y$,
    \begin{align*}
        \pE_{\mu^{\otimes 2}} \Brac{\agree^{(\ell)}(x, y)}
        &= \E_{u_1,\dots,u_\ell \in [n]} \sum_{\pi\in \bbS_3} \sum_{\sigma_1,\dots, \sigma_\ell \in [3]} \pPr_\mu[x_{u_i} = \sigma_i,\ \forall i] \cdot \pPr_\mu[x_{u_i} = \pi^{-1}(\sigma_i),\ \forall i] \\
        &\leq \E_{u_1,\dots,u_\ell \in [n]} \sum_{\pi\in \bbS_3} \sum_{\sigma_1,\dots, \sigma_\ell\in [3]} \frac{1}{2} \Paren{\pPr_\mu[x_{u_i} = \sigma_i,\ \forall i]^2 + \pPr_\mu[x_{u_i} = \pi^{-1}(\sigma_i),\ \forall i]^2 } \\
        \intertext{then since the summation is over all permutations $\pi$ and $\sigma_1,\dots,\sigma_\ell\in [3]$,}
        &= \Abs{\bbS_3} \cdot \E_{u_1,\dots,u_\ell \in [n]} \sum_{\sigma_1,\dots, \sigma_\ell \in [3]} \pPr_\mu[x_{u_i} = \sigma_i,\ \forall i]^2 \mper
        \numberthis \label{eq:bound-agreement}
    \end{align*}

    Now, suppose $\E_{u_1,\dots,u_\ell \in[n]} \KL(\mu(X_{u_1},\dots,X_{u_\ell}) \| \mu(X_{u_1}) \times \cdots \times \mu(X_{u_\ell})) \leq \delta$, then by Pinsker's inequality (\Cref{fact:pinskers}) and Jensen's inequality,
    \begin{equation*}
        \E_{u_1,\dots,u_\ell\in [n]} \sum_{\sigma_1,\dots,\sigma_\ell \in [3]} \Abs{\pPr_{\mu}[x_{u_i} = \sigma_i,\ \forall i] - \prod_{i=1}^\ell \pPr_{\mu}[x_{u_i} = \sigma_i]} \leq \sqrt{2\delta} \mper
    \end{equation*}
    Then, using the fact that $p^2 - q^2 = (p-q)(p+q) \leq 2|p-q|$ for all $p,q\in [0,1]$,
    we can bound \Cref{eq:bound-agreement} by
    \begin{align*}
        \pE_{\mu^{\otimes 2}} \Brac{\agree^{(\ell)}(x, y)}
        &\leq 6 \Paren{\E_{u_1,\dots,u_\ell \in [n]} \sum_{\sigma_1,\dots, \sigma_\ell \in [3]} \prod_{i=1}^\ell \pPr_\mu[x_{u_i} = \sigma_i]^2 + 2\sqrt{2\delta}} \\
        &= 6 \Paren{ \Paren{\E_{u\in[n]} \sum_{\sigma\in [3]} \pPr_{\mu}[x_u = \sigma]^2}^\ell + 2\sqrt{2\delta} } \\
        &= 6 \Paren{\CP(\mu)^\ell + 2\sqrt{2\delta}} \mper
    \end{align*}
    This completes the proof of \Cref{eq:ub-agr}.

    Therefore, since $\pE_{\mu^{\otimes 2}}[\agree^{(\ell)}(x,y)] \geq (\frac{1}{2}+\gamma)^\ell$, we have
    \begin{equation*}
        \CP(\mu)^{\ell} \geq \frac{1}{6}\parens*{\frac{1}{2}+\gamma}^{\ell} - 2\sqrt{2\delta} \mper
    \end{equation*}
    For any $\gamma > 0$, there exists a large enough $\ell\in \N$ and small enough $\delta$ (here $\ell = O(1/\gamma)$ and $\delta = 2^{-O(\ell)}$ suffice) such that the above is at least $(\frac{1}{2}+\frac{\gamma}{2})^\ell$, which means that $\CP(\mu) \geq \frac{1}{2}+\frac{\gamma}{2}$.

    Then, let $I_{\sigma} = \{u: \pPr_{\mu}[x_u=\sigma] > \frac{1}{2}\}$ for $\sigma\in [3]$, which are independent sets.
    By \Cref{lem:large-CP}, one of the sets has size at least $\Omega(\gamma n)$, thus completing the proof.
\end{proof}

We can now finish the analysis of \Cref{alg:coloring-expander} and prove \Cref{thm:3-colorable}.

\begin{proof}[Proof of \Cref{thm:3-colorable}]
    Fix $\gamma = 10^{-3}$.
    If there is an independent set in $G$ with size larger than $(\frac{1}{2}+\gamma)n$, then \Cref{fact:obvious-algo} says that we can find an independent set of size at least $2\gamma n$, and the first step of~\Cref{alg:coloring-expander} would succeed. Therefore, let us assume that this is not the case, and in particular the second step of the algorithm outputs a valid pseudo-distribution $\mu'$ satisfying the constraints listed therein.

    Fix $\ell = 10^4$, and let $\delta$ be some small enough constant as in \Cref{lem:round-3-col}. First, by \Cref{lem:total-corr-reduction}, we can assume that the third step of~\Cref{alg:coloring-expander} reduces the total $\ell$-wise correlation of $\mu'$ to output a pseudo-distribution $\mu$ with total $\ell$-wise correlation $\leq \delta$.

    By \Cref{lem:sum-of-S-pi-squared} we have
    \begin{equation*}
        \ColorConst_G(x,y) \sststile{4}{x,y} \Set{\sum_{\pi\in \bbS_3} \w(S_{\pi})^2 \geq 2 - \frac{1}{1-\lambda_2} - 2\eps \geq 1-\gamma}
    \end{equation*}
    since $\pE[\w(S_\bot)] \leq 2\eps$ by the constraints on $\mu$ and $\lambda_2 \leq 10^{-4}$, $\eps \leq 10^{-4}$.
    Then, consider $3$ assignments $\bx = (x^{(1)}, x^{(2)}, x^{(3)})$.
    By \Cref{lem:27-variable-lemma-sos}, it follows that the pseudo-distribution $\mu$ satisfies
    \begin{equation*}
        \pE_{\mu^{\otimes 3}} \sum_{i<j\in[3]} \sum_{\pi\in \bbS_3}  w \parens*{S^{(ij)}_{\pi}}^{\ell} \geq \parens*{\frac{1+\gamma}{2}}^{\ell} \mper
    \end{equation*}
    By symmetry between the $3$ assignments, it follows that
    \begin{equation*}
        \pE_{\mu^{\otimes 3}} \sum_{\pi \in \bbS_3}  w(S_{\pi})^{\ell}
        = \pE_{\mu^{\otimes 2}} \bracks*{\agree^{(\ell)}(x,y)}
        \geq \frac{1}{3} \parens*{\frac{1+\gamma}{2}}^{\ell}
        \geq \parens*{\frac{1}{2} + \frac{\gamma}{4}}^{\ell}
    \end{equation*}
    since $\ell = 10^4$.
    Then, \Cref{lem:round-3-col} shows that one of the sets $I_{\sigma} = \{u: \pPr_{\mu}[x_u=\sigma] > \frac{1}{2}\}$ for $\sigma\in [3]$ has size at least $\Omega(\gamma n)$.
    The degree of the SoS algorithm required is $O(1/\delta) + d = O(1)$, where $d = d(\eps,\gamma,\ell)$ is the constant from \Cref{lem:27-variable-lemma-sos}.
\end{proof}

%% file: gap-IS.tex
\section{Independent Sets on Certified Small-Set Vertex Expanders}
In this section, we show how to recover large independent sets in graphs that have certificates of small-set vertex expansion (SSVE). Formally,

\begin{theorem}[Formal version of \Cref{thm:ssve-main-intro}] \label{thm:ssve-main}
Let $\eps,\delta \in (0,1/2)$ such that $\eps \leq \delta^3/100$ and $D\in \N$.
Let $G$ be an $n$-vertex graph that is a $(D,\delta)$-certified small-set vertex expander (see~\Cref{def:certified-VE}) and is promised to have an independent set of size $(1/2-\eps)n$. Then there is an algorithm that runs in time $n^{O(D)+\poly(1/\delta)}$ and outputs an independent set of size $\Omega(\delta^3 n)$.
\end{theorem}

Let us start by formally defining SSVEs and SoS certificates for them.

\subsection{Certified Small-Set Vertex Expansion}
\label{sec:certified-SSVE}
To define \emph{certified} small-set vertex expansion, we first need to define the \emph{neighborhood constraints}.
Recall that we use $\BoolConst(x) \coloneqq \Set{x_i^2 = x_i,\  \forall i}$ to denote the Booleanity constraints, and we denote the \emph{neighborhood} of $S$ as $\Gamma_G(S) = S \cup N_G(S)$.

\begin{definition}[Neighborhood constraints] \label{def:neighborhood-constraints}
    For a graph $G = (V,E)$, we define the following system of constraints on variables $\{x_u, y_u\}_{u\in V}$,
    \begin{equation*}
    \begin{aligned}
        \NBConst_G(x, y) = \BoolConst(x,y)
        & \cup \Set{y_u \geq x_v,\ \forall u\in V,\ v\in \Gamma_G(u)}
        \mper
    \end{aligned}
    \end{equation*}
    When the graph $G$ is clear from context, we will drop the subscript $G$.
\end{definition}

For intuition, let $x, y \in \zo^n$ be the indicator vectors of subsets $S, T\subseteq V$ respectively.
The constraints $y_u \geq x_v$ for all $v\in \Gamma_G(u)$ imply that $T \supseteq \Gamma_G(S)$, a superset of the neighborhood of $S$.
This allows us to define the \emph{certified vertex expansion} of a graph.

\begin{definition}[Certified Small-Set Vertex Expansion]\label{def:certified-VE}
Let $G$ be a graph, $D \in \N$, $\delta \in (0,1)$. We say that $G$ is a $(D,\delta)$-certified small-set vertex expander (SSVE) if there exists a univariate polynomial $p$ of degree $\leq D$ such that
 \begin{equation*}
        \NBConst_G(x,y) \sststile{D}{x,y} \Set{\E_{u\in V}[y_u] \geq p \parens*{ \E_{u\in V}[x_u]} } \mcom
    \end{equation*}
where $p(0) = 0$, $p(1) = 1$, and that $p(z) \geq 3 z$ for $z \in [0,\delta]$. Additionally, without loss of generality, we can assume that $p(\delta) = 3\delta$, since otherwise the graph is a $(\delta',D)$-certified SSVE for a larger $\delta'$.
\end{definition}

We remark here that $3$ is an arbitrary constant that we have chosen and any constant $> 2$ suffices for the equation $p(z) \geq 3z,\ \forall z \in [0,\delta]$.
Note also that the conditions on the polynomial $p$ directly implies that $\Psi_\delta(G) \geq 2$, where $\Psi_\delta(G) \coloneqq \min_{S\subseteq V: 0 < |S| \leq \delta|V|} \frac{|N_G(S)|}{|S|}$ is the usual definition of small-set vertex expansion (\Cref{def:ssve}).

Our arguments actually do not require $p$ to be univariate, and one can consider other more general forms of the certificate. For the sake of convenience though we work with the cleaner to state definition given above.

\subsection{Bounding Number of Distinct Independent Sets on SSVEs}
We start by proving that SSVEs cannot have too many distinct $(1/2-\eps)n$-sized independent sets. We will first need a structural result that is true for all graphs.

\subsubsection{Structural result for Independent Sets}
Any assignments $x^{(1)}, x^{(2)}, \dots, x^{(t)} \in \zo^n$ naturally partition the vertices into $2^t$ subsets $\{u\in [n]: x_u^{(i)} = \alpha_i, \forall i\in[t]\}$ for each $\alpha \in \zo^t$.
We will use the same notations as \Cref{sec:IS-on-edge-expanders} (see \Cref{def:w}), where
\begin{gather*}
    \1(u\gets \alpha) \coloneqq \1(x_u^{(1)}=\alpha_1,\ldots, x_u^{(t)}=\alpha_t)= \prod_{i \in [t]} \Paren{x_u^{(i)}}^{\alpha_i} \Paren{1- x_u^{(i)}}^{1-\alpha_i} \mcom \\
    \1(u\gets S) \coloneqq \sum_{\alpha\in S} \1(u\gets \alpha) \mper
\end{gather*}
for each $\alpha \in \zo^t$ and $S\subseteq \zo^t$, and $\w(\alpha) \coloneqq \E_{u\in[n]}[\1(u \gets \alpha)]$, $\w(S) \coloneqq \E_{u\in [n]}[\1(u\gets S)]$.
Moreover, we use the symbol ``$*$'' to denote ``free variables. For example $\{00*\} = \{000, 001\}$.
    
Suppose each $x^{(i)}$ is an indicator vector of some independent set in the original graph, then each of the subsets except $\alpha = \vec{0}$ are independent sets.
Thus, there cannot be any edges between subsets $\alpha, \alpha'$ if $\alpha_i = \alpha_i' = 1$ for some $i\in [t]$, i.e., assignment $x^{(i)}$ labels these subsets as part of an independent set.

This motivates the definition of the following graph:

\begin{definition} \label{def:H-graph}
    For $t\in \N$, define $H_t$ to be the graph on vertex set $\zo^t$ where $\{\alpha, \beta\} \in E(H_t)$ if and only if $\supp(\alpha) \cap \supp(\beta) = \emptyset$, i.e., at least one of $\alpha_i, \beta_i$ is zero for all $i\in [t]$ (thus $\vec{0}$ has a self-loop). 
\end{definition}

Note that the graph $H_2$ is used in \Cref{sec:IS-on-edge-expanders} and is shown in \Cref{fig:IS-gadget}.

The following is the simple fact, written in SoS form, that if $T \subseteq \zo^t$ is an independent set of $H_t$, then vertices in $G$ that are assigned labels from $T$ must form an independent set in $G$.
\begin{claim}\label{claim:is-using-H}
    Let $G = (V, E)$ be a graph,
    let $t\in \N$ and $T \subseteq \zo^t$ be an independent set of $H_t$.
    Then, writing $y_u \coloneqq \1(u\gets T)$ for each $u\in[n]$, we have
    \begin{equation*}
        \ISConst_G(\bx) \sststile{2t}{\bx} \ISConst_G(y)
    \end{equation*}
\end{claim}
\begin{proof}
    First, $y_u$ satisfies the Booleanity constraints $y_u^2 = y_u$ (\Cref{fact:u-alpha-facts}).
    Since $T$ is an independent set, we have $\supp(\alpha) \cap \supp(\beta) \neq \emptyset$ for all $\alpha,\beta\in T$ (note that $T$ cannot contain $\vec{0}$ as it has a self-loop).
    Thus, by \Cref{lem:ones-to-zeros} we have $y_u y_v = 0$.
\end{proof}

We next identify some families of independent sets in $H_t$ which will be used later.

\begin{claim}[Independent sets in $H_t$] \label{lem:IS-in-H}
    Let $H_t$ be as defined in \Cref{def:H-graph}. Then, the following families of subsets of $\zo^t$ are independent sets in $H_t$:
    \begin{enumerate}[(1)]
        \item Subcubes $\calS = \{S_i: i\in [t]\}$, where $S_i = \{\alpha \in \zo^t: \alpha_i = 1\}$ for $i\in [t]$.
        \item $\calT = \{T_{U,i}: U \subseteq [t], |U|\geq 2, i\in U\}$, where $T_{U,i} = A_{U,i} \cup B_{U,i}$, for  
        $A_{U,i} \coloneqq \{\alpha: \alpha_i = 1, \sum_{j\in U\setminus i} \alpha_j \geq 1\}$ and $B_{U,i} \coloneqq \{\alpha: \alpha_i=0, \alpha_j = 1, \forall j\in U\setminus i\}$. 
    \end{enumerate}
\end{claim}
\begin{proof}
    It is clear by definition that $S_i$ is an independent set for each $i\in[t]$.
    For $\calT$, it is also clear that $A_{U,i}$ and $B_{U,i}$ are independent sets, so it suffices to show that there are no edges between $A_{U,i}$ and $B_{U,i}$.
    Consider any $\alpha \in A_{U,i}$ and $\beta \in B_{U,i}$.
    There must exist $j\in U\setminus i$ such that $\alpha_j = 1$, but $\beta_{j'} = 1$ for all $j' \in U \setminus i$, thus $(\alpha, \beta)$ cannot be an edge.
\end{proof}

To interpret \Cref{lem:IS-in-H}, note that the subcubes $S_i$ correspond to the original independent sets indicated by $x^{(1)},x^{(2)},\dots,x^{(t)}$.
On the other hand, the family $\calT$ corresponds to ``derived'' independent sets.
For example, if $t=4$, $i = 1$ and $U = \{1,2,3\}$, then $A_{U,i} = \{1***\} \setminus \{100*\}$ and $B_{U,i} = \{011*\}$.
Then, the vertices in $G$ that are assigned labels from $T_{U,i}$ also form an independent set in $G$.

We next prove the generalization of \Cref{lem:w00-leq-w11}: suppose $\w(S_i) \geq \frac{1}{2}-\eps$ (the independent sets indicated by $x^{(1)},\dots,x^{(t)}$ are large) but $\w(T_{U,i}) \leq \frac{1}{2}+\eta$ (the derived independent sets are not too large), then $\w(00\cdots 0) \leq \w(11\cdots 1) + t(\eps+\eta)$.
Note that when $t=2$ (\Cref{lem:w00-leq-w11}), we don't need the conditions that $\w(T_{U,i}) \leq \frac{1}{2}+\eta$.

\begin{lemma} \label{lem:w0-leq-w1}
    Let $t\in \N$, $t \geq 2$, and $\eps, \eta \geq 0$.
    Let $S_i, T_{U,i} \subseteq \zo^t$ be independent sets of $H_t$ from \Cref{lem:IS-in-H}.
    Let $\calA$ be the following linear constraints:
    \begin{enumerate}[(1)]
        \item $\sum_{\alpha\in \zo^t} \w(\alpha) = 1$. 
        \label{sum-to-one}
        \item $\w(S_i) \geq \frac{1}{2} - \eps$ for all $i\in [t]$.
        \label{subcube}
        \item $\w(T_{U,i}) \leq \frac{1}{2} + \eta$ for all $U \subseteq [t]$, $|U| \geq 2$ and $i\in U$.
        \label{T}
    \end{enumerate}
    Then,
    \begin{equation*}
        \calA \sststile{1}{\{\w(\alpha)\}} \Set{ \w(\vec{0}) \leq \w(\vec{1}) + t(\eps + \eta) } \mper
    \end{equation*}
    More specifically, there are coefficients $\lambda_0 \in \R, \lambda \geq 0$ and $\lambda_2',\dots,\lambda_t' \geq 0$ such that
    \begin{align*}
        \w(\vec{1}) - \w(\vec{0}) + t(\eps+\eta)
        & = \lambda_0 \Paren{1 - \sum_{\alpha\in \zo^t} \w(\alpha) } + \lambda \sum_{i\in[t]} \Paren{ \w(S_i) - \Paren{\frac{1}{2}-\eps}} \\
        &\quad \quad + \sum_{U\subseteq[t],\ |U| \geq 2} \lambda'_{|U|} \sum_{i\in U} \Paren{\Paren{\frac{1}{2}+\eta} - \w(T_{U,i}) } \mper
    \end{align*}
\end{lemma}

\begin{proof}
    We prove by induction on $t$.
    For $t=2$, constraint \ref{subcube} states that $\w(10) + \w(11) \geq \frac{1}{2} - \eps$ and $\w(10) + \w(11) \geq \frac{1}{2} - \eps$, meaning $\w(10) + \w(01) + 2\w(11) \geq 1-2\eps$.
    Subtracting constraint \ref{sum-to-one} gives $\w(00) \leq \w(11) \leq 2\eps$.

    For $t > 2$, denote $W_i \coloneqq \sum_{\alpha: |\alpha|=i} \w(\alpha)$, where $|\alpha| = \alpha_1 + \cdots + \alpha_t$.
    Summing over constraint \ref{subcube} for all $i\in[t]$ gives 
    \begin{equation*}
        \sum_{i=1}^t i W_i \geq \Paren{\frac{1}{2} - \eps} t \mcom
    \end{equation*}
    since each $\alpha$ gets counted $|\alpha|$ times.
    Next, we sum over constraint \ref{T} with $U = [t]$ for all $i\in[t]$.
    Recall from \Cref{lem:IS-in-H} that $T_{U,i} = A_{U,i} \cup B_{U,i}$ where $A_{U,i} \coloneqq \{\alpha: \alpha_i = 1, \sum_{j\in U\setminus i} \alpha_j \geq 1\}$ and $B_{U,i} \coloneqq \{\alpha: \alpha_i=0, \alpha_j = 1, \forall j\in U\setminus i\}$.
    Each $\alpha$ with $|\alpha| \geq 2$ is counted $|\alpha|$ times from $A_{[t],i}$, while each $\alpha$ with $|\alpha| = t-1$ is counted one extra time from $B_{[t],i}$.
    Thus,
    \begin{equation*}
        \sum_{i=2}^t i W_i + W_{t-1} \leq \Paren{\frac{1}{2} + \eta} t \mper
    \end{equation*}
    The above, combined with the previous inequality, yields $W_{t-1} - W_1 \leq (\eps + \eta) t$.

    On the other hand, by induction, for all $U \subseteq [t]$, $|U| = t-1$, we have
    \begin{equation*}
        \sum_{\alpha: \alpha|_U = \vec{0}} \w(\alpha) \leq \sum_{\alpha: \alpha|_U = \vec{1}} \w(\alpha) + (t-1)(\eps + \eta) \mper
    \end{equation*}
    Summing over all $U$ of size $t-1$ gives
    \begin{equation*}
        tW_0 + W_1 \leq W_{t-1} + tW_t + t(t-1)(\eps+\eta) \mcom
    \end{equation*}
    since $\vec{0}$ gets counted $t$ times on the left-hand side and each $\alpha$ with $|\alpha|=1$ gets counted once; similarly for the right-hand side.

    With $W_{t-1} - W_1 \leq (\eps + \eta)t$, we have $t(W_0 - W_t) \leq t^2(\eps+\eta)$.
    As $W_0 = \w(\vec{0})$ and $W_1 = \w(\vec{1})$, this proves that $\calA \sststile{1}{\{\w(\alpha)\}} \Set{ \w(\vec{0}) \leq \w(\vec{1}) + t(\eps + \eta)}$.

    The second statement that $w_{\vec{1}} - w_{\vec{0}} + t(\eps+\eta)$ can be written as a (non-negative) linear combination of the constraints follows by noting that all the derivations above are linear.
    Moreover, by symmetry, the coefficients for $\w(S_i) - (\frac{1}{2}-\eps)$ are the same for all $i\in[t]$;
    similarly, the coefficients for $(\frac{1}{2}+ \eta) - \w(T_{U,i})$ are the same for all $U$ of a fixed size and $i\in U$.
\end{proof}

\subsubsection{Deriving an SoS Certificate for Few Distinct Independent Sets}

The crucial observation is that vertices that are assigned $\vec{1}$ can only have neighbors that are assigned $\vec{0}$.
Therefore, for vertex expanders, $\w(\vec{0})$ must be large compared to $\w(\vec{1})$.

\begin{lemma}  \label{lem:VE-to-w0}
    Let $\delta \in (0,1/2)$, and let $G = (V,E)$ be a $(D,\delta)$-certified SSVE with polynomial $p$ as in \Cref{def:certified-VE}.
    Then, for $t\in \N$ and variables $\bx = (x^{(1)}, \dots, x^{(t)})$,
    \begin{equation*}
        \ISConst(\bx) \sststile{tD}{\bx} 
        \Set{ \w(\vec{0}) + \w(\vec{1}) \geq  p\parens*{\w(\vec{1})} } \mper
    \end{equation*}
 \end{lemma}

\begin{proof}
    Recall from \Cref{def:w} the notations $\1(u\gets \vec{1}) = \prod_{i\in [t]} x_u^{(i)}$ and $\1(u\gets \vec{0}) = \prod_{i\in[t]} (1-x^{(i)}_u)$.
    Define $\ol{x}_u = \1(u\gets \vec{1})$ (indicator that $u$ gets assigned $\vec{1}$) and $y_u = \1(u\gets \vec{1}) + \1(u\gets \vec{0})$ (indicator that $u$ gets assigned $\vec{1}$ or $\vec{0}$).
    We now verify that the constraints in $\NBConst(\ol{x}, y)$ are all satisfied.
    First, by the Booleanity constraints are satisfied due to \Cref{fact:u-alpha-facts}.
    Next, $y_u \geq x_u$ is obvious.
    Finally, by \Cref{lem:ones-to-zeros}, for all edges $\{u,v\}\in E$ we have $\1(u\gets \vec{0}) \geq \1(v\gets \vec{1})$ (this can be interpreted as ``$v$ gets $\vec{1}$ $\implies$ $u$ gets $\vec{0}$'').
    It follows that $y_u \geq x_v$ for all $v\in N(u)$.

    Then, by the vertex expansion certificate, we have
    \begin{equation*}
        \ISConst(\bx) \sststile{tD}{\bx} 
        \Set{ \E_u[y_u] \geq  p(\E_u[\ol{x}_u]) } \mper
    \end{equation*}
    Noting that $\E_u[\ol{x}_u] = \w(\vec{1})$ and $\E_u[y_u] = \w(\vec{1}) + \w(\vec{0})$ completes the proof.
\end{proof}

In \Cref{lem:w0-leq-w1}, we showed that $\w(\vec{0}) \leq \w(\vec{1}) + t(\eps+\eta)$ if the independent sets are large and the ``derived'' independent sets in $\calT$ are not too large.
On the other hand, in \Cref{lem:VE-to-w0} we showed that $\w(\vec{0}) \geq p(\w(\vec{1})) - \w(\vec{1})$.
We now combine the two and use a covering argument to prove the following key approximate packing statement. We prove that given any set of $2t$ independent sets $x^{(1)},\ldots,x^{(2t)}$, either one of the derived independent sets in $H_{2t}$ is large, or there exist $t$ sets $x^{(i_1)},\ldots,x^{(i_t)}$, that have an intersection that is much larger than $1/2^t$ (which is what one would expect from $t$ random $n/2$-sized sets). Formally, we get the following SoS certificate,

\begin{lemma}\label{lem:is-sos-pf}
    Let $G = (V,E)$ be a $(D,\delta)$-certified SSVE with polynomial $p$, and let $q(z) = p(z)-3z$.
    Let $t\in \N$, $t \geq 2$, and $\eps,\eta \geq 0$,
    and moreover let $\calT$ be the family of independent sets of $H_{2t}$ defined in \Cref{lem:IS-in-H}.
    Then, for variables $\bx = (x^{(1)}, \dots, x^{(2t)})$, we have the following polynomial equality:
    \begin{gather*}
        \sum_{U\subseteq[2t], |U|=t} q\parens*{\w(S_{U\to\vec{1}})} + \lambda \sum_{i\in[2t]}  \Paren{ \E_u[x_u^{(i)}] - \Paren{\frac{1}{2}-\eps}} + \sum_{T \in \calT} \lambda_T \Paren{\Paren{\frac{1}{2}+\eta} - \w(T) } + s(\bx) \\
        = \frac{3}{2}\binom{2t}{t} t(\eps+\eta) - \frac{1}{2} \mcom
        \numberthis \label{eq:cert-ve-is}
    \end{gather*}
    where $\lambda$ and $\{\lambda_T\}_{T\in \calT} \geq 0$ and $s(\bx)$ is a combination of polynomials in $\ISConst$ and SoS polynomials of degree at most $tD$. 
\end{lemma}

\begin{proof}
    We consider strings in $\zo^{2t}$, and for any $U \subseteq [2t]$ and $\beta \in \zo^{|U|}$, denote $S_{U\to\beta} \coloneqq \{\alpha\in \zo^{2t}: \alpha_U = \beta\}$.

    For all $U\subseteq [2t]$ with $|U|=t$, we apply \Cref{lem:w0-leq-w1} to $\{x^{(i)}\}_{i\in U}$.
    Note that constraint~\ref{sum-to-one} in \Cref{lem:w0-leq-w1} is automatically satisfied by definition, and $\w(S_i) = \E_u[x_u^{(i)}]$.
    Moreover, denote $\calT_{U} \coloneqq \{T_{U',i}: U' \subseteq U, |U'| \geq 2, i\in U'\} \subseteq \calT$, i.e., the independent sets in $\calT$ restricted to $U$.
    Then, \Cref{lem:w0-leq-w1} with parameter $\eta$ gives
    \begin{gather*}
        \w(S_{U\to\vec{1}}) - \w(S_{U\to\vec{0}}) + t(\eps+\eta) = r_U(\bx) \numberthis \label{eq:w0-leq-w1}  \\
        \text{where} \quad r_U(\bx) \coloneqq \lambda' \sum_{i\in U}\Paren{ \E_u[x_u^{(i)}] - \Paren{\frac{1}{2}-\eps}} + \sum_{T \in \calT_U} \lambda_T' \Paren{\Paren{\frac{1}{2}+\eta} - \w(T) } \mper
    \end{gather*}
    with coefficients $\lambda', \lambda_T' \geq 0$.
    \Cref{eq:w0-leq-w1} should be interpreted as ``$r_U(\bx) \geq 0 \implies \w(S_{U\to\vec{0}}) \leq \w(S_{U\to\vec{1}}) + t(\eps+\eta)$''.
    For convenience, we will denote $\eta' \coloneqq \eps + \eta$.

    On the other hand, \Cref{lem:VE-to-w0} gives an opposite inequality:
    \begin{equation*}
        \w(S_{U\to\vec{0}}) + \w(S_{U\to\vec{1}}) - p\parens*{\w(S_{U\to\vec{1}})} = s_U(\bx) \mcom
        \numberthis \label{eq:w0-geq-w1}
    \end{equation*}
    where $s_U(\bx)$ is a combination of polynomials in $\ISConst$ and SoS polynomials.

    Summing up \Cref{eq:w0-leq-w1,eq:w0-geq-w1} gives $2\cdot \w(S_{U\to\vec{1}}) - p(\w(S_{U\to\vec{1}})) + t\eta' = r_U(\bx) + s_U(\bx)$.
    Then, denoting $q(z) = p(z) - 3z$, we have $\w(S_{U\to\vec{1}}) = t\eta' - q(\w(S_{U\to\vec{1}})) - r_U(\bx) - s_U(\bx)$.
    Again using \Cref{eq:w0-leq-w1}, we get
    \begin{align*}
        \w(S_{U\to\vec{0}}) + \w(S_{U\to\vec{1}}) 
        &= 2\w(S_{U\to\vec{1}}) + t\eta' - r_U(\bx) \\
        &= 3t\eta' - 2 q(\w(S_{U\to\vec{1}})) - 3 r_U(\bx) - 2s_U(\bx) \mper
        \numberthis \label{eq:w0-w1-ub}
    \end{align*}
    This is interpreted as ``$q(\w(S_{U\to\vec{1}})), r_U(\bx), s_U(\bx) \geq 0 \implies \w(S_{U\to\vec{0}}) + \w(S_{U\to\vec{1}}) \leq O(t\eta')$''.

    Now, we sum up \Cref{eq:w0-w1-ub} for all $U \subseteq [2t]$ with $|U| = t$.
    For all $\alpha\in \zo^{2t}$, there must be a $U\subseteq[2t]$ of size $t$ such that either $\alpha|_U= \vec{0}$ or $\alpha|_U = \vec{1}$ (simply take the $0$s or $1$s), thus every $\alpha$ is covered, i.e.,
    \begin{equation*}
        \sum_{U \subseteq [2t]: |U|=t} \w(S_{U\to\vec{0}}) + \w(S_{U\to\vec{1}}) = 1 + s'(\bx) \mcom
    \end{equation*}
    where $s'(\bx)$ is a non-negative polynomial given the Booleanity constraints.
    Thus,
    \begin{equation*}
        \binom{2t}{t} 3t\eta' - \sum_{U\subseteq[2t], |U|=t} 2q\parens*{\w(S_{U\to\vec{1}})} - 3r_U(\bx) - 2s_U(\bx) = 1 + s'(\bx) \mper
    \end{equation*}
    
    Then, writing out $r_U$ (\Cref{eq:w0-leq-w1}), we have
    \begin{gather*}
        \sum_{U\subseteq[2t], |U|=t} q\parens*{\w(S_{U\to\vec{1}})} + \lambda \sum_{i\in[2t]}  \Paren{ \E_u[x_u^{(i)}] - \Paren{\frac{1}{2}-\eps}} + \sum_{T \in \calT} \lambda_T \Paren{\Paren{\frac{1}{2}+\eps} - \w(T) } + s(\bx) \\
        = \frac{3}{2}\binom{2t}{t} t\eta' - \frac{1}{2} \mcom
    \end{gather*}
    where $\lambda$ and $\{\lambda_T\}_{T\in \calT} \geq 0$ and $s$ is a combination of polynomials in $\ISConst$ and SoS polynomials.
    This completes the proof. 
\end{proof}

\subsection{Rounding Algorithm}\label{sec:rounding-is-ssve}
In this section we complete the proof of~\Cref{thm:ssve-main}. We use the SoS certificate in \Cref{lem:is-sos-pf} to show that given a pseudo-distribution $\mu$ with sufficiently large degree, there is a rounding algorithm that outputs an independent set of size $\poly(\delta n)$.

To prove~\Cref{thm:ssve-main}, we solve an SoS relaxation, with sufficiently large degree, for \eqref{eq:is-program} to obtain $\mu$ that satisfies the independent set and Booleanity constraints. We fix $\mu$ throughout this section. Then for $t = 2\log(1/\delta)$ we can apply the $\pE_{\mu^{\otimes 2t}}$ operator on the SoS certificate (\Cref{eq:cert-ve-is}).
Since $x^{(1)},\dots,x^{(2t)} \sim \mu^{\otimes 2t}$, by symmetry, $\pE_{\mu}[q(\w(S_{U\to\vec{1}})]$ is the same for all $U$ of size $t$, and we can simply write it as $\pE_{\mu}[q(\w(\vec{1}))]$ where $\vec{1}$ has length $t$. Furthermore by setting $\eta$ appropriately, the right-hand side of \Cref{eq:cert-ve-is} is $\leq -\frac{1}{4}$. Then, since $\mu$ satisfies the constraint $\E_u[x_u] \geq \frac{1}{2}-\eps$, one of the following must be true:
\begin{enumerate}[(1)]
    \item $\pE_{\mu^{\otimes 2t}}[\w(T) - (1/2+\eta)] > 0$ for some $T \in \calT$,
    \label{item:easy-case}
    \item $\pE_{\mu^{\otimes t}}[q(\w(\vec{1}))] \leq \frac{-1}{4\binom{2t}{t}}$.
    \label{item:hard-case}
\end{enumerate}

We handle these two cases separately via \Cref{lem:rounding-easy,lem:rounding-hard} below. Then in \Cref{sec:proof-of-ssve-thm}, we combine these lemmas in a straightforward way to get a proof of \Cref{thm:ssve-main}. 

\begin{lemma}[Case \ref{item:easy-case}] \label{lem:rounding-easy}
There is an $O(n)$-time algorithm that given $\mu$ satisfying $\pE_{\mu^{\otimes 2t}}[\w(T)] \geq \frac{1}{2}+\eta$, outputs an $\Omega(\eta n)$-sized independent set.
\end{lemma}

\begin{proof}
We start by obtaining a new pseudo-distribution $\cD$ over variables $y = (y_1,\ldots,y_n)$, using the pseudo-distribution $\mu^{\otimes 2t}$, where we will show that $y$ satisfies the independent set and Booleanity constraints.
Given an assignment $\bx = (x^{(1)},\ldots, x^{(2t)})$ from $\mu$, let $y_u = 1$ if $u \in T$ and $0$ otherwise. Since $T$ is an independent set in $H_{2t}$, we get that $y$ is also an independent set. Formally, for all monomials $y_S$, define $\pE_\cD[y_S]$ as $\pE_{\bx \sim \mu^{\otimes 2t}}[\prod_{u \in S}\1(u \gets T)]$. It is easy to check that this is a valid degree-$O(1)$ pseudo-distribution and furthermore it satisfies the independent set constraints by~\Cref{claim:is-using-H}.

We know that $\pE_\cD[y_u] = \pE_{\mu^{\otimes 2t}}[\w(T)]$ which is at least $1/2+\eta$ by assumption. By averaging there are at least $\eta/2$-fraction of vertices with $\pE_{\cD}[y_u] \geq 1/2+\eta/2$, and thus outputting this set of vertices gives us an independent set of size $\eta n/2$ by~\Cref{fact:obvious-is}.
\end{proof}

Rounding when $\pE_{\mu}[q(\w(\vec{1}))] < 0$ turns out to be much more non-trivial:
\begin{lemma}[Case \ref{item:hard-case}] \label{lem:rounding-hard}
Given a pseudo-distribution $\mu$ with $\deg(\mu) \geq \poly(1/\delta)+O(tD)$, there is an $n^{\poly(1/\delta)}$-time algorithm that outputs an $\Omega(\delta n)$-sized independent set when $\pE_{\mu^{\otimes t}}[q(\w(\vec{1}))] \leq \frac{-1}{4\binom{2t}{t}}$ for $t = \ceil{\log_2(4/\delta)}$.
\end{lemma}

Even though $q(z) \geq 0$ for all $z\in [0,\delta]$, in pseudo-expectation, $\pE_{\mu}[q(\w(\vec{1}))] < 0$ does not imply that $\pE_{\mu}[\w(\vec{1})] \geq \delta$.
To remedy this, we utilize the indicator function $\1(z\geq \delta)$ and its degree-$O(\frac{1}{\nu}\log^2 \frac{1}{\nu})$ polynomial approximation $Q_{\delta,\nu}(z)$ from \Cref{lem:step-approx}, as well as techniques developed in \cite{BafnaMinzer}.
Our strategy to prove \Cref{lem:rounding-hard} is as follows:

\begin{enumerate}[(1)]
    \item Show that $\pE_{\mu}[Q_{\delta,\nu}(\w(\vec{1}))]$ is large (\Cref{lem:Q-large}).

    \item Show that conditioning via global correlation reduction (\Cref{lem:ragh-tan}) gives a product pseudo-distribution $\mu' = \mu_1 \times \cdots \times \mu_t$ such that each $\mu_i$ has small global correlation while maintaining that $\pE_{\mu_1 \times \cdots \times \mu_t}[Q_{\delta,\nu}(\w(\vec{1}))]$ is large (\Cref{lem:gc-reduction}).

    \item Show that small global correlation implies that conditioning $\mu'$ on $Q_{\delta,\nu}(\w(\vec{1}))$ does not change the marginal distributions much for most vertices $u\in [n]$ (\Cref{lem:indep-after-conditioning}).

    \item Show that conditioning $\mu'$ on $Q_{\delta,\nu}(\w(\vec{1}))$ results in large $\pE_{\mu'|Q_{\delta,\nu}(\w(\vec{1}))}[\w(\vec{1})]$ (\Cref{lem:condition-on-Q}).

    \item Combining the results above, we round to an independent set from one of $\mu_1,\mu_2,\dots,\mu_t$.
\end{enumerate}

We start with the first lemma,

\begin{lemma} \label{lem:Q-large}
    Let $C,\beta > 0$ and $0 < \nu < \delta < 1$.
    Let $q(z)$ be a univariate polynomial such that $q(z) \geq 0$ for all $z\in [0,\delta]$ and $|q(z)| \leq C$ for all $z\in [0,1]$.
    Let $\mu$ be a pseudo-distribution over $z$ of degree $\max(\deg(q), O(\frac{1}{\nu}\log^2 \frac{1}{\nu}))$ satisfying $0 \leq z \leq 1$ such that $\pE_{\mu}[q(z)] \leq -\beta$.
    Then,
    \begin{equation*}
        \pE_{\mu}[Q_{\delta,\nu}(z)] \geq \beta/C \mper
    \end{equation*}
\end{lemma}
\begin{proof}
    The lemma follows immediately from the following claim: $-q(z) \leq C Q_{\delta,\nu}(z)$ for all $z\in [0,1]$.
    This can be verified by a simple case analysis.
    If $z \in [0,\delta]$, then $q(z) \geq 0$, hence $-q(z) \leq 0 \leq C Q_{\delta,\nu}(z)$.
    If $z \in [\delta, 1]$, then $-q(z) \leq |q(z)| \leq C \leq C Q_{\delta,\nu(z)}$ since by \Cref{lem:step-approx} we have $Q_{\delta,\nu}(z) \in [1,1+\nu]$ for $z \in [\delta,1]$.

    Then, $-q(z) \leq C Q_{\delta,\nu}(z)$ for $z\in[0,1]$ is a univariate inequality and thus has an SoS proof of degree $\max(\deg(q), O(\frac{1}{\nu}\log^2 \frac{1}{\nu}))$ by \Cref{fact:univariate-interval}.
    Thus, $\pE_{\mu}[Q_{\delta,\nu}(z)] \geq \frac{1}{C}\pE_{\mu}[-q(z)] \geq \beta/C$.
\end{proof}

The next two lemmas are variants of results proved in \cite{BafnaMinzer}. We give the proof in \Cref{sec:lemmas-bm} for the sake of completeness.

\begin{lemma}\label{lem:gc-reduction}
    For all $\tau, \beta \in (0,1)$ and $t, D \in \N$, the following holds:
    Let $\mu$ be a pseudo-distribution of degree $D + \Omega(t^2/\beta\tau)$ over $x\in \zo^n$ satisfying the Booleanity constraints.
    Let $\bx = (x^{(1)},x^{(2)},\dots,x^{(t)})$ and let $P(\bx)$ be a polynomial such that $\pE_{\mu^{\otimes t}}[P(\bx)] \geq \beta$ and $\mu^{\otimes t}$ satisfies the constraint $P(\bx) \leq 1$.
    Then there exist subsets $A_1, \ldots, A_t \subseteq [n]$ of size at most $O(\frac{t}{\beta\tau})$ and strings $y_1,\ldots, y_t$ such that conditioning $\mu$ on the events $x|_{A_1} = y_1,\ldots, x|_{A_t} = y_t$ gives pseudo-distributions $\mu_1,\ldots,\mu_t$ of degree at least $D$ such that: 
    \begin{enumerate}
        \item $\pE_{\bx \sim \mu_1\times\ldots \times \mu_t}[P(\bx)] \geq \frac{\beta}{2}$.
        \item For all $i \in [t]$, $\E_{u,v \sim [n]}[I_{\mu_i}(x_u;x_v)] \leq \tau$.
    \end{enumerate}
\end{lemma}

\begin{lemma}\label{lem:indep-after-conditioning}
    For all $\tau,\beta,\nu \in (0,1)$, $t,D \in \N$, the following holds:
    Let $\mu_1,\dots,\mu_t$ be pseudo-distributions of degree $D + \Omega(t)$ over $x\in \zo^n$ satisfying the Booleanity constraints, and let $\mu = \mu_1 \times \cdots \times \mu_t$.
    Let $P(\bx)$ be a polynomial such that $\pE_{\mu}[P(\bx)] \geq \beta$ and $\mu$ satisfies the constraint $0 \leq P(\bx) \leq 1$.
    Moreover, suppose for all $i \in [t]$, we have $\E_{u,v \sim [n]}[I_{\mu_i}(x_u;x_v)] \leq \tau$.
    Then, conditioning on $P$ preserves independence for most $u\in [n]$:
    \begin{equation*}
        \Pr_{u \in [n]} \bracks*{TV(\bx_u|P , \bx_u) \geq \nu} \leq O \parens*{\frac{\sqrt{t\tau}}{\beta\nu^2}} \mcom
    \end{equation*}
    where $\bx_u = (x_u^{(1)},\dots,x_u^{(t)})$ is the marginal distribution from $\mu$ and $\bx_u|P$ refers to the marginal from the reweighted distribution $\mu | P$. 
\end{lemma}

Finally, we prove the following,
\begin{lemma}
\label{lem:condition-on-Q}
    Let $\beta > 0$ and $0 < \nu < \delta < 1$.
    Let $\mu$ be a pseudo-distribution of degree $O(\frac{1}{\nu}\log^2 \frac{1}{\nu})$ on variable $z$ that satisfies $0 \leq z \leq 1$.
    Suppose $\pE_{\mu}[Q_{\delta,\nu}(z)] \geq \beta$.
    Then, conditioned on $Q_{\delta,\nu}$, we have
    \begin{equation*}
        \pE_{\mu|Q_{\delta,\nu}}[z] \geq \delta - \frac{\nu}{\beta} \mper
    \end{equation*}
\end{lemma}
\begin{proof}
    We first claim that for all $z\in [0,1]$,
    \begin{equation*}
        Q_{\delta,\nu}(z) (z-\delta) \geq - \nu \mper
    \end{equation*}
    This can be verified using the properties of $Q_{\delta,\nu}$ (\Cref{lem:step-approx}) and some case analysis on $z$.
    \begin{itemize}
        \item For $z\in [0,\delta-\nu]$, we have $Q_{\delta,\nu}(z) \in [0,\nu]$, so $Q_{\delta,\nu}(z) (z-\delta) \leq 0$ and $Q_{\delta,\nu}(z) |z-\delta| \leq \nu$.
        \item For $z\in [\delta-\nu, \delta]$, we have $Q_{\delta,\nu}(z) \leq 1$, so $Q_{\delta,\nu}(z) (z-\delta) \leq 0$ and $Q_{\delta,\nu}(z) |z-\delta| \leq \nu$.
        \item For $z\in [\delta,1]$, we have $Q_{\delta,\nu}(z) (z-\delta) \geq 0$.
    \end{itemize}
    Since this is a univariate inequality, by \Cref{fact:univariate-interval} we automatically get a degree-$\wt{O}(1/\nu)$ SoS proof.
    It follows that $\pE_{\mu}[Q_{\delta,\nu}(z) z] \geq \delta \cdot \pE_{\mu}[Q_{\delta,\nu}(z)] - \nu$.
    Thus, the conditioned pseudo-distribution satisfies
    \begin{equation*}
        \pE_{\mu|Q_{\delta,\nu}}[z] = \frac{\pE_{\mu}[z \cdot Q_{\delta,\nu}(z)]}{\pE_{\mu}[Q_{\delta,\nu}(z)]}
        \geq \delta - \frac{\nu}{\beta} \mper
        \qedhere
    \end{equation*}
\end{proof}

We are now ready to prove \Cref{lem:rounding-hard}.

\begin{proof}[Proof of \Cref{lem:rounding-hard}]
Set $\beta = \frac{1}{4\binom{2t}{t}}$.
Let $Q_{\delta,\nu}(z)$ be the polynomial approximation to the indicator function $\Ind[z \geq \delta]$ with error $\nu = \delta\beta^2 = \poly(\delta)$ from \Cref{lem:step-approx}. Consider the polynomial $P(\bx) = Q_{\delta,\nu}(\w(\vec{1}))$ that approximates $\Ind[\w(\vec{1}) \geq \delta]$.
Since $q(z) = p(z) - 3z$ and by the assumption that $p(z) \in [0,1]$, we have $|q(z)| \leq 3$.
Thus, by \Cref{lem:Q-large}, $\pE_{\mu^{\otimes t}}[q(\w(\vec{1}))] \leq -\beta$ implies that
\begin{equation}\label{eq:averaging}
\pE_{\mu^{\otimes t}}[P(\bx)] \geq \beta/3 \mper
\end{equation}



\parhead{Global correlation reduction:}
We can now apply global correlation reduction via \Cref{lem:gc-reduction} with $\tau = \beta^4 \nu^4/t$ and $E[\bx] = P(\bx)$ to get the pseudo-distribution $\cD = \mu_1 \times \ldots \times \mu_t$ such that,
\begin{itemize}
    \item For all $i \in [t]$: $\E_{a,b \in V}[I_{\mu_i}[x_a;x_b]] \leq \tau$.
    \item $\pE_{\cD}[P(\bx)] \geq \beta/6$.
\end{itemize}

\parhead{Conditioning $\cD$ on $P$ preserves marginals:} We can now apply \Cref{lem:indep-after-conditioning} to show that after conditioning $\cD$ on $P(\bx)$, most marginals are preserved. More precisely,
\[\Pr_{u \in V}[ TV( \bx_u|P, \bx_u) \geq \nu] \leq O\left(\frac{\sqrt{t\tau}}{\beta\nu^2}\right),\]
where the distribution $\bx_u = (x_u^{(1)},\ldots,x_u^{(t)})$ is the marginal from $\cD$ and $\bx_u|P$ refers to the marginal from the reweighted distribution $\cD~|~P(\bx)$. 

\parhead{After conditioning on $P(\bx)$:}
By \Cref{lem:condition-on-Q}, we have
We will show that,
\begin{equation}\label{eq:100}
\pE_{\cD|P(\bx)}[\w(\vec{1})] \geq \delta - O\left(\frac{\nu}{\beta}\right) = \delta(1-O(\beta)) \mcom    
\end{equation}
as we would expect if $\cD$ was an actual distribution and $P$ was truly equal to $\1(\w(\vec{1}) \geq \delta)$.

\parhead{Rounding to a large independent set:}
We know that for most $u \in V$, $\TV(\bx_u|P, \bx_u) \leq \nu$, which gives that
\begin{align*}
\pE_{\cD|P}[\w(\vec{1})]&=\pE_{\cD|P}[\E_u[\Ind[x_u^{(1)} = 1,\ldots,x_u^{(t)} = 1]]]\\
&\leq \pE_{\cD}[\E_u[\Ind(x_u^{(1)} = 1,\ldots,x_u^{(t)} = 1)]] + O(\nu)+O\left(\frac{\sqrt{t\tau}}{\beta\nu^2}\right)\\
&\leq \pE_{\cD}[\E_u[\Ind(x_u^{(1)} = 1,\ldots,x_u^{(t)} = 1)]] +O(\beta) \mper
\end{align*}
We can now bound the first term:
\begin{align*}
\pE_{\cD}[\E_u[\Ind[x_u^{(1)} = 1,\ldots,x_u^{(t)} = 1]]] &= \E_u[\pPr_{\mu_1}[x_u = 1]\ldots \pPr_{\mu_t}[x_u = 1]]\\
&\leq \E_u\left[\E_{i \in [t]}[\pPr_{\mu_i}[x_u = 1]]^t\right]\\
&\leq \E_u\left[\E_{i \in [t]}[\pPr_{\mu_i}[x_u = 1]^t]\right]\\
&= \E_i[\E_{u}[\pPr_{\mu_i}[x_u = 1]^t]] \mcom
\end{align*}
where the first inequality is the AM-GM inequality, and the second one follows by Jensen's inequality. By using \eqref{eq:100} to get a lower bound on the above, we get that there is an $i \in [t]$ for which, $\E_{u}[\pPr_{\mu_i}[x_u = 1]^t] \geq \delta(1-O(\beta)) \geq \delta/2$.
Denoting $p_u = \pPr_{\mu_i}[x_u=1]$, we have $\E_u[p_u^t] \geq \delta/2$. Let $\alpha$ be the fraction of $u$ with $p_u > \frac{1}{2}$, then since $p_u \leq 1$,
\begin{equation*}
\frac{\delta}{2} \leq \E_u[p_u^t] \leq \alpha + (1-\alpha)\cdot 2^{-t} \leq \alpha + 2^{-t} \mper
\end{equation*}  
Thus, $\alpha \geq \Omega(\delta)$ since $t \geq \log_2(4/\delta)$ implies $2^{-t} \leq \delta/4$. By~\Cref{fact:obvious-is}, the set of vertices with $p_u>1/2$ form an independent set.
\end{proof}

\subsubsection{Proof of \texorpdfstring{\Cref{thm:ssve-main}}{Theorem~\ref{thm:ssve-main}}}
\label{sec:proof-of-ssve-thm}

Let $\pE_\mu$ be the degree $O(tD)+\poly(1/\delta)$ pseudo-expectation operator found by the SDP and let $\mu$ be the corresponding pseudo-distribution. Let $t = \ceil{\log_2(4/\delta)}$ and $\eta = \delta^3/100$. Applying $\pE_{\mu^{\otimes 2t}}$ on both sides of \Cref{eq:cert-ve-is} from \Cref{lem:is-sos-pf} we get that,
\begin{align*}
&\sum_{U\subseteq[2t], |U|=t}\pE_{\mu^{\otimes 2t}}[q\parens*{\w(S_{U\to\vec{1}})}] + \lambda \sum_{i\in[2t]}\pE_{\mu^{\otimes 2t}}\left[ \E_u[x_u^{(i)}] - \Paren{\frac{1}{2}-\eps}\right] \\
& \quad + \sum_{T \in \calT} \lambda_T \pE_{\mu^{\otimes 2t}}\left[\Paren{\frac{1}{2}+\eps} - \w(T)\right]
+\pE_{\mu^{\otimes 2t}}[s(\bx)] 
= \frac{3}{2}\binom{2t}{t} t(\eps+\eta) - \frac{1}{2} \leq -\frac{1}{4},
\end{align*}
since we have chosen parameters so that $\frac{3}{2}\binom{2t}{t}t(\eps+\eta) \leq 1/4$. Let us now examine each of the terms above.
By symmetry we have that $\sum_{U\subseteq[2t], |U|=t}\pE_{\mu^{\otimes 2t}}[q\parens*{\w(S_{U\to\vec{1}})}] = \binom{2t}{t}\pE_{\mu^{\otimes t}}[q(\w(\vec{1}))]$ where $\vec{1}$ has length $t$. We know that $\mu^{\otimes 2t}$ satisfies the axioms $\ISConst_G(x)$, so we get that for all $i$, $\pE[\E_u[x_u^{(i)}] - \Paren{\frac{1}{2}-\eps}] \geq 0$ and $\pE[s(\bx)] \geq 0$, therefore one of the following must be true:
\begin{enumerate}[(1)]
    \item $\pE_{\mu^{\otimes 2t}}[\w(T) - (1/2+\eta)] > 0$ for some $T \in \calT$,
    \item $\pE_{\mu^{\otimes t}}[q(\w(\vec{1}))] \leq \frac{-1}{4\binom{2t}{t}}$.
\end{enumerate}

If (1) above is true then we apply \Cref{lem:rounding-easy} to round to an independent set of size $\Omega(\eta n)=\Omega(\delta^3 n)$. On the other hand if (2) is true then we apply~\Cref{lem:rounding-hard} to round to an independent set of size $\Omega(\delta n)$, thus completing the proof of the theorem.

%% file: hypercube.tex
\section{Vertex Expansion of the Hypercube}
\label{sec:ve-of-hypercube}

The $n$-dimensional hypercube graph is the graph on vertex set $\zo^n$ where two vertices $x$ and $y$ are connected if $\dist(x,y) = 1$.
The vertex isoperimetry of the hypercube is precisely determined by Harper~\cite{Harper66}.
However, we only need a weaker isoperimetric inequality:
\begin{equation*}
    |N(S)| \geq \Omega\parens*{\frac{1}{\sqrt{n}}} \cdot |S| \parens*{1 - \frac{|S|}{2^n}} \mper
    \numberthis \label{eq:hypercube-ve}
\end{equation*}

First, similar to the neighborhood constraints in \Cref{def:neighborhood-constraints}, for the hypercube graph, we define the \emph{outer boundary} constraints to be
\begin{equation*}
\begin{aligned}
    \VEConst(f,g) = \BoolConst(f,g)
    & \cup \Set{g(x) \geq f(x^{\oplus i}) - f(x),\ \forall x\in \zo^n, i\in [n]} \\
    & \cup \Set{g(x) \leq 1 - f(x),\ \forall x\in \zo^n} \mcom
\end{aligned}
\end{equation*}
where $f$ indicates a subset $S \subseteq \zo^n$ and $g$ indicates the outer boundary $N(S) \coloneqq \{u\notin S: \exists v\in S, (u,v) \in E \}$.

In this section, we prove the following SoS version of \Cref{eq:hypercube-ve}:

\begin{lemma} \label{lem:boolean-ve-sos}
    Let $n\in \N$, and for each $x\in \zo^n$, let $f(x), g(x)$ be indeterminates.
    Then,
    \begin{equation*}
        \VEConst(f,g) \sststile{O(n^2)}{f,g}
        \Set{ \E[g] \geq \frac{c}{\sqrt{n}} \cdot \E[f] (1-\E[f]) } \mcom
    \end{equation*}
    where $c > 0$ is a universal constant.
\end{lemma}

We note that \Cref{eq:hypercube-ve} is implied by the result of Margulis~\cite{Margulis74} and its strengthening by Talagrand~\cite{Talagrand93} (lower bound on the average square root sensitivity).
However, our SoS proof follows a recent alternative proof of Talagrand's result by~\cite{EKLM22} (see \Cref{sec:proof-of-EKLM}).

\subsection{Preliminaries for Boolean Functions}

\parhead{Notations.}
We will follow the notations used in O'Donnell~\cite{O14}.
We only consider Boolean functions $f: \zo^n \to \zo$, and we treat $\{f(x)\}_{x\in\zo^n}$ as indeterminates satisfying the Booleanity constraints $\BoolConst(f) \coloneqq \{f(x)^2 - f(x) = 0,\ \forall x\in \zo^n \}$.
We will often write $\E[f]$ to denote $\E_{x\sim \zo^n}[f(x)]$ for convenience.
For any $x\in \zo^n$ and $i\in[n]$, we denote $x^{\oplus i}$ to be the vector $x$ with the $i$-th bit flipped, and we denote $x^{(i\mapsto 0)}$ and $x^{(i \mapsto 1)}$ to be $x$ with the $i$-th bit set to $0$ and $1$ respectively.

We next define the \emph{sensitivity} of a Boolean function.

\begin{definition}[Gradient and sensitivity] \label{def:grad-sensitivity}
    For $f: \zo^n \to \zo$, denote $\partial_i f(x) \coloneqq f(x^{(i\mapsto 0)}) - f(x^{(i\mapsto 1)})$ (which does not depend on $x_i$).
    The gradient $\nabla f: \zo^n \to \R^n$ is defined as $\nabla f(x) = (\partial_1 f(x), \dots, \partial_n f(x))$.
    Finally, the \emph{sensitivity} of $f$ at $x$, denoted $\sens_f(x)$, is the number of pivotal coordinates for $f$ at $x$, i.e., $\sens_f(x) = \sum_{i=1}^n \1(\partial_i f(x) \neq 0) = \Norm{\nabla f(x)}_2^2$.
\end{definition}

\parhead{Fourier coefficients.}
The functions $\{\chi_S\}_{S\subseteq[n]}$ defined by $\chi_S(x) = \prod_{i\in S} (-1)^{x_i}$ form an orthonormal basis, and $f$ can be written as $f(x) = \sum_{S\subseteq [n]} \wh{f}(S) \chi_S(x)$ where $\wh{f}(S) = \E_x[f(x) \chi_S(x)]$.
The degree-$k$ Fourier weight is defined as $W^k[f] \coloneqq \sum_{S: |S|=k} \wh{f}(S)^2$.
Moreover, we denote $W^{\geq k}[f] = \sum_{\ell \geq k} W^{\ell}[f]$ and $W^{[k_1,k_2]} \coloneqq \sum_{\ell=k_1}^{k_2} W^{\ell}[f]$.
Note that $\wh{f}(S)$ and $W^k[f]$ are \emph{linear} and \emph{quadratic} polynomials in the indeterminates $\{f(x)\}_{x\in\zo^n}$ respectively.

The following is the standard Parseval's theorem written in SoS form.

\begin{fact} \label{fact:variance-W1}
    $\BoolConst(f) \sststile{2}{f} \Set{  W^{\geq 1}[f] = \E[f](1-\E[f]) }
    \cup \Set{ W^1[f] = \frac{1}{4} \Norm{\E[\nabla f]}_2^2 }$.
\end{fact}


\parhead{Random restrictions.}
Given a set of coordinates $J \subseteq [n]$ and $z\in \zo^{\ol{J}}$ (where $\ol{J} = [n] \setminus J$), the \emph{restriction} of $f$ to $J$ using $z$, denoted $f_{J|z}: \zo^{J} \to \zo$ (following \cite{O14}), is the subfunction of $f$ given by fixing the coordinates in $\ol{J}$ to $z$.

The following is a simple fact (Fact 2.4 of \cite{EKLM22}).

\begin{fact} \label{fact:restrictions-W1}
    Let $d\in \N$ and $d \geq 2$. Suppose $J \subseteq [n]$ is sampled by including each $i\in[n]$ with probability $1/d$ and $z \in \zo^{\ol{J}}$ is sampled uniformly, then
    \begin{equation*}
        \E_{J,z}[W^1[f_{J|z}]] \geq \Omega(1) \cdot W^{[d,2d]}[f] \mper
    \end{equation*}
\end{fact}

\parhead{Vertex boundary.}
For $f,g$ satisfying $\VEConst(f,g)$, $g$ indicates the vertex boundary of $f$.
We first prove a simple but crucial lemma stating that $\partial_i f(x) (g(x) + g(x^{\oplus i})) = \partial_i f(x)$, which is true because if $\partial_i f(x) \neq 0$, say $f(x) = 1$ and $f(x^{\oplus i}) = 0$, then it must be that $g(x) = 0$ and $g(x^{\oplus i}) = 1$.

\begin{lemma} \label{lem:g-partial-f}
    For any $x\in \zo^n$ and $i\in [n]$,
    $\VEConst(f,g) \sststile{2}{f,g} \Set{\partial_i f(x) (g(x) + g(x^{\oplus i})) = \partial_i f(x)}$.
    In particular, we have that
    $\VEConst(f,g) \sststile{2}{f,g} \Set{\E[g \nabla f] = \frac{1}{2} \E[\nabla f]}$.
\end{lemma}
\begin{proof}
    Fix an $i\in [n]$.
    For convenience, denote $x^0$ and $x^1$ to be $x^{(i\mapsto 0)}$ and $x^{(i\mapsto 1)}$ respectively.
    Recall from \Cref{def:grad-sensitivity} that $\partial_i f(x) = f(x^0) - f(x^1)$.
    We will show that $\partial_i f(x) (g(x^0) + g(x^1)) \geq \partial_i f(x)$ and $\partial_i f(x) (g(x^0) + g(x^1)) \leq \partial_i f(x)$.

    First, observe that $\VEConst(f,g) \sststile{2}{f,g} \{ f(x)g(x) = 0\}$ since $0 \leq f(x) g(x) \leq f(x)(1-f(x)) = 0$.
    Thus,
    \begin{equation*}
        \VEConst(f,g) \sststile{2}{f,g}
        \Set{ (f(x^0) - f(x^1)) \cdot (g(x^0) + g(x^1)) = f(x^0) g(x^1) - f(x^1) g(x^0) } \mper
    \end{equation*}
    Then, using $g(x^0) \geq f(x^1)-f(x^0)$ and $g(x^1) \leq 1-f(x^1)$, we get
    \begin{equation*}
        \VEConst(f,g) \sststile{2}{f,g} \Set{\partial_i f(x) \cdot (g(x^0) + g(x^1)) \leq f(x^0)(1-f(x^1)) - f(x^1)(f(x^1)-f(x^0)) = \partial_i f(x) } \mper
    \end{equation*}
    Similarly, using $g(x^1) \geq f(x^0)-f(x^1)$ and $g(x^0) \leq 1-f(x^0)$, we get
    \begin{equation*}
        \VEConst(f,g) \sststile{2}{f,g} \Set{\partial_i f(x) \cdot (g(x^0) + g(x^1)) \geq f(x^0)(f(x^0)-f(x^1)) - f(x^1)(1-f(x^0)) = \partial_i f(x) } \mper
    \end{equation*}
    This completes the proof.
\end{proof}

\subsection{Proof of the Isoperimetric Inequality by \texorpdfstring{\cite{EKLM22}}{[EKLM22]}}
\label{sec:proof-of-EKLM}

The isoperimetric inequality for the hypercube (\Cref{eq:hypercube-ve}) can be proved by lower bounding the average root sensitivity: $\E_x\Brac{\sqrt{\sens_f(x)}} = \E[ \Norm{\nabla f}_2]$, which is also called the \emph{Talagrand boundary}.
Different proofs of such lower bounds were given by \cite{Talagrand93,EG20,EKLM22}.
In this section, we state the (simplified) proof by \cite{EKLM22} of the following weaker form\footnote{The stronger form is that $\E[\|\nabla f\|_2] \geq \Omega(1) \cdot \Var[f] \cdot \log(1/\Var[f])$.}:

\begin{lemma}[Talagrand boundary] \label{lem:avg-root-sensitivity}
    Given $f: \zo^n \to \zo$,
    \begin{equation*}
        \E[\|\nabla f\|_2] \geq \Omega(1) \cdot W^{\geq 1}[f] = \Omega(1) \cdot \E[f](1-\E[f]) \mper
    \end{equation*}
\end{lemma}

\begin{proof}
First, by the convexity of $\|\cdot \|_2$ and \Cref{fact:variance-W1},
\begin{equation*}
    \E[\|\nabla f\|_2] \geq \Norm{\E[\nabla f]}_2 = 2 \sqrt{W^1[f]} \geq 2 W^1[f] \mper
    \numberthis  \label{eq:root-sens-W1}
\end{equation*}

Next, we consider random restrictions of $f$ with various probabilities.
Fix $d\in \N$, and suppose $J \subseteq [n]$ is sampled by including each $i\in[n]$ with probability $1/d$.
By \Cref{eq:root-sens-W1}, for any $z\in \zo^{\ol{J}}$, the restricted function $f_{J|z}$ satisfies $\E_{x_J}[\|\nabla f_{J|z}(x_J)\|_2] \geq 2 W^1[f_{J|z}]$.
Taking the expectation over $J$ and $z = x_{\ol{J}}$, by \Cref{fact:restrictions-W1} we have $\E_J \E_{x}[\|\nabla f_{J|x_{\ol{J}}}(x_J)\|_2] \geq \Omega(1) \cdot W^{[d,2d]}[f]$.

On the other hand, fix any $x \in \zo^n$,
\begin{equation*}
    \E_J \Brac{\Norm{\nabla f_{J|x_{\ol{J}}}(x_J)}_2} \leq \sqrt{\E_J \Brac{\Norm{\nabla f_{J|x_{\ol{J}}}(x_J)}_2^2} }
    = \sqrt{\E_J \sum_{i=1}^n \partial_i f(x)^2 \cdot \1(i\in J)}
    = \frac{1}{\sqrt{d}}  \cdot \|\nabla f(x)\|_2 \mper
\end{equation*}
Thus, we have
\begin{equation*}
    \frac{1}{\sqrt{d}} \cdot \E[\|\nabla f\|_2] \geq \Omega(1) \cdot W^{[d,2d]}[f] \mper
    \numberthis \label{eq:root-sens-Wd}
\end{equation*}
Summing over $d = 2^k$ for $k=0,1,2,\dots$, we get
\begin{equation*}
    \sum_{k=0}^{\infty} 2^{-k/2} \cdot \E[\|\nabla f\|_2]
    \geq \Omega(1) \sum_{k=0}^{\infty} W^{[2^k,2^{k+1}]}[f]
    = \Omega(1) \cdot W^{\geq 1}[f] \mper
\end{equation*}
This finishes the proof of \Cref{lem:avg-root-sensitivity}.
\end{proof}

\subsection{SoS Isoperimetric Inequality for the Hypercube}

In this section, we prove \Cref{lem:boolean-ve-sos} by proving an SoS version of \Cref{lem:avg-root-sensitivity} (see \Cref{lem:root-sens-variance-sos}).
Unfortunately, $\|\nabla f\|_2$ is not a polynomial of $\{f(x)\}$, hence we need a polynomial approximation of the square root function with constant multiplicative error.
This can be achieved using the Bernstein polynomials (\Cref{def:bernstein}), and we prove the following lemma in \Cref{sec:poly-approx}.

\begin{lemma}[Proxy for square root] \label{lem:proxy}
    Fix $n,m\in \N$ such that $m \geq 64n$.
    There is a degree-$m$ univariate polynomial $B_m(x)$ that satisfies the following properties:
    \begin{enumerate}[(1)]
        \item {\boldmath $0 \leq \sqrt{x^2} \leq x$:} $\Set{0 \leq x \leq n} \sststile{2m}{x} \Set{ 0 \leq B(x^2) \leq x}$.
        \label{item:B-x2-x-sos}

        \item \textbf{Monotone:} $\Set{0 \leq x \leq y \leq n} \sststile{m}{x,y} \Set{B_m(x) \leq B_m(y)}$.
        \label{item:B-increasing-sos}
        

        \item {\bfseries\boldmath $\sqrt{x} \gtrsim x$ for $x \leq 1$:} $\Set{0 \leq x \leq 1} \sststile{m}{x} \Set{B_m(x) \geq \frac{1}{2} x}$.
        \label{item:B-geq-x-sos}

        \item \textbf{Concavity:} For any $N\in \N$, let $z$ be an $N$-dimensional indeterminate, and let $p_1,\dots,p_N$ be probabilities such that $\sum_{i=1}^N p_i = 1$.
        Then, $\{0 \leq z_i \leq n,\ \forall i\in [N]\} \sststile{m}{z} \{ \sum_{i=1}^N p_i B_m(z_i) \leq B_m(\sum_{i=1}^N p_i z_i) \}$.
        \label{item:concavity-sos}

    \end{enumerate}

    For $n$-dimensional indeterminates $u$ and $v$ with Booleanity constraints,
    \begin{enumerate}[(1)]
        \setcounter{enumi}{4}
        \item \textbf{Cauchy-Schwarz:}  $\BoolConst(u,v) \sststile{2m}{u,v} \{ \iprod{u,v} \leq 4 \cdot B_m(\sum_i u_i) B_m(\sum_i v_i) \}$.
        \label{item:cauchy-schwarz-sos}

        \item {\boldmath $\sqrt{\frac{k}{d}} \lesssim \frac{1}{\sqrt{d}} \sqrt{k}$:} For any $d\in \N$,
        $\BoolConst(u) \sststile{m}{u,v} \Set{ B_m(\frac{1}{d} \sum_i u_i) \leq \frac{2}{\sqrt{d}} \cdot B_m(\sum_i u_i) }$.
        \label{item:B-1-over-d}
    \end{enumerate}
\end{lemma}

Thus, we can use $B_m(\Norm{\nabla f}_2^2)$, which has degree $2m$, as a proxy for $\|\nabla f\|_2$.
One difference from \Cref{lem:avg-root-sensitivity} is that we consider $\E[g\|\nabla f\|_2]$ instead of just the square root sensitivity $\E[\|\nabla f\|_2]$, where $g$ is the vertex boundary of $f$.
This is simply for convenience later in the proof of \Cref{lem:boolean-ve-sos}.
Specifically, we will prove:

\begin{lemma}[SoS lower bound of the Talagrand boundary]
\label{lem:root-sens-variance-sos}
    Let $m = 64n$ and let $B_m$ be the polynomial as in \Cref{lem:proxy}.
    Let $f(x),g(x)$ be indeterminates for each $x\in \zo^n$.
    Then,
    \begin{equation*}
        \VEConst(f,g) \sststile{6m^2}{f,g} \Set{ \E\Brac{B_m\Paren{g\Norm{\nabla f}_2^2}} \geq \Omega(1) \cdot W^{\geq 1}[f] = \Omega(1)\cdot \E[f](1-\E[f]) } \mper
    \end{equation*}
\end{lemma}

We first finish the proof of \Cref{lem:boolean-ve-sos} using \Cref{lem:root-sens-variance-sos}.

\begin{proof}[Proof of \Cref{lem:boolean-ve-sos}]
    Since $\BoolConst(f,g) \sststile{4}{f,g} \Set{ g(x)\|\nabla f(x)\|_2^2 \leq n \cdot g(x)^2}$, by \ref{item:B-x2-x-sos} and \ref{item:B-increasing-sos} of \Cref{lem:proxy},
    \begin{equation*}
        \BoolConst(f,g) \sststile{4m}{f,g} \Set{\E \Brac{B_m\Paren{g \Norm{\nabla f}_2^2}}
        \leq \E \Brac{B_m(ng^2)}
        \leq \sqrt{n} \cdot \E[g]} \mper
    \end{equation*}
    Since $m = \Theta(n)$, by \Cref{lem:root-sens-variance-sos} we have
    \begin{equation*}
        \VEConst(f,g) \sststile{O(n^2)}{f,g} \Set{\sqrt{n} \cdot \E[g] \geq \Omega(1) \cdot \E[f](1-\E[f])} \mcom
    \end{equation*}
    which completes the proof.
\end{proof}

To prove \Cref{lem:root-sens-variance-sos}, we start with the SoS version of \Cref{eq:root-sens-W1} that $\E[ \Norm{\nabla f}_2] \gtrsim W^1[f]$.
Recall that this requires the convexity of $\|\cdot \|_2$, which we will SoSize using Cauchy-Schwarz.

\begin{lemma} \label{lem:root-sens-W1-sos}
    Under the same assumptions as \Cref{lem:root-sens-variance-sos},
    \begin{equation*}
        \VEConst(f,g) \sststile{6m^2}{f,g} \Set{ \E\Brac{B_m\Paren{g\Norm{\nabla f}_2^2}} \geq \frac{1}{4} W^1[f] } \mper
    \end{equation*}
\end{lemma}
\begin{proof}
    For any $x, y \in \zo^n$, the Booleanity constraints imply that $\partial_i f(x) \leq \partial_i f(x)^2$.
    Thus,
    \begin{equation*}
        \BoolConst(f,g) \sststile{6}{f,g}
        \Set{ \Iprod{g(x) \nabla f(x), g(y) \nabla f(y)}  \leq \sum_{i=1}^n g(x)g(y) \partial_i f(x)^2 \partial_i f(y)^2 } \mper
    \end{equation*}
    Now, $g(x)\partial_i f(x)^2$ satisfies the Booleanity constraint $g(x)^2 \partial_i f(x)^4 = g(x) \partial_i f(x)^2$.
    Thus, applying \ref{item:cauchy-schwarz-sos} of \Cref{lem:proxy} (Cauchy-Schwarz) with variables $\{g(x)\partial_i f(x)^2\}_{i}$ and $\{g(y) \partial_i f(y)^2\}_i$,
    \begin{equation*}
        \BoolConst(f,g) \sststile{6m}{f,g} 
        \Set{ \Iprod{g(x)\nabla f(x), g(y) \nabla f(y)} 
        \leq 4 \cdot B_m\Paren{g(x)\Norm{\nabla f(x)}_2^2} B_m\Paren{g(y)\Norm{\nabla f(y)}_2^2} } \mper
    \end{equation*}
    Next, by \Cref{lem:g-partial-f}, we have $\VEConst(f,g) \sststile{2}{f,g} \Set{\E[g\nabla f] = \frac{1}{2}\E[\nabla f]}$, and further we have $W^1[f] = \frac{1}{4} \Norm{\E[\nabla f]}_2^2$ by \Cref{fact:variance-W1}.
    Thus, by expanding $\Norm{\E[g\nabla f]}_2^2 = \E_{x,y} \Iprod{g(x)\nabla f(x), g(y) \nabla f(y)}$ and using the above, we get
    \begin{equation*}
    \begin{aligned}
        \VEConst(f,g) \sststile{6m}{f,g}
        \Big\{ W^1[f] &= \Norm{\E[g\nabla f]}_2^2 = \E_{x,y} \Iprod{g(x) \nabla f(x), g(y) \nabla f(y)} \\
        &\leq 4\cdot \E_x\Brac{B_m\Paren{g(x)\Norm{\nabla f(x)}_2^2}}^2 \Big\} \mper
    \end{aligned}
    \numberthis \label{eq:W1-ub}
    \end{equation*}
    Next, $\BoolConst(f) \sststile{2}{f} \set{0 \leq W^1[f] \leq 1}$, so by \ref{item:B-geq-x-sos} of \Cref{lem:proxy}, we can derive $\BoolConst(f) \sststile{2m}{f} \set{ W^1[f] \leq 2\cdot B_m(W^1[f]) }$.
    By \ref{item:B-x2-x-sos} and \ref{item:B-increasing-sos} of \Cref{lem:proxy} and \Cref{eq:W1-ub}, we get
    \begin{equation*}
        \VEConst(f,g) \sststile{6m^2}{f,g}
        \Set{ W^1[f] \leq 2\cdot B_m\Paren{4\cdot \E\Brac{B_m\Paren{g\Norm{\nabla f}_2^2}}^2} 
        \leq  4 \cdot \E\Brac{B_m\Paren{g\Norm{\nabla f}_2^2}} } \mper
    \end{equation*}
    This completes the proof.
\end{proof}

Next, we prove the SoS version of \Cref{eq:root-sens-Wd}: $\E[ \Norm{\nabla f}_2] \gtrsim \sqrt{d} \cdot W^{[d,2d]}[f]$.

\begin{lemma} \label{lem:root-sens-Wd-sos}
    Let $d\in \N$ and $d \geq 2$.
    Under the same assumptions as \Cref{lem:root-sens-variance-sos},
    \begin{equation*}
        \VEConst(f,g) \sststile{6m^2}{f,g} \Set{ \E\Brac{B_m\Paren{g\Norm{\nabla f}_2^2}} \geq \Omega(\sqrt{d}) \cdot W^{[d,2d]}[f] } \mper
    \end{equation*}
\end{lemma}
\begin{proof}
    Fix any $J\subseteq [n]$ and $x_{\ol{J}} \in \zo^{\ol{J}}$,
    by applying \Cref{lem:root-sens-W1-sos} to the restricted function $f_{J|x_{\ol{J}}}: \zo^J \to \zo$, we have
    \begin{equation*}
        \VEConst(f,g) \sststile{6m^2}{f,g}
        \Set{ W^1[f_{J|x_{\ol{J}}}] \leq 4\cdot \E_{x_J}\Brac{B_m\Paren{g(x)\|\nabla_J f(x)\|_2^2}} } \mcom
    \end{equation*}
    where $\nabla_J f(x) = \nabla f_{J|x_{\ol{J}}}(x_J) = (\partial_i f(x))_{i\in J}$.

    Suppose $J \subseteq [n]$ is sampled by including each $i\in[n]$ with probability $1/d$.
    Then, averaging over $J$ and $x_{\ol{J}}$, by \Cref{fact:restrictions-W1} we have
    \begin{equation*}
        \VEConst(f,g) \sststile{6m^2}{f,g}
        \Set{ W^{[d,2d]}[f] \leq O(1) \cdot \E_{J} \E_x \Brac{ B_m\Paren{g(x)\Norm{\nabla_J f(x)}_2^2}} } \mper
    \end{equation*}
    Next, we upper bound the right-hand side.
    Since $\BoolConst(f)$ implies that $\|\nabla f\|_2^2 \leq n$,
    by \ref{item:concavity-sos} of \Cref{lem:proxy} (concavity),
    \begin{equation*}
    \begin{aligned}
        \BoolConst(f,g) \sststile{4m}{f,g}
        \Bigg\{ \E_{J} \E_x \Brac{ B_m\Paren{g(x)\Norm{\nabla_J f(x)}_2^2}}
        &\leq \E_x \Brac{B_m\Paren{ \E_J \Brac{g(x)\Norm{\nabla_J f(x)}_2^2}} } \\
        &= \E_x\Brac{ B_m \Paren{\frac{1}{d} \cdot g(x)\Norm{\nabla f(x)}_2^2}} \Bigg\} \mper
    \end{aligned}
    \end{equation*}
    Here we use the fact that $\E_J \|\nabla_J f(x)\|_2^2 = \frac{1}{d} \|\nabla f(x)\|_2^2$.
    Then, since $g(x)\partial_i f(x)^2$ satisfies the Booleanity constraint, by \ref{item:B-1-over-d} of \Cref{lem:proxy}, for any $x$ we have,
    \begin{equation*}
        \BoolConst(f,g) \sststile{3m}{f,g} \Set{B_m \Paren{\frac{1}{d} g(x) \|\nabla f(x)\|_2^2} \leq \frac{2}{\sqrt{d}} \cdot B_m\Paren{g(x)\Norm{\nabla f(x)}_2^2} } \mper
    \end{equation*}
    This completes the proof.
\end{proof}

With \Cref{lem:root-sens-W1-sos,lem:root-sens-Wd-sos}, we can now prove \Cref{lem:root-sens-variance-sos}.

\begin{proof}[Proof of \Cref{lem:root-sens-variance-sos}]
    From \Cref{lem:root-sens-W1-sos,lem:root-sens-Wd-sos}, we have an SoS proof that $\frac{1}{\sqrt{d}} \cdot \E\Brac{B_m\Paren{g\|\nabla f\|_2^2}} \geq \Omega(1) \cdot W^{[d,2d]}[f]$ for all $d \in \N$.
    Sum over $d = 2^k$ for $k=0,1,2,\dots$ completes the proof.
\end{proof}

%% file: noisy-hypercube.tex
\section{Vertex Expansion of the Noisy Hypercube}
\label{sec:noisy-hypercube}

We first define the noisy hypercube graph.
\begin{definition}[$\gamma$-Noisy hypercube]
    Let $\gamma\in [0,1]$.
    The $n$-dimensional $\gamma$-noisy hypercube is the graph with vertex set $\zo^n$ where two vertices $x,y\in \zo^n$ are connected if $\dist(x,y) \leq \gamma n$.
\end{definition}

It was shown by \cite{FranklR87,GeorgiouMPT10} that when $\gamma \gg 1/\sqrt{n}$, the $\gamma$-noisy hypercube has no large independent set.

\begin{fact}[\cite{FranklR87,GeorgiouMPT10}] \label{fact:frankl-rodl}
    There is a universal constant $K$ such that for all $\gamma \leq 1/4$, the maximum independent set in the $\gamma$-noisy hypercube has size at most $(1-\gamma^2/K)^n \cdot 2^n$.
\end{fact}

Throughout this section, we will assume that $\gamma n$ is an integer for simplicity.
In this section, we prove the following,

\begin{theorem} \label{thm:noisy-hypercube-ssve}
    For any $\gamma \geq C/\sqrt{n}$ where $C$ is a universal constant,
    the $n$-dimensional $\gamma$-noisy hypercube is a $(n^{O(\sqrt{n})}, 1/32)$-certified SSVE.
\end{theorem}

Recall from \Cref{def:certified-VE} that this implies that $\NBConst_G(x,y) \sststile{n^{O(\sqrt{n})}}{x,y} \Set{ \E_u[y_u]  \geq p(\E_u[x_u]) }$ for some polynomial $p$ of degree $\leq n^{O(\sqrt{n})}$ such that $p(z) \geq 3z$ for $z \in [0, 1/32]$.

Then, by \Cref{thm:ssve-main} and \Cref{fact:frankl-rodl}, we have the following corollary.

\begin{corollary}\label{cor:noisy-hyp}
    There are universal constants $C > 0$ and $\eps \in (0,1)$ such that
    for any $\gamma \geq C\sqrt{\frac{\log n}{n}}$, the degree-$n^{O(\sqrt{n})}$ SoS certifies that the $\gamma$-noisy hypercube has maximum independent set size $\leq (\frac{1}{2}-\eps) 2^n$.

    In other words, the degree-$n^{O(\sqrt{n})}$ SoS relaxation of minimum Vertex Cover has integrality gap $\leq 2 - \eps$.
\end{corollary}
\begin{proof}
    Suppose not, then by \Cref{thm:noisy-hypercube-ssve} and \Cref{thm:ssve-main}, one can round to an independent set of size $\delta \cdot 2^n$ for some constant $\delta$, which contradicts \Cref{fact:frankl-rodl} since $(1-\gamma^2/K)^n \leq o_n(1)$ when $\gamma \geq C\sqrt{\frac{\log n}{n}}$.
\end{proof}

\begin{remark} \label{rem:KOTZ14}
    \Cref{cor:noisy-hyp} is motivated by the study of SoS integrality gaps for Vertex Cover on the ``Frankl-R{\"o}dl'' graphs, which are similar to the noisy hypercube and often considered as ``gap instances'' for Independent Set and Vertex Cover (see \cite{KauersOTZ14} for the definition, history and references).
    In particular, \cite{KauersOTZ14} showed that the degree-$O(1/\gamma)$ SoS relaxation of minimum Vertex Cover on the $\gamma$-Frankl-R{\"o}dl graph has integrality gap $\leq 1+o(1)$ when $\gamma \gg \frac{1}{\log n}$.
    However, their techniques do not work in the regime $\sqrt{\frac{\log n}{n}} \ll \gamma \ll \frac{1}{\log n}$.
\end{remark}

In \Cref{sec:ve-of-hypercube} (\Cref{lem:boolean-ve-sos}), we proved the SoS version of the weak isoperimetric inequality of the hypercube graph $H$:
any $S \subseteq \{0,1\}^n$ satisfies $w(N_H(S)) \geq \frac{c}{\sqrt{n}} \cdot w(S) (1 - w(S))$ for some constant $c > 0$.
This implies that $w(\Gamma_H(S)) \geq w(S) + \frac{c}{\sqrt{n}} \cdot w(S) (1 - w(S))$.
Since the noisy hypercube graph can be viewed as powers of the hypercube, we will iteratively apply this to certify the vertex expansion of the noisy hypercube.
Thus, we define the following,

\begin{definition}
    Let $P$ be the univariate degree-2 polynomial
    \begin{equation*}
        P(x) \coloneqq x + \frac{c}{\sqrt{n}} x(1-x) \mper
    \end{equation*}
    where $c > 0$ is the constant in \Cref{lem:boolean-ve-sos}.
    Moreover, for $t \in \N$, let $P_t$ be the degree-$2^t$ polynomial defined iteratively as follows:
    \begin{equation*}
        P_1(x) = P(x), \quad
        P_t(x) = P(P_{t-1}(x)) \ \text{for $t > 1$.}
    \end{equation*}
\end{definition}

Then, we can restate \Cref{lem:boolean-ve-sos} in terms of variables $x,y$ and the polynomial $P$.

\begin{lemma}[Equivalent to \Cref{lem:boolean-ve-sos}]
\label{lem:boolean-ve-sos-P}
    Let $H$ be the $n$-dimensional hypercube graph.
    Then,
    \begin{equation*}
        \NBConst_H(x, y) \sststile{O(n^2)}{x,y}
        \Set{ \E_u[y_u] \geq P\parens*{\E_u[x_u]} } \mper
    \end{equation*}
\end{lemma}

We first prove a useful result.
\begin{lemma} \label{lem:P-monotone}
    $\Set{0 \leq a \leq b \leq 1} \sststile{2}{a,b} \Set{P(a) \leq P(b)}$.
\end{lemma}
\begin{proof}
    $P(b)-P(a) = (b-a) + \frac{c}{\sqrt{n}}((b-b^2)-(a-a^2)) = (b-a)(1 + \frac{c}{\sqrt{n}}(1 - (a+b)))$.
    With constraints $a \leq b \leq 1$, this is an SoS proof.
\end{proof}


We now use \Cref{lem:boolean-ve-sos-P} iteratively to prove the following vertex expansion bound on the noisy hypercube.

\begin{lemma} \label{lem:noisy-hypercube-ve-P}
    Let $G$ be the $\gamma$-noisy hypercube.
    Then,
    \begin{equation*}
        \NBConst_G(x,y) \sststile{n^{O(\gamma n)}}{x,y}
        \Set{ \E_u[y_u] \geq P_{\gamma n}(\E_u[x_u]) } \mper
    \end{equation*}
\end{lemma}
\begin{proof}
    Let $H$ be the hypercube graph, and $\ell = \gamma n$.
    We now define variables $x^{(1)},x^{(2)},\dots,x^{(\ell)}$ as polynomials of $x$:
    \begin{equation*}
        x^{(i)}_u \coloneqq \max \parens*{ \{x^{(i-1)}_v\}_{v\in \Gamma_H(u)} } \mper
    \end{equation*}
    Note that the maximum is over $n+1$ Boolean variables ($H$ has degree $n$ so $|\Gamma_H(u)| = n+1$), so here $\max$ is a polynomial of degree $n+1$:
    $\max(a_1,\dots,a_{n+1}) = 1 - \prod_{i=1}^{n+1} (1-a_i)$.
    Thus, for each $i \leq \ell$ and $u\in V$, $x^{(i)}_u$ is a polynomial of degree $(n+1)^{i}$ in $x$.
    If $x$ is the indicator vector of some set $S \subseteq V$, then $x^{(i)}$ is the indicator of the step $i$ neighborhood of $S$ in $H$.

    Now, it is easy to see that $x^{(i-1)}$ and $x^{(i)}$ satisfy $\NBConst_H(x^{(i-1)}, x^{(i)})$.
    Thus, by \Cref{lem:boolean-ve-sos-P}, we have
    \begin{equation*}
        \BoolConst(x)
        \sststile{O(n^2)\cdot (n+1)^i}{x}
        \Set{ \E_u[x^{(i)}_u] \geq P\parens*{\E_u[x_u^{(i-1)}] } } \mper
    \end{equation*}
    Therefore, by repeatedly applying \Cref{lem:P-monotone}, we have
    \begin{equation*}
        \BoolConst(x)
        \sststile{n^{O(\ell)}}{x}
        \Set{ \E_u[x^{(\ell)}_u] \geq P_{\ell}\parens*{\E_u[x_u] } } \mper
    \end{equation*}
    
    Finally, we prove that $\NBConst_G(x,y) \sststile{n^{O(\ell)}}{x,y} \Set{ y_u \geq x^{(\ell)}_u}$ for all $u\in V$.
    To see this, note that $x_u^{(\ell)} = \max(\{x_v\}_{v\in \Gamma_G(u)})$ where $|\Gamma_G(u)| \leq n^{O(\ell)}$, whereas the constraints $y_u \geq x_v$ for all $v\in \Gamma_G(u)$ in $\NBConst_G(x,y)$ implies that $y_u \geq \max(\{x_v\}_{v\in \Gamma_G(u)}) = x_u^{(\ell)}$ for true Boolean assignments.
    Thus, by \Cref{fact:boolean-function}, it has an SoS proof of degree $n^{O(\ell)}$.
    This completes the proof.
\end{proof}

With \Cref{lem:noisy-hypercube-ve-P}, the proof of \Cref{thm:noisy-hypercube-ssve} is straightforward.

\begin{proof}[Proof of \Cref{thm:noisy-hypercube-ssve}]
    Set $C = \frac{8}{7c}\log 3 + o_n(1)$ such that $(1 + \frac{7c}{8\sqrt{n}})^{C\sqrt{n}} \geq 3$, where $c$ is the constant in \Cref{lem:boolean-ve-sos}.
    Note that we can assume that $\gamma = C /\sqrt{n}$, since for any larger $\gamma$, the $(C/\sqrt{n})$-noisy hypercube is a subgraph of the $\gamma$-noisy hypercube.

    When $x \leq 1/8$, we have $x (1 + \frac{7c}{8\sqrt{n}}) \leq P(x) \leq x (1 + \frac{c}{\sqrt{n}})$.
    For $\ell = \gamma n$, $P_{\ell}(x) \leq x (1+\frac{c}{\sqrt{n}})^\ell \leq e^{Cc}x \leq 4x$.
    Then, if $x \leq 1/32$, then $P_i(x) \leq 1/8$ for all $i\leq \ell$.
    This also implies that for $x \leq 1/32$, we have $P_\ell(x) \geq x (1 + \frac{0.9c}{\sqrt{n}})^\ell \geq 3x$.
\end{proof}

%% file: hardness.tex
\section{Hardness of Finding Independent Sets in \texorpdfstring{$k$}{k}-colorable Expanders}\label{sec:hardness-k-col}

Bansal and Khot~\cite{BansalK09} proved the following hardness result of finding linear-sized independent sets in almost $2$-colorable graphs, which is a strengthening of \cite{KhotR08}.

\begin{proposition}[\cite{BansalK09}] \label{prop:2-colorable-hardness}
    Assuming the Unique Games Conjecture, for any constants $\eps, \gamma > 0$, given an $n$-vertex graph $G$, it is NP-hard to decide between
    \begin{enumerate}
        \item $G$ has $2$ disjoint independent sets of size $(\frac{1}{2}-\eps)n$,
        \item $G$ has no independent set of size larger than $\gamma n$.
    \end{enumerate}
\end{proposition}

Moreover, the above holds if we additionally assume that the graph has degrees $o(n)$.

\begin{proposition}[Formal version of \Cref{prop:hardness}] \label{prop:hardness-formal}
    Assuming the Unique Games Conjecture, for any constants $\eps,\gamma > 0$, given a regular $n$-vertex graph $G$ which is a one-sided spectral expander with $\lambda_2(G) \leq o_n(1)$, it is NP-hard to decide between
    \begin{enumerate}
        \item $G$ is $\eps$-almost $4$-colorable,
        \item $G$ has no independent set of size larger than $\gamma n$.
    \end{enumerate}
\end{proposition}
\begin{proof}
    We start the reduction from \Cref{prop:2-colorable-hardness}.
    Given a graph $G$, we add a regular bipartite graph $H$ (potentially introducing multi-edges) such that $H$ has degree $\Omega(n)$ and the second eigenvalue of its normalized adjacency matrix $\lambda_2(H) = o_n(1)$.
    If $G$ is not regular, we can make the resulting graph $G'$ regular by removing $o(1)$ fraction of edges, denoted $H'$, from $H$.

    If $G$ is $\eps$-almost $2$-colorable, then $G'$ is clearly $\eps$-almost $4$-colorable (since $H$ is $2$-colorable).
    On the other hand, adding edges cannot increase the size of the maximum independent set.

    Next, we prove that $G'$ has small normalized second eigenvalue.
    We can assume that $G$ and $H'$ have maximum degrees $d_G, d_{H'} = o(n)$ while $H$ has degree $d_H = \Omega(n)$.
    Then, $\lambda_2(G') = \frac{1}{d_{G'}} \lambda_2(A_G + A_H - A_{H'}) = \frac{1}{d_{G'}} \cdot \max_{x\perp \vec{1}, \|x\|_2=1} x^\top (A_G + A_H - A_{H'}) x \leq o_n(1)$.
\end{proof}

\paragraph{Hardness for $k$-colorable graphs.}
In this case (as opposed to almost $k$-colorable), we need a hardness conjecture with \emph{perfect completeness}.
The natural candidate is the $2$-to-$1$ (or $d$-to-$1$) conjecture:

\begin{conjecture}[$2$-to-$1$ conjecture with perfect completeness~\cite{Khot02}]
\label{conj:2-to-1}
    For every $\eps > 0$, there exists some $R \in \N$ such that given a label cover instance $\psi$ with alphabet size $R$ such that all constraints are $2$-to-$2$ constraints, it is $\NP$-hard to decide between
    \begin{enumerate}
        \item $\psi$ is satisfiable,
        \item no assignment satisfies more than $\eps$ fraction of the constraints in $\psi$.
    \end{enumerate}
\end{conjecture}

Dinur, Mossel and Regev~\cite{DinurMR06} introduced the following variant of the $2$-to-$1$ conjecture.
We note that the ``$\ltimes$'' constraints (termed ``alpha'' or ``fish-shaped'' constraints) have also appeared in \cite{DinurS05}.
See \cite{DinurS05,DinurMR06} for a precise statement.

\begin{conjecture}[``$\ltimes$'' variant of the $2$-to-$1$ conjecture~\cite{DinurMR06}] \label{conj:fish}
    \Cref{conj:2-to-1} is true even assuming that all constraints in the label cover instance are ``$\ltimes$'' constraints.
\end{conjecture}

Dinur, Mossel and Regev~\cite{DinurMR06} proved the following,

\begin{proposition} \label{prop:fish-3-colorable}
    Assuming \Cref{conj:fish}, for any constant $\gamma > 0$, given an $n$-vertex graph $G = (V,E)$, it is $\NP$-hard to decide between
    \begin{enumerate}
        \item $G$ is $3$-colorable,
        \item $G$ has no independent set of size larger than $\gamma n$.
    \end{enumerate}
\end{proposition}

In particular, the gadget used to prove \Cref{prop:fish-3-colorable} is regular, hence we can additionally assume that the graph is regular.
Moreover, we can assume that the degrees are $o(n)$.

With \Cref{prop:fish-3-colorable}, we can prove the following:

\begin{proposition} \label{prop:hardness-6-colorable}
    Assuming \Cref{conj:fish}, for any constant $\gamma > 0$, given a regular $n$-vertex graph $G$ which is a one-sided spectral expander with $\lambda_2(G) \leq o_n(1)$, it is NP-hard to decide between
    \begin{enumerate}
        \item $G$ is $6$-colorable,
        \item $G$ has no independent set of size larger than $\gamma n$.
    \end{enumerate}
\end{proposition}
\begin{proof}
    Given a graph $G$,
    the reduction is to add a regular bipartite graph $H$ (potentially introducing multi-edges) such that $H$ has degree $\Omega(n)$ and the second eigenvalue of its normalized adjacency matrix $\lambda_2(H) = o_n(1)$.
    If $G$ is $k$-colorable, then the resulting graph $G'$ is clearly $2k$-colorable (since $H$ is $2$-colorable).
    On the other hand, adding edges cannot increase the size of the maximum independent set.

    Next, we prove that $G'$ has small normalized second eigenvalue.
    Since we can assume that $G$ has degree $d_G = o(n)$ while $H$ has degree $d_H = \Omega(n)$.
    Then, $\lambda_2(G') \leq \frac{1}{d_G+d_H} (d_G \lambda_2(G) + d_H \lambda_2(H)) \leq o_n(1)$.
\end{proof}

%% file: KMS.tex
\section{Rounding Independent Sets via Karger-Motwani-Sudan}
\label{sec:kms}

In this section, we recall a folklore result (we were unable to find a reference, though this argument seems to be known to experts) that extends the rounding algorithm of Karger, Motwani and Sudan~\cite{KMS98} to prove the following:

\begin{theorem} \label{thm:large-IS}
    For any $\eps > 0$, there exists a polynomial-time algorithm such that given an $n$-vertex graph containing an independent set of size $(1/2-\eps)n$, it finds an independent set of size at least $(\eps n)^{1-O(\eps)}$.
\end{theorem}

We first prove the following crucial lemma.

\begin{lemma} \label{lem:polarized-vectors}
    Let $G = (V,E)$ be a graph and $\eps > 0$.
    Suppose each vertex $i\in V$ is associated with a unit vector $u_i$ such that for all $(i,j)\in E$, $\angles{u_i, u_j} \leq -1+\eps$.
    Then, there is a polynomial-time algorithm that finds an independent set in $G$ of size at least $n^{1-O(\eps)}$.
\end{lemma}
\begin{proof}
    Set $t \coloneqq 4\sqrt{\eps \log n}$.
    The algorithm is as follows,
    \begin{enumerate}[(1)]
        \item Sample a Gaussian vector $g \sim \calN(0, \Id_n)$.
        \item Let $S \coloneqq \{i\in V: \angles{g, u_i} \geq t\}$.
        \item Output $T \coloneqq \{i\in S: \forall j\in N(i),\ j\notin S\}$.
    \end{enumerate}
    Here, $N(i)$ denotes the set of neighbors of $i$.
    By definition, $T$ is an independent set.
    We next claim that in expectation over $g$, $|T| \geq n^{1-O(\eps)}$, which finishes the proof.

    First, note that $\Pr[i\in S] = \Pr_{h\sim \calN(0,1)}[h \geq t] \geq \Omega(\frac{1}{t} e^{-t^2/2}) \geq n^{-O(\eps)}$.
    Next, for each $i \in V$,
    \begin{align*}
        \Pr\bracks*{i\in S \text{ and } \forall j \in N(i),\ j\notin S}
        &= \Pr[i\in S] \cdot \parens*{1 - \Pr\bracks*{
        \exists j\in N(i),\ j\in S| i\in S}} \\
        &\geq \Pr[i\in S] \cdot \parens*{1 - \sum_{j \in N(i)} \Pr[j\in S|i\in S] } \mper
    \end{align*}
    where the second inequality follows by union bound.

    We now analyze $\Pr[j\in S|i\in S]$.
    Since $\angles{u_i, u_j} \leq -1+\eps$, we can write $u_j = \alpha u_i + \beta w$, where $w \perp u_i$, $-1 \leq \alpha \leq -1+\eps$ and $\beta = \sqrt{1-\alpha^2} \leq \sqrt{2\eps}$.
    Then, $j\in S$ means that $\angles{g, u_j} = \alpha \angles{g, u_i} + \beta \angles{g, w} \geq t$, and combined with $\angles{g, u_i} \geq t$, we have $\angles{g, w} \geq (1-\alpha) t/\beta \geq t/\beta$.
    Thus,
    \begin{equation*}
        \Pr[j\in S|i\in S] \leq \Pr[\angles{g, w} \geq t/\beta]
        \leq e^{-t^2/2\beta^2}
        \leq 1/n^2 \mcom
    \end{equation*}
    since $\beta \leq \sqrt{2\eps}$ and $t = 4\sqrt{\eps \log n}$.
    As $i$ has at most $n$ neighbors, this implies that $\Pr[i\in S \text{ and } \forall j 
    \in N(i),\ j\notin S] \geq \Pr[i\in S] \cdot (1-o(1))$.
    In particular, we have $\E|T| \geq n \cdot \Pr[i\in S] \cdot (1-o(1)) \geq n^{1-O(\eps)}$, completing the proof.
\end{proof}

We now prove \Cref{thm:large-IS}.

\begin{proof}[Proof of \Cref{thm:large-IS}]
    Consider the following independent set formulation:
    \begin{equation*}
    \begin{aligned}
        \max \quad &\sum_{i \in V} x_i \\
        \text{s.t.} \quad & (1+x_i) (1+x_j) = 0 \quad \forall (i,j) \in E(G), \\
        & x_i^2 = 1 \quad \forall i\in V(G) \mper
    \end{aligned}
    \numberthis \label{eq:pm1-IS}
    \end{equation*}
    Note that any vector $x\in \pmo^n$ where $\{i: x_i=1\}$ is an independent set in $G$ is a feasible solution to the above, since $(1+x_i)(1+x_j)$ is nonzero only if $x_i = x_j = 1$.
    Since $G$ has an independent set of size $(1/2-\eps)n$, the above program has value at least $(1/2-\eps)n - (1/2+\eps)n = -2\eps n$.

    We can solve the SDP relaxation of (\ref{eq:pm1-IS}) and obtain a pseudo-distribution $\mu$, and we have that $\sum_{i\in V} \pE_\mu[x_i] \geq -2\eps n$.
    Let $S \coloneqq \{i: \pE_\mu[x_i] \geq -4\eps\}$.
    Then,
    \begin{equation*}
        -2\eps n \leq \sum_{i\in S} \pE_\mu[x_i] + \sum_{i\notin S} \pE_\mu[x_i]
        \leq |S| + (n-|S|) \cdot (-4\eps)
        \implies |S| \geq \frac{2\eps n}{1+4\eps} \geq \eps n \mper
    \end{equation*}
    For any $i\sim j\in S$, we have $\pE_{\mu}[x_i x_j] = -1 - \pE_\mu[x_i + x_j] \leq -1 + 8\eps$.
    Moreover, each vertex $i$ is associated with a unit vector $u_i$ such that $\angles{u_i, u_j} = \pE_\mu[x_ix_j]$.
    Thus, the subgraph $G[S]$ and the unit vectors satisfy the conditions in \Cref{lem:polarized-vectors}.
    Thus, there is a polynomial-time algorithm that finds an independent set in $G[S]$ of size at least $|S|^{1-O(\eps)} \geq (\eps n)^{1-O(\eps)}$.
\end{proof}

%% file: bm-lemmas.tex
\section{Lemmas from \texorpdfstring{\cite{BafnaMinzer}}{[BM23]}}
\label{sec:lemmas-bm}
In this section, we modify two lemmas from \cite{BafnaMinzer} which are useful for the analysis of our rounding algorithm. A variant of the following lemma appears as Lemma 3.24 in~\cite{BafnaMinzer}.

\begin{lemma}\label{lem:gc-reduction-restated}
For all $\tau,p \in (0,1)$ and $D \in \Z$,  the following holds: Suppose there is a degree $D + \Omega(t^2/p\tau)$ pseudo-distribution $\mu$ over independent sets that satisfies $\cA_I$, and a polynomial $E(\bx)$ satisfying $\pE_{\bx \sim \mu^{\otimes t}}[E(\bx)] \geq p$ and $\cA_{IS} \vdash_D E(\bx) \leq 1$. Then there exist subsets $A_1, \ldots, A_t \subseteq V(G)$ of size at most $O(\frac{t}{p\tau})$ and strings $y_1,\ldots, y_t$ such that conditioning $\mu$ on the events $x|_{A_1} = y_1,\ldots, x|_{A_t} = y_t$ gives pseudo-distributions $\mu_1,\ldots,\mu_t$ of degree at least $D'$ such that: 
\begin{enumerate}
\item $\pE_{\bx \sim \mu_1\times\ldots \times \mu_t}[E(\bx)] \geq \frac{p}{2}$.
\item For all $i \in [t]$, $\E_{\substack{u,v \sim V(G)}}[I_{\mu_i}(x_{u};x_v)] \leq \tau.$
\end{enumerate}
\end{lemma}

\begin{proof}
Applying \Cref{lem:ragh-tan} to $\mu$ we get that for all $\tau' > 0$, conditioning on $r = O(\frac{1}{\tau'})$ variables gives:
\begin{equation}\label{eq:rt-1}
\E_{i_1,\ldots,i_r \sim V} \E_{a,b \sim V} [I_\mu(x_a ; x_b \mid x_{i_1}, \ldots,x_{i_r})] \leq \tau' \mper
\end{equation}
By expanding the definition of conditional mutual information, we get:
\begin{align*}
I(x_a;x_b \mid x_{i_1}, \ldots, x_{i_r}) &= \E_{(y_{i_1}, \ldots,y_{i_r}) \sim \mu}[I(x_a ; x_b \mid x_{i_1} = y_{i_1}, \ldots,x_{i_r} = y_{i_r})]. 
\end{align*}
Plugging the above into \Cref{eq:rt-1} and applying Markov's inequality, we get that for all $\alpha \in (0,1)$:
\begin{equation}\label{eq:avg1}
\Pr_{\substack{i_1,\ldots,i_r \sim V \\ y_{i_1},\ldots,y_{i_r} \sim \mu}}\left[\E_{a,b \sim V}[I_\mu(x_a ; x_b \mid x_{i_1} = y_{i_1},\ldots)] \geq \frac{\tau'}{\alpha} \right] \leq \alpha.
\end{equation}
For later, let us write down the equation above in a more convenient way, for ``different copies'' of $\mu$. That is, for all $i \in [t]$:
\begin{equation}\label{eq:avg2}
\Pr_{\substack{A_i \subset_r V \\ y_i \sim \mu|_{A_r}}}\left[\E_{a,b \sim V}[I_\mu(x^{(i)}_a ; x^{(i)}_b \mid x^{(i)}|_{A_i} = y_i)] \geq \frac{\tau'}{\alpha} \right] \leq \alpha.
\end{equation}
Recall the following expression of conditional expectations for a polynomial:
\begin{align*}
\pE_{\mu^{\otimes t}}[E(\bx) \mid x^{(1)}|_{A_1} = y_1,\ldots, x^{(t)}|_{A_t} = y_t] = \frac{\pE_{\mu^{\otimes t}}[E(X,X')\Ind[x^{(1)}|_{A_1} =y_1,\ldots]]}{\pPr_{\mu^{\otimes t}}[x^{(1)}|_{A_1} = y_1,\ldots]},
\end{align*}
where we've used $\Ind[\cdot]$ to denote the unique polynomial corresponding to the event $x^{(1)}|_{A_1} =y_1,\ldots$ Analogous to the definition of conditional expectation we can check that:
\[\pE_{\mu^{\otimes t}}[E(\bx)] = \E_{\substack{A_1 \subset_r V,\ldots, A_t \subset_r V \\ y_1 \sim \mu|_{A_1},\ldots, y_t \sim \mu|_{A_t}}}[\pE_{\mu^{\otimes t}}[E(\bx) \mid x^{(1)}|_{A_1} = y_1,\ldots]].\]
Since $\cA_{IS} \vdash_D E(\bx) \leq 1$ and $\deg(\mu) \geq D + \Omega(rt)$ we get that $E(\bx) \leq 1$ even after conditioning on any non-negative event $Q$ of degree $O(rt)$: $\pE[E(\bx) \mid Q] \leq 1$. An averaging argument implies that:
\begin{equation}\label{eq:avg3}
\Pr_{\substack{A_1 \subset_r V,\ldots, A_t \subset_r V \\ y_1 \sim \mu|_{A_1},\ldots, y_t \sim \mu|_{A_t}}}\left[\pE_{\mu^{\otimes t}}[E(\bx) \mid x^{(1)}|_{A_1} = y_1,\ldots] \geq \frac{p}{2}\right] \geq \frac{p}{2}.
\end{equation}

Choosing $\alpha = p/8t$ and $\tau' < p\tau/8t$ (i.e. $r \geq \Omega(t/p\tau)$) we can take a union bound over the events in equations~\eqref{eq:avg2} (over all $t$ copies), and~\eqref{eq:avg3} to get that there exist sets $A_1,\ldots,A_t \in V$ and strings $y_1,\ldots,y_t$ such that,   
\[\pE_{\mu^{\otimes t}}[E(\bx) \mid x^{(1)}|_{A_1} = y_1,\ldots] \geq \frac{p}{2},\]
and for each $i \in [t]$,
\[\E_{a,b \sim V}[I_\mu(x^{(i)}_a ; x^{(i)}_b \mid x^{(i)}|_{A_i} = y_i] \leq \frac{\tau'}{\alpha} < \tau.\]

Let $\mu_i$ be the pseudo-distribution on $x^{(i)}$ that we get by conditioning $\mu$ on $x^{(i)}|_{A_i} = y_i$. The above gives us the properties we need in the lemma. 
\end{proof}

A variant of the following lemma appears as Lemma 3.25 in~\cite{BafnaMinzer}.

\begin{lemma}\label{lem:indep-after-conditioning-restated}
For all $\tau,\alpha,\nu \in (0,1)$, $D \in \N$, the following holds: Let $\mu = \mu_1 \times \ldots \times \mu_t$ be a degree $D+\Omega(t)$ pseudo-distribution over $\bx =(x^{(1)},\ldots,x^{(t)})$ satisfying $\calA_{IS}(\bx)$ and $E(\bx)$ be a polynomial such that, $\calA_{IS} \vdash_D E(\bx) \in [0,1]$ and $\pE_{\mu}[E(\bx)] \geq \alpha$. Suppose we have that, 
for all $i \in [t]$, $\E_{u,v \in V}[I_{\mu_i}(x^{(i)}_u;x^{(i)}_v)] \leq \tau.$
Then we get that conditioning on $E$ preserves independence for most vertices: 
\[\Pr_{u \in V}[ TV((x_u^{(1)},\ldots,x_u^{(t)})|E , (x_u^{(1)},\ldots,x_u^{(t)})) \geq \nu] \leq O\left(\frac{\sqrt{t\tau}}{\alpha\nu^2}\right),\]
where the distribution $(x_u^{(1)},\ldots)$ is the marginal from $\mu$ and $(x_u^{(1)},\ldots)|E$ refers to the marginal from the reweighted distribution $\mu | E$. 
\end{lemma}

The proof of \Cref{lem:indep-after-conditioning-restated} requires the following standard fact.

\begin{fact}[Data processing inequality]
    Let $X,Y,A,B$ be random variables such that $H(A|X) = 0$ and $H(B|Y) = 0$, i.e. $A$ is fully determined by $X$ and $B$ is fully determined by $Y$. Then,
    \[I(A;B) \leq I(X;Y) \mper \]
\end{fact}

\begin{proof}[Proof of \Cref{lem:indep-after-conditioning-restated}]
Let $\mu$ denote the pseudo-distribution $\mu_1 \times \ldots \times \mu_t$, $Y_u$ denote the random variable $(x^{(1)}_u,\ldots, x_u^{(t)})$ drawn from $\mu$ and $Y_u | E$ denote the random variable drawn from the conditioned/reweighted pseudo-distribution $\mu~|~E$. 

Let $U$ be the set of variables $u \in V(G)$ for which $TV(Y_u|E , Y_u) \geq \delta$ and let the fractional size of $U$ be $\gamma$. For every $u \in U$ there exists a set $S_u \subseteq \{0,1\}^t$ such that: 
\begin{equation}\label{eq:e-tv}
\pPr_{\mu}[Y_u \in S_u \mid E] - \pPr_{\mu}[Y_u \in S_u] \geq \nu.
\end{equation}

For simplicity of notation we think of ``events'' on the random variables $Y$ and use $\Ind(Y_u \in S_u)$ to denote the unique degree $\leq t$ polynomial for this function. Let $e_u$ denote $\pPr_{\mu}[Y_u \in S_u]$. Define the random variables $Z_u = \Ind(Y_u \in S_u) - e_u$. Define:
\[Z = \E_{u \in U}[Z_u] = \E_{u \in U}[\Ind(Y_u \in S_u) - e_u].\]
One can check that $\pE_{\mu}[Z] = \pE_{\mu}[Z_u] = 0$, and we now calculate its variance. For two events $A,B$ on the variables $Y$ let $\pCov_\mu(A,B)$ denote $\pE_{\mu}[AB] - \pE_{\mu}[A]\pE_{\mu}[B]$. Firstly for all $u,v \in V(G)$  using Pinsker's inequality and the data processing inequality we have that,
\begin{align*}
    \pCov_\mu(\Ind(Y_u \in S_u),\Ind(Y_v \in S_v)) 
    &\leq TV((Y_u,Y_v), Y_u\times Y_v) \\
    &\leq O\parens*{\sqrt{I_\mu(Y_u;Y_v)}} \\
    &= O(1) \sqrt{\sum_{i \in [t]} I_{\mu_i}(x^{(i)}_u;x^{(i)}_v)} \mper
\end{align*}

The proof will proceed by proving upper and lower bounds on $\pE_{\mu}[Z^2]$, where the upper bound uses low global correlation properties of $\mu$ and the lower bound uses the large deviation we have by \Cref{eq:e-tv}. 

\parhead{Upper bound for $\pE_{\mu}[Z^2]$:}
We have the following upper bound:
\begin{align*}
\pE_{\mu}[Z^2] &= \E_{u,v \sim U} \bracks*{ \pE_{\mu}[(\Ind(Y_u \in S_u) - e_u)(\Ind(Y_v \in S_v) - e_v)] } \\
&= \E_{u,v \sim U}\bracks*{ \pE_{\mu}[\Ind(Y_u \in S_u) - e_u]\pE_{\mu}[\Ind(Y_v \in S_v) - e_v] + \pCov_\mu(\Ind(Y_u \in S_u),\Ind(Y_v \in S_v)) } \\
&\leq \E_{u,v \sim U} \bracks*{ O(1) \sqrt{\sum_{i \in [t]} I(x^{(i)}_u;x^{(i)}_v)} } \\
&\leq O(1) \sqrt{\E_{u,v \sim U} \sum_{i \in [t]}I(x^{(i)}_u;x^{(i)}_v)} \\
&\leq O \parens*{ \frac{\sqrt{t\tau}}{\gamma} } \mcom
\end{align*}
where the last inequality follows because $\E_{u,v \sim V}[I(x^{(i)}_u;x^{(i)}_v)] \geq \gamma^2 \E_{u,v \sim U}[I(x^{(i)}_u;x^{(i)}_v)]$, and by assumption $\E_{u,v \sim V}[I(x^{(i)}_u;x^{(i)}_v)] \leq \tau$ for all $i \in [t]$.

\parhead{Lower bound for $\pE_\mu[Z^2]$:}
Since $\cA_{IS} \vdash_D E(\bx) \in [0,1]$, we have that,
\begin{align*}
    \pE_\mu[Z^2] &= \pE[E(\bx)]\pE_\mu[Z^2 | E(x)]+\pE[1-E(\bx)]\pE_\mu[Z^2 | 1-E(\bx)]  \\
    &\geq \pE[E(\bx)]\pE_\mu[Z^2 | E(x)] \geq \alpha\pE_\mu[Z | E]^2
\end{align*}
where in the last inequality we used Cauchy-Schwarz for the polynomial $Z$ with respect to the reweighted pseudo-distribution $\mu | E$.
Using \Cref{eq:e-tv} we know that for all $u \in U$:
\[\pE_{\mu}[Z_u \mid E] = \pPr_{\mu}[Y_u \in S_u \mid E] - e_u  \geq \nu,\]
which implies that $\pE_{\mu}[Z \mid E] \geq \nu$.
Combining the upper and lower bounds on $\pE_{\mu}[Z^2]$ we get that $\gamma \leq O\left(\frac{\sqrt{t\tau}}{\alpha\nu^2}\right)$,
completing the proof of the lemma.
\end{proof}

%% file: sqrt-approx.tex
\section{Polynomial Approximation for \texorpdfstring{$\sqrt{x}$}{sqrt(x)}}
\label{sec:poly-approx}

The Bernstein polynomials have been widely used to approximate continuous functions (see e.g.\ \cite{Lor53,DL93} for an exposition).
They were first used in a constructive proof for the Weierstrass approximation theorem, and they have since been applied in many fields.
An example closely related to our work is the SoS proof of the ``Majority is Stablest Theorem'' in \cite{DMN13}.

\begin{definition}[Bernstein polynomial] \label{def:bernstein}
    Let $f: [0,1] \to \R$ be any function. For any $m \in \N$, define $\Bernstein_m f$ to be the following degree-$m$ polynomial,
    \begin{equation*}
        (\Bernstein_m f)(x) = \sum_{k=0}^m f(k/m) \binom{m}{k} x^k (1-x)^{m-k} \mper
    \end{equation*}
\end{definition}

The following are some standard facts (see e.g.\ \cite{Farouki12}).

\begin{fact} \label{fact:bernstein-derivative}
    For any function $f$, the derivative of $(\Bernstein_m f)(x)$ is 
    \begin{equation*}
        \frac{d}{dx}(\Bernstein_m f)(x) = \sum_{k=0}^{m-1} m \Paren{ f\Paren{\frac{k+1}{m}} - f\Paren{\frac{k}{m}}} \cdot \binom{m-1}{k} x^k (1-x)^{m-1-k}
        = (\Bernstein_{m-1} g)(x) \mcom
    \end{equation*}
    where $g(x) \coloneqq m \cdot (f(x + \frac{1}{m}) - f(x))$.
    Consequently, if $f(x)$ is increasing on $[0,1]$, then $(\Bernstein_{m} f)(x)$ is also increasing on $[0,1]$.
\end{fact}

\begin{fact} \label{fact:bernstein-convex}
    If $f$ is convex on $[0,1]$, then $(\Bernstein_m f)(x)$ is also convex and $(\Bernstein_m f)(x) \geq f(x)$ for all $x\in[0,1]$.
    Similarly, if $f$ is concave, then $(\Bernstein_m f)(x)$ is also concave and $(\Bernstein_m f)(x) \leq f(x)$ for all $x\in[0,1]$.
\end{fact}

The approximation errors of Bernstein polynomials have been well studied.
The following fact follows from, for e.g., Theorem 1 of \cite{Mathe99} applied to the square root function.

\begin{fact}[Approximation error] \label{fact:bernstein-error}
    Let $f(x) = \sqrt{x}$. For all $m\in \N$ and $x\in [0,1]$,
    \begin{equation*}
        \Abs{ (\Bernstein_m f)(x) - \sqrt{x}} \leq \Paren{\frac{x(1-x)}{m}}^{1/4} \mper
    \end{equation*}
\end{fact}

We now use the Bernstein polynomial to define a proxy for the $\sqrt{x}$ function over $x \in [0, n]$, and in the following lemma we prove the requirements we need.

\begin{lemma}[Polynomial proxy for $\sqrt{x}$; restatement of \Cref{lem:proxy}] \label{lem:bernstein-proxy}
    Let $n,m \in \N$ such that $m \geq 64n$.
    Let $f(x) = \sqrt{x}$.
    Define the univariate polynomial
    \begin{equation*}
        B_m(x) \coloneqq \sqrt{n} \cdot (\Bernstein_m f)(x/n) \mper
    \end{equation*}
    Then, $B_m$ satisfies the following properties:
    \begin{enumerate}[(1)]
        \item $0 \leq B_m(x^2) \leq x$ for all $x\in [0, n]$.
        \label{item:Bx2-x}

        \item $B_m(x) \leq B_m(y)$ for $x \leq y \in [0,n]$.
        \label{item:B-increasing}

        \item $B_m(x^2) \geq \frac{1}{2} x$ for $x \in [1/2, n]$.
        \label{item:B-x2-geq-x}

        \item $B_m(x) \geq \frac{1}{2} x$ for $x\in [0,1]$.
        \label{item:B-geq-x}

        \item For any Boolean vector $u, v\in \zo^n$, $\iprod{u,v} \leq 4 \cdot B_m(\sum_i u_i) B_m(\sum_i v_i)$.
        \label{item:c-s}

        \item For any $d\in \N$ and any Boolean vector $u\in \zo^n$, $B_m(\frac{1}{d} \sum_i u_i) \leq \frac{2}{\sqrt{d}} \cdot B_m(\sum_i u_i)$.
        \label{item:B-half}

        \item For any $N\in \N$ and probabilities $p_1,p_2,\dots,p_N$ such that $\sum_{i=1}^N p_i = 1$, if $x_i \in [0,n]$ then $\sum_{i=1}^N p_i B_m(x_i) \leq B_m(\sum_{i=1}^N p_i x_i)$.
        \label{item:concave}

    \end{enumerate}
    Moreover, all of the above have SoS proofs of degree $O(m)$.
\end{lemma}

\begin{proof}
    \ref{item:Bx2-x} and \ref{item:B-increasing} follow from \Cref{fact:bernstein-derivative,fact:bernstein-convex} that $(\Bernstein_m f)(x) \leq f(x)$ and $(\Bernstein_m f)(x)$ is increasing since $f$ is increasing and concave.

    \ref{item:B-x2-geq-x}: By \Cref{fact:bernstein-error}, $B_m(x^2) \geq \sqrt{n} \cdot ( \sqrt{x^2/n} - (\frac{x^2}{nm})^{1/4}) = x - (\frac{nx^2}{m})^{1/4}$.
    When $m \geq 64n$, for all $x \geq 1/2$ we have $(\frac{nx^2}{m})^{1/4} \leq \frac{1}{2}x$.

    \ref{item:B-geq-x}:
    By \Cref{fact:bernstein-derivative}, $\frac{d}{dx} B_m(x) = \frac{1}{\sqrt{n}} \cdot (\Bernstein_{m-1} g)(x/n)$ where $g(x) = m \cdot (f(x+ \frac{1}{m}) - f(x))$.
    One can verify that $g$ is decreasing and convex, thus by \Cref{fact:bernstein-convex}, $(\Bernstein_{m-1}g)(x) \geq g(x)$ for all $x\in [0,1]$.
    Moreover, since $f$ is concave, $g(x) \geq f'(x+\frac{1}{m}) = \frac{1}{2\sqrt{x+1/m}}$.
    Thus, as $g$ is decreasing, for all $x \leq \frac{1}{2}$, $\frac{d}{dx} B_m(x) \geq \frac{1}{\sqrt{n}} \cdot g(x/n) \geq \frac{1}{\sqrt{n}} \cdot g(\frac{1}{2n}) > \frac{1}{2}$ by our choice of $m$.
    Thus, $B_m(x) \geq \frac{1}{2}x$ for $x \in [0, 1/2]$.

    For $x\in [1/2, 1]$, we simply use \ref{item:B-x2-geq-x}: $B_m(x) \geq \frac{1}{2}\sqrt{x} \geq \frac{1}{2}x$.

    \ref{item:c-s}: by Cauchy-Schwarz and Booleanity of $u,v$ we have $\angles{u,v} \leq \|u\|_2 \|v\|_2 = \sqrt{\sum_i u_i} \cdot \sqrt{\sum_i v_i}$.
    Then, since $\sum_i u_i \leq n$, by \ref{item:B-x2-geq-x} we have $\sqrt{\sum_i u_i} \leq 2 \cdot B_m(\sum_i u_i)$, and similarly $\sqrt{\sum_i v_i} \leq 2 \cdot B_m(\sum_i v_i)$.

    \ref{item:B-half}: by \ref{item:Bx2-x} and \ref{item:B-x2-geq-x}, we have $B_m(\frac{1}{d}\sum_i u_i) \leq \sqrt{\frac{1}{d}\sum_i u_i} \leq \frac{2}{\sqrt{d}} B_m(\sum_i u_i)$.

    \ref{item:concave} follows directly from concavity of $B_m$ (\Cref{fact:bernstein-convex}).

    Finally, all statements except \ref{item:B-increasing} and \ref{item:concave} are univariate inequalities or inequalities over the Boolean hypercube.
    Thus, by \Cref{fact:univariate-interval,fact:boolean-function}, they all exhibit SoS proofs of degree $O(m)$.
    For \ref{item:B-increasing} and \ref{item:concave}, we use the result from \cite{AP14} that monotonicity and convexity/concavity of univariate polynomials have SoS proofs.
\end{proof}


    





%% file: main.bbl
\newcommand{\etalchar}[1]{$^{#1}$}
\begin{thebibliography}{GMPT10}

\bibitem[ABS10]{ABS10}
Sanjeev Arora, Boaz Barak, and David Steurer.
\newblock {Subexponential Algorithms for Unique Games and Related Problems}.
\newblock In {\em Proceedings of the 2010 IEEE 51st Annual Symposium on
  Foundations of Computer Science}, pages 563--572, 2010.

\bibitem[ABS15]{ABS15}
Sanjeev Arora, Boaz Barak, and David Steurer.
\newblock {Subexponential Algorithms for Unique Games and Related Problems}.
\newblock {\em Journal of the ACM (JACM)}, 62(5):1--25, 2015.

\bibitem[ACC06]{ACC06}
Sanjeev Arora, Eden Chlamtac, and Moses Charikar.
\newblock New approximation guarantee for chromatic number.
\newblock In {\em 38th Annual ACM Symposium on Theory of Computing, STOC'06},
  pages 215--224, 2006.

\bibitem[AG11]{AroraG11}
Sanjeev Arora and Rong Ge.
\newblock {New Tools for Graph Coloring}.
\newblock In {\em International Workshop on Approximation Algorithms for
  Combinatorial Optimization}, pages 1--12. Springer, 2011.

\bibitem[AK97]{AlonK97}
Noga Alon and Nabil Kahale.
\newblock {A Spectral Technique for Coloring Random 3-Colorable Graphs}.
\newblock {\em SIAM Journal on Computing}, 26(6):1733--1748, 1997.

\bibitem[AK98]{AlonK98}
Noga Alon and Nabil Kahale.
\newblock {Approximating the independence number via the $\vartheta$-function}.
\newblock {\em Mathematical Programming}, 80(3):253--264, 1998.

\bibitem[AKK{\etalchar{+}}08]{AroraKKSTV08}
Sanjeev Arora, Subhash~A Khot, Alexandra Kolla, David Steurer, Madhur Tulsiani,
  and Nisheeth~K Vishnoi.
\newblock Unique games on expanding constraint graphs are easy.
\newblock In {\em Proceedings of the fortieth annual ACM symposium on Theory of
  computing}, pages 21--28, 2008.

\bibitem[AP13]{AP14}
Amir~Ali Ahmadi and Pablo~A Parrilo.
\newblock {A complete characterization of the gap between convexity and
  SOS-convexity}.
\newblock {\em SIAM Journal on Optimization}, 23(2):811--833, 2013.

\bibitem[AR06]{AR06}
Dimitris Achlioptas and Federico {Ricci-Tersenghi}.
\newblock On the solution-space geometry of random constraint satisfaction
  problems.
\newblock In {\em Proceedings of the thirty-eighth annual ACM symposium on
  Theory of computing}, pages 130--139, 2006.

\bibitem[BBK{\etalchar{+}}21]{BBKSS}
Mitali Bafna, Boaz Barak, Pravesh~K. Kothari, Tselil Schramm, and David
  Steurer.
\newblock {Playing Unique Games on Certified Small-Set Expanders}.
\newblock In {\em {STOC} '21: 53rd Annual {ACM} {SIGACT} Symposium on Theory of
  Computing, Virtual Event, Italy, June 21-25, 2021}, pages 1629--1642. {ACM},
  2021.

\bibitem[BBKS24]{BBKS24}
Jaroslaw Blasiok, Rares{-}Darius Buhai, Pravesh~K. Kothari, and David Steurer.
\newblock Semirandom planted clique and the restricted isometry property.
\newblock {\em CoRR}, abs/2404.14159, 2024.

\bibitem[BH92]{BoppanaH92}
Ravi Boppana and Magn{\'u}s~M Halld{\'o}rsson.
\newblock {Approximating maximum independent sets by excluding subgraphs}.
\newblock {\em BIT Numerical Mathematics}, 32(2):180--196, 1992.

\bibitem[BK97]{BlumK97}
Avrim Blum and David Karger.
\newblock {An $\tilde{O}(n^{3/14})$-coloring algorithm for 3-colorable graphs}.
\newblock {\em Information processing letters}, 61(1):49--53, 1997.

\bibitem[BK09]{BansalK09}
Nikhil Bansal and Subhash Khot.
\newblock {Optimal Long Code Test with One Free Bit}.
\newblock In {\em 2009 50th Annual IEEE Symposium on Foundations of Computer
  Science}, pages 453--462. IEEE, 2009.

\bibitem[BKS17]{BKS17}
Boaz Barak, Pravesh~K. Kothari, and David Steurer.
\newblock {Quantum entanglement, sum of squares, and the log rank conjecture}.
\newblock In {\em Proceedings of the 49th Annual {ACM} {SIGACT} Symposium on
  Theory of Computing, {STOC} 2017, Montreal, QC, Canada, June 19-23, 2017},
  2017.

\bibitem[BKS23]{BuhaiKS23}
Rares-Darius Buhai, Pravesh~K Kothari, and David Steurer.
\newblock {Algorithms approaching the threshold for semi-random planted
  clique}.
\newblock In {\em Proceedings of the 55th Annual ACM Symposium on Theory of
  Computing}, pages 1918--1926, 2023.

\bibitem[Blu94]{Blum94}
Avrim Blum.
\newblock New approximation algorithms for graph coloring.
\newblock {\em Journal of the ACM (JACM)}, 41(3):470--516, 1994.

\bibitem[BM23]{BafnaMinzer}
Mitali Bafna and Dor Minzer.
\newblock Solving unique games over globally hypercontractive graphs.
\newblock {\em CoRR}, abs/2304.07284, 2023.

\bibitem[BRS11]{BRS11}
Boaz Barak, Prasad Raghavendra, and David Steurer.
\newblock {Rounding Semidefinite Programming Hierarchies via Global
  Correlation}.
\newblock In {\em {IEEE} 52nd Annual Symposium on Foundations of Computer
  Science, {FOCS} 2011, Palm Springs, CA, USA, October 22-25, 2011}, pages
  472--481. IEEE, 2011.

\bibitem[BS95]{BlumS95}
Avrim Blum and Joel Spencer.
\newblock Coloring random and semi-random k-colorable graphs.
\newblock {\em Journal of Algorithms}, 19(2):204--234, 1995.

\bibitem[BS16]{BS16}
Boaz Barak and David Steurer.
\newblock {Proofs, beliefs, and algorithms through the lens of sum-of-squares}.
\newblock {\em Course notes:
  \url{http://www.sumofsquares.org/public/index.html}}, 2016.

\bibitem[Chl09]{Chlamtac09}
Eden Chlamtac.
\newblock {\em {Non-local analysis of SDP-based approximation algorithms}}.
\newblock Princeton University, 2009.

\bibitem[CSV17]{CharikarSV17}
Moses Charikar, Jacob Steinhardt, and Gregory Valiant.
\newblock Learning from untrusted data.
\newblock In {\em Proceedings of the 49th Annual ACM SIGACT Symposium on Theory
  of Computing}, pages 47--60, 2017.

\bibitem[DF16]{DavidF16}
Roee David and Uriel Feige.
\newblock {On the effect of randomness on planted 3-coloring models}.
\newblock In {\em Proceedings of the forty-eighth annual ACM symposium on
  Theory of Computing}, pages 77--90, 2016.

\bibitem[DGJ{\etalchar{+}}10]{diakonikolas2010bounded}
Ilias Diakonikolas, Parikshit Gopalan, Ragesh Jaiswal, Rocco~A Servedio, and
  Emanuele Viola.
\newblock Bounded independence fools halfspaces.
\newblock {\em SIAM Journal on Computing}, 39(8):3441--3462, 2010.

\bibitem[DHV16]{DeshpandeHV16}
Amit Deshpande, Prahladh Harsha, and Rakesh Venkat.
\newblock {Embedding Approximately Low-Dimensional $\ell_2^2$ Metrics into
  $\ell_1$}.
\newblock In {\em 36th IARCS Annual Conference on Foundations of Software
  Technology and Theoretical Computer Science (FSTTCS 2016)}. Schloss
  Dagstuhl-Leibniz-Zentrum fuer Informatik, 2016.

\bibitem[DKPS10]{DinurKPS10}
Irit Dinur, Subhash Khot, Will Perkins, and Muli Safra.
\newblock {Hardness of Finding Independent Sets in Almost 3-Colorable Graphs}.
\newblock In {\em 2010 IEEE 51st Annual Symposium on Foundations of Computer
  Science}, pages 212--221. IEEE, 2010.

\bibitem[DL93]{DL93}
Ronald~A DeVore and George~G Lorentz.
\newblock {\em Constructive approximation}, volume 303.
\newblock Springer Science \& Business Media, 1993.

\bibitem[DMN13]{DMN13}
Anindya De, Elchanan Mossel, and Joe Neeman.
\newblock {Majority is stablest: discrete and SoS}.
\newblock In {\em Proceedings of the forty-fifth annual ACM symposium on Theory
  of computing}, pages 477--486, 2013.

\bibitem[DMR06]{DinurMR06}
Irit Dinur, Elchanan Mossel, and Oded Regev.
\newblock {Conditional Hardness for Approximate Coloring}.
\newblock In {\em Proceedings of the thirty-eighth annual ACM symposium on
  Theory of Computing}, pages 344--353, 2006.

\bibitem[DS05]{DinurS05}
Irit Dinur and Samuel Safra.
\newblock {On the hardness of approximating minimum vertex cover}.
\newblock {\em Annals of mathematics}, pages 439--485, 2005.

\bibitem[EG20]{EG20}
Ronen Eldan and Renan Gross.
\newblock {Concentration on the Boolean hypercube via pathwise stochastic
  analysis}.
\newblock In {\em Proceedings of the 52nd Annual ACM SIGACT Symposium on Theory
  of Computing}, pages 208--221, 2020.

\bibitem[EKLM22]{EKLM22}
Ronen Eldan, Guy Kindler, Noam Lifshitz, and Dor Minzer.
\newblock {Isoperimetric Inequalities Made Simpler}.
\newblock {\em arXiv preprint arXiv:2204.06686}, 2022.

\bibitem[Far12]{Farouki12}
Rida~T Farouki.
\newblock {The Bernstein polynomial basis: A centennial retrospective}.
\newblock {\em Computer Aided Geometric Design}, 29(6):379--419, 2012.

\bibitem[Fei04]{Feige04}
Uriel Feige.
\newblock {Approximating maximum clique by removing subgraphs}.
\newblock {\em SIAM Journal on Discrete Mathematics}, 18(2):219--225, 2004.

\bibitem[FK01]{FeigeK01}
Uriel Feige and Joe Kilian.
\newblock {Heuristics for semirandom graph problems}.
\newblock {\em Journal of Computer and System Sciences}, 63(4):639--671, 2001.

\bibitem[FKP19]{FKP19}
Noah Fleming, Pravesh Kothari, and Toniann Pitassi.
\newblock {Semialgebraic Proofs and Efficient Algorithm Design}.
\newblock {\em Foundations and Trends{\textregistered} in Theoretical Computer
  Science}, 14(1-2):1--221, 2019.

\bibitem[FR87]{FranklR87}
Peter Frankl and Vojt{\v{e}}ch R{\"o}dl.
\newblock {Forbidden Intersections}.
\newblock {\em Transactions of the American Mathematical Society},
  300(1):259--286, 1987.

\bibitem[GMPT10]{GeorgiouMPT10}
Konstantinos Georgiou, Avner Magen, Toniann Pitassi, and Iannis Tourlakis.
\newblock {Integrality Gaps of $2-o(1)$ for Vertex Cover SDPs in the
  Lov{\'a}sz--Schrijver Hierarchy}.
\newblock {\em SIAM Journal on Computing}, 39(8):3553--3570, 2010.

\bibitem[GS17]{GS17}
David Gamarnik and Madhu Sudan.
\newblock Limits of local algorithms over sparse random graphs.
\newblock {\em Annals of probability: An official journal of the Institute of
  Mathematical Statistics}, 45(4):2353--2376, 2017.

\bibitem[GS20]{GS20}
Venkatesan Guruswami and Sai Sandeep.
\newblock d-to-1 hardness of coloring 3-colorable graphs with {O(1)} colors.
\newblock In Artur Czumaj, Anuj Dawar, and Emanuela Merelli, editors, {\em 47th
  International Colloquium on Automata, Languages, and Programming, {ICALP}
  2020, July 8-11, 2020, Saarbr{\"{u}}cken, Germany (Virtual Conference)},
  volume 168 of {\em LIPIcs}, pages 62:1--62:12. Schloss Dagstuhl -
  Leibniz-Zentrum f{\"{u}}r Informatik, 2020.

\bibitem[Har66]{Harper66}
Lawrence~H Harper.
\newblock {Optimal numberings and isoperimetric problems on graphs}.
\newblock {\em Journal of Combinatorial Theory}, 1(3):385--393, 1966.

\bibitem[Hof70]{Hoffman70}
Alan~J Hoffman.
\newblock On eigenvalues and colorings of graphs.
\newblock In {\em Graph Theory and its Applications}, pages 79--91. Acad.
  Press, 1970.

\bibitem[Jer92]{Jerrum92}
Mark Jerrum.
\newblock {Large cliques elude the Metropolis process}.
\newblock {\em Random Structures \& Algorithms}, 3(4):347--359, 1992.

\bibitem[Kar72]{Karp72}
Richard~M Karp.
\newblock {\em Reducibility among combinatorial problems}.
\newblock Springer, 1972.

\bibitem[Kho02]{Khot02}
Subhash Khot.
\newblock {On the power of unique 2-prover 1-round games}.
\newblock In {\em Proceedings of the thiry-fourth annual ACM symposium on
  Theory of computing}, pages 767--775, 2002.

\bibitem[KLT18]{KumarLT18}
Akash Kumar, Anand Louis, and Madhur Tulsiani.
\newblock {Finding Pseudorandom Colorings of Pseudorandom Graphs}.
\newblock In {\em 37th IARCS Annual Conference on Foundations of Software
  Technology and Theoretical Computer Science (FSTTCS 2017)}. Schloss
  Dagstuhl-Leibniz-Zentrum fuer Informatik, 2018.

\bibitem[KMS98]{KMS98}
David Karger, Rajeev Motwani, and Madhu Sudan.
\newblock {Approximate Graph Coloring by Semideﬁnite Programming}.
\newblock {\em Journal of the ACM (JACM)}, 45(2):246--265, 1998.

\bibitem[KOTZ14]{KauersOTZ14}
Manuel Kauers, Ryan O'Donnell, Li-Yang Tan, and Yuan Zhou.
\newblock {Hypercontractive inequalities via SOS, and the Frankl--R{\"o}dl
  graph}.
\newblock In {\em Proceedings of the twenty-fifth annual ACM-SIAM symposium on
  Discrete algorithms}, pages 1644--1658. SIAM, 2014.

\bibitem[KR08]{KhotR08}
Subhash Khot and Oded Regev.
\newblock Vertex cover might be hard to approximate to within 2-epsilon.
\newblock {\em J. Comput. Syst. Sci.}, 74(3):335--349, 2008.

\bibitem[KS12]{KhotS12}
Subhash Khot and Rishi Saket.
\newblock {Hardness of Finding Independent Sets in Almost q-Colorable Graphs}.
\newblock In {\em 2012 IEEE 53rd Annual Symposium on Foundations of Computer
  Science}, pages 380--389. IEEE, 2012.

\bibitem[KT17]{KawarabayashiT17}
Ken-ichi Kawarabayashi and Mikkel Thorup.
\newblock {Coloring 3-Colorable Graphs with Less than $n^{1/5}$ Colors}.
\newblock {\em Journal of the ACM (JACM)}, 64(1):1--23, 2017.

\bibitem[KTY24]{KTY24}
Ken-ichi Kawarabayashi, Mikkel Thorup, and Hirotaka Yoneda.
\newblock Better coloring of 3-colorable graphs.
\newblock In {\em Proceedings of the 56th Annual ACM Symposium on Theory of
  Computing}, STOC 2024, page 331–339, New York, NY, USA, 2024. Association
  for Computing Machinery.

\bibitem[Ku{\v{c}}95]{Kucera95}
Lud{\v{e}}k Ku{\v{c}}era.
\newblock {Expected complexity of graph partitioning problems}.
\newblock {\em Discrete Applied Mathematics}, 57(2-3):193--212, 1995.

\bibitem[Las01]{Las01}
Jean~B. Lasserre.
\newblock {Global Optimization with Polynomials and the Problem of Moments}.
\newblock {\em {SIAM} Journal on Optimization}, 11(3):796--817, 2001.

\bibitem[Lor53]{Lor53}
George~G Lorentz.
\newblock {\em {Bernstein Polynomials}}.
\newblock University of Toronto Press, 1953.

\bibitem[Mar74]{Margulis74}
Grigorii~Aleksandrovich Margulis.
\newblock {Probabilistic characteristics of graphs with large connectivity}.
\newblock {\em Problemy peredachi informatsii}, 10(2):101--108, 1974.

\bibitem[Mat99]{Mathe99}
Peter Math{\'e}.
\newblock {Approximation of H{\"o}lder continuous functions by Bernstein
  polynomials}.
\newblock {\em The American mathematical monthly}, 106(6):568--574, 1999.

\bibitem[MM11]{MakarychevM11}
Konstantin Makarychev and Yury Makarychev.
\newblock {How to play unique games on expanders}.
\newblock In {\em Approximation and Online Algorithms: 8th International
  Workshop, WAOA 2010, Liverpool, UK, September 9-10, 2010. Revised Papers 8},
  pages 190--200. Springer, 2011.

\bibitem[MMT20]{MckenzieMT20}
Theo McKenzie, Hermish Mehta, and Luca Trevisan.
\newblock {A New Algorithm for the Robust Semi-random Independent Set Problem}.
\newblock In {\em Proceedings of the Fourteenth Annual ACM-SIAM Symposium on
  Discrete Algorithms}, pages 738--746. SIAM, 2020.

\bibitem[MMZ05]{MMZ05}
Marc M{\'e}zard, Thierry Mora, and Riccardo Zecchina.
\newblock Clustering of solutions in the random satisfiability problem.
\newblock {\em Physical Review Letters}, 94(19):197205, 2005.

\bibitem[MR17]{MR17}
Pasin Manurangsi and Prasad Raghavendra.
\newblock {A Birthday Repetition Theorem and Complexity of Approximating Dense
  CSPs}.
\newblock In {\em 44th International Colloquium on Automata, Languages, and
  Programming (ICALP 2017)}. Schloss Dagstuhl-Leibniz-Zentrum fuer Informatik,
  2017.

\bibitem[O'D14]{O14}
Ryan O'Donnell.
\newblock {\em {Analysis of Boolean Functions}}.
\newblock Cambridge University Press, 2014.

\bibitem[Par00]{Par00}
Pablo~A. Parrilo.
\newblock {\em Structured semidefinite programs and semialgebraic geometry
  methods in robustness and optimization}.
\newblock PhD thesis, California Institute of Technology, 2000.

\bibitem[PD01]{PD01}
Alexander Prestel and Charles Delzell.
\newblock {\em {Positive Polynomials: From Hilbert’s 17th Problem to Real
  Algebra}}.
\newblock Springer Science \& Business Media, 2001.

\bibitem[RT12]{RT12}
Prasad Raghavendra and Ning Tan.
\newblock Approximating csps with global cardinality constraints using sdp
  hierarchies.
\newblock In {\em Proceedings of the twenty-third annual ACM-SIAM symposium on
  Discrete Algorithms}, pages 373--387. SIAM, 2012.

\bibitem[RV17]{RabaniV17}
Yuval Rabani and Rakesh Venkat.
\newblock {Approximating Sparsest Cut in Low Rank Graphs via Embeddings from
  Approximately Low Dimensional Spaces}.
\newblock {\em Approximation, Randomization, and Combinatorial Optimization.
  Algorithms and Techniques}, 2017.

\bibitem[Sch04]{Schweighofer04}
Markus Schweighofer.
\newblock {On the complexity of Schm{\"u}dgen's Positivstellensatz}.
\newblock {\em Journal of Complexity}, 20(4):529--543, 2004.

\bibitem[Tal93]{Talagrand93}
Michel Talagrand.
\newblock {Isoperimetry, logarithmic Sobolev inequalities on the discrete cube,
  and Margulis' graph connectivity theorem}.
\newblock {\em Geometric \& Functional Analysis GAFA}, 3(3):295--314, 1993.

\bibitem[Tre08]{Trevisan08}
Luca Trevisan.
\newblock {Approximation Algorithms for Unique Games}.
\newblock {\em Theory OF Computing}, 4:111--128, 2008.

\bibitem[Wig83]{Wigderson83}
Avi Wigderson.
\newblock {Improving the performance guarantee for approximate graph coloring}.
\newblock {\em Journal of the ACM (JACM)}, 30(4):729--735, 1983.

\bibitem[Zhu24]{Zhu}
Honglin Zhu.
\newblock Sphere packing proper colorings of an expander graph.
\newblock {\em CoRR}, abs/2405.20368, 2024.

\end{thebibliography}
